\documentclass[10pt]{article}

\usepackage[colorlinks, linkcolor=blue,citecolor=blue]{hyperref}           
\usepackage{cmap}
\usepackage{color}
\usepackage[percent]{overpic}
\usepackage{graphicx,subfigure,amsmath,amssymb,amsfonts,bm,epsfig,epsf,url,dsfont,tcolorbox,bbm,microtype}
\usepackage{amsthm,mathrsfs}
\usepackage{epigraph}
\usepackage{tikz}
\usepackage{multirow}
\usepackage{bbm}      
\usepackage{booktabs}
\usepackage{cases}
\usepackage{enumitem}
\usepackage[small,bf]{caption}
\usepackage[top=1in,bottom=1in,left=1in,right=1in]{geometry}
\usepackage{fancybox}
\usepackage{hyperref}
\hypersetup{
    colorlinks,
    linkcolor={blue!80!black},
    citecolor={blue!80!black},
    urlcolor={blue!80!black}
}
\usepackage{algorithm}
\usepackage{verbatim}
\usepackage{algorithmic}
\usepackage{mathtools}
\usepackage{braket}

\newtheorem*{theorem*}{Theorem}
\usepackage[normalem]{ulem}
\newcommand\redout{\bgroup\markoverwith{\textcolor{red}{\rule[0.5ex]{2pt}{0.8pt}}}\ULon}

\usepackage{scalerel,stackengine}

\stackMath
\newcommand\reallywidehat[1]{
\savestack{\tmpbox}{\stretchto{%
0  \scaleto{%
    \scalerel*[\widthof{\ensuremath{#1}}]{\kern-.6pt\bigwedge\kern-.6pt}%
    {\rule[-\textheight/2]{1ex}{\textheight}}
  }{\textheight}%
}{0.5ex}}%
\stackon[1pt]{#1}{\tmpbox}%
}

\newcommand{\qc}[1]{\left[\!\left[#1\right]\!\right]}

\newtheorem{theorem}{Theorem}
\newtheorem{cor}[theorem]{Corollary}

\newtheorem{lemma}[theorem]{Lemma}
\newtheorem{prop}[theorem]{Proposition}

\theoremstyle{remark}
\newtheorem{definition}{Definition}
\newtheorem{remark}{Remark}

\newenvironment{fminipage}%
  {\begin{Sbox}\begin{minipage}}%
  {\end{minipage}\end{Sbox}\fbox{\TheSbox}}

\makeatletter
\newcommand*{\rom}[1]{\expandafter\@slowromancap\romannumeral #1@}
\makeatother

\DeclareMathOperator{\supp}{supp}
\DeclareMathOperator{\Span}{Span}

\newcommand{\SIN}{\mathrm{in}}
\newcommand{\SOUT}{\mathrm{out}}

\usepackage{mathrsfs}

\newcommand{\nc}{\newcommand}
\nc{\us}{{\underline{s}}}
\nc{\un}{{\underline{n}}}  \nc{\ux}{{\underline{x}}} 
\nc{\uX}{{\underline{X}}}  \nc{\uY}{{\underline{Y}}}  \nc{\uv}{{\underline{v}}}     \nc{\uz}{{\underline{z}}} 
\nc{\uzero}{{\underline{0}}} 
\nc{\uTheta}{{\underline{\theta}}} 
\nc{\uy}{{\underline{y}} } \nc{\ue}{{\underline{e}}}  \nc{\uf}{{\underline{f}} } \nc{\ur}{{\underline{r}}} 
\nc{\ub}{{\underline{b}} } \nc{\ua}{{\underline{a}}} 
 \nc{\ui}{{\underline{i}}}  
 \nc{\wt}{{\mathrm{wt}}}  
  \nc{\lcs}{{\mathrm{LCS}}}  
  \nc{\uj}{{\underline{j}}}  
\nc{\Corr}{\mathrm{corr}}
\nc{\1}{{\mathbbm{1}}}
\nc\bfx{\boldsymbol x}
\nc{\Cov}{\mathrm{Cov}}

\newcommand{\ima}{\mathrm{Im}\:}

\newcommand{\tr}{\mathrm{tr}}

\newcommand{\abs}[1]{\left|#1\right|}
\newcommand{\R}{\mathbb{R}} 
\newcommand{\N}{\mathbb{N}}

\newcommand{\C}{\mathbb{C}}

 \nc\sC{{\mathscr C}}

\newcommand{\calB}{{\cal B}}

\newcommand{\calE}{{\cal E}}
\newcommand{\calF}{{\cal F}}

\newcommand{\calH}{{\cal H}}
\newcommand{\calI}{{\cal I}}

\newcommand{\calK}{{\cal K}}

\newcommand{\calM}{{\cal M}}
\newcommand{\calN}{{\cal N}}

\newcommand{\calP}{{\cal P}}

\newcommand{\calR}{{\cal R}}

\newcommand{\calV}{{\cal V}}

\DeclarePairedDelimiterX{\cond}[1]{[}{]}{\setargs{#1}}
\NewDocumentCommand{\setargs}{>{\SplitArgument{1}{;}}m}
{\setargsaux#1}
\NewDocumentCommand{\setargsaux}{mm}
{\IfNoValueTF{#2}{#1} {#1\,\delimsize|\,\mathopen{}#2}}

\newcommand{\be}{\begin{equation}}
\newcommand{\ee}{\end{equation}}
\newcommand{\beqna}{\begin{eqnarray}}
\newcommand{\eeqna}{\end{eqnarray}}

\newcommand{\p}[1]{\left(#1\right)}

\newcommand{\ppp}[1]{\left\{#1\right\}}
\newcommand{\norm}[1]{\left\|#1\right\|}
\newcommand{\innerP}[1]{\left\langle#1\right\rangle}

\usepackage{xspace}
\usepackage{array}
\usepackage{tabularx}
\usepackage{booktabs} 
\usepackage{adjustbox} 

\newcommand{\s}[1]{\mathsf{#1}}

\makeatletter
\def\thanks#1{\protected@xdef\@thanks{\@thanks
        \protect\footnotetext{#1}}}
\makeatother

\makeatletter
\renewcommand{\paragraph}{%
  \@startsection{paragraph}{4}%
  {\normalfont\normalsize\bfseries}%
}
\makeatother

\makeatletter
\patchcmd{\thebibliography}
  {\list}
  {\small
   \list}
  {}{}

\patchcmd{\thebibliography}
  {\sloppy}
  {\sloppy
   \setlength{\itemsep}{1pt}
   \setlength{\parsep}{3pt}
   \setlength{\parskip}{3pt}}
  {}{}
\makeatother

\usepackage[refpage]{nomencl}
\makenomenclature

\allowdisplaybreaks
\date{}

\begin{document}
\title{Universal recovery in approximate quantum error correction   }
 \author{Dor~Elimelech$^{1,2}$\hspace*{.4in} Victor~V.~Albert$^2$\hspace*{.4in}  Alexander~Barg$^{1,2,3}$}\thanks{$^1$Institute for Systems Research, University of Maryland, College Park, MD 20742. $^2$Joint Center for Quantum Information and Computer Science, NIST/University of Maryland, College Park, MD 20742. $^3$Department of ECE, University of Maryland, College Park, MD 20742. Emails: 
  \{dor,vva,abarg\}@umd.edu.}

\maketitle
 \begin{abstract}
 Universal recovery---the existence of a single recovery map that corrects an entire family of error channels---is a central feature of quantum error correction (QEC). In exact QEC, linearity guarantees that a code correcting a given error set also corrects every channel whose Kraus operators lie in its linear span, and that a single recovery map suffices for all such channels. Approximate quantum error correction (AQEC), which relaxes perfect recovery to recovery with controlled error, has traditionally lacked this structure. In a recent paper (arXiv:2607.22995), we developed a theory of approximate quantum error correction showing that a restricted form of linearity persists in the approximate setting, yielding uniform AQEC guarantees for the family of channels controlled by a given error set.

In this work, we complete the picture by establishing the second half of universal recovery in the approximate setting: a single recovery map can simultaneously correct every channel controlled by a given error set. The error-set theory we proposed quantifies approximate correctability through two parameters: the environment-leakage distance, governing worst-case performance, and the Knill–Laflamme Hellinger distance, governing average-case performance. We show here that both quantities also control universal decoding. We further study the Petz map naturally associated with an error set as an explicit universal recovery, and obtain uniform average- and worst-case guarantees across the entire family of channels. 
 \end{abstract}

\section{Introduction}

Quantum error correction (QEC) protects encoded quantum information against noise by arranging that the relevant error operators act reversibly on a code subspace \cite{shor1995scheme,knill1997theory}. It is a basic component of fault-tolerant quantum computation, quantum communication, and the study of robust quantum phases of matter \cite{shor1996fault,aharonov2008fault,bennett1996mixed,kitaev2003fault,dennis2002topological}. A central strength of exact QEC is the linear structure provided by the Knill--Laflamme (KL) conditions \cite{knill1997theory}: once a code corrects a prescribed error set, it also corrects every channel whose Kraus operators lie in its linear span \cite{knill2000theory}. This gives exact QEC an adversarial error-set formulation, in which one can protect uniformly against an entire family of channels by constraining only the type and severity of the errors they may introduce, rather than specifying the noise channel itself \cite{gottesman2010introduction}.

Approximate quantum error correction relaxes the requirement of perfect recovery and asks instead that, after noise and decoding, the recovered state remain close to the original encoded state according to a suitable fidelity or distance measure. This relaxation is important both because exact correction is typically too rigid for physically relevant noise and because approximate codes can achieve coding regimes unavailable to exact QEC, as first demonstrated for amplitude-damping noise and more recently for codes approaching the quantum Singleton and Hamming bounds \cite{leung1997approximate,bergamaschi2024approaching,ma2025haar}. This has led to a broad theory of AQEC, including quantitative conditions for approximate recoverability, coherent-information criteria, and near-optimal recovery maps \cite{schumacher1996sending,schumacher2001approximate,barnum2002reversing,beny2010general,klesse2007approximate}; refined analyses of channel performance and the Petz recovery map \cite{ng2010simple,noh2018quantum,zheng2024near,li2025optimality,kim2026optimal}; explicit constructions and asymptotic existence results \cite{bergamaschi2024approaching,ma2025haar,xu2025letting}; and connections to many-body order, circuit complexity, and information masking \cite{yi2024complexity,yi2025lov,li2025random}. Despite this progress, the theory has remained overwhelmingly channel-based, reflecting the longstanding view that the linearity underlying the adversarial error-set formulation of exact QEC does not extend to the approximate setting \cite{crepeau2005approximate}.

This view was recently reexamined in \cite{Elimelech2026Theory}, where it was shown that AQEC admits a meaningful adversarial error-set formulation under both the worst-case and average-case fidelity criteria. Although the full linearity of exact QEC does not survive, a restricted form remains: for a prescribed error set $\calE$, one can obtain uniform AQEC guarantees for a family of channels whose Kraus operators lie in $\Span(\calE)$ and whose mixing coefficients satisfy a spectral constraint. These channels, called $\calE$-controlled channels, give rise to an approximate error-set model in which structural features of exact QEC extend to the approximate setting.
The two notable features are the relation between correcting erasures and general errors and a corresponding notion of approximate distance. The same framework also supports explicit families of codes that correct structured error sets uniformly rather than one noise channel at a time.

Linearity in exact QEC has a consequence that goes beyond uniform correctability: a {\em single recovery map} can perfectly correct every channel whose Kraus operators lie in the span of a correctable error set. This channel-independent decoding is a basic part of the adversarial error-set formulation, since the recovery map depends only on the errors that may occur and not on the particular noise channel. The error-set theory of \cite{Elimelech2026Theory}, based on the B\'eny--Oreshkov framework \cite{beny2010general} in the worst-case setting and on Petz recovery in the average-case setting \cite{barnum2002reversing, zheng2024near}, established uniform AQEC guarantees across all $\calE$-controlled channels. These arguments, however, provide a suitable recovery map for each controlled channel separately and do not by themselves show that one recovery map can be used for the entire family. Whether the uniform protection against $\calE$-controlled channels also extends to a universal decoding operation is therefore an important remaining question for adversarial AQEC. Universal approximate decoders were previously obtained in several structured settings \cite{bergamaschi2024approaching,ma2025haar}, providing evidence that such a theory should be possible, but no general results were known for the error-set model.

The goal of this work is to complete this picture by showing that the error-set AQEC model of \cite{Elimelech2026Theory} supports universal adversarial decoding. The AQEC guarantees developed there are controlled by two error-set/code parameters: the {\em environment-leakage distance} in the worst-case setting and the {\em Knill--Laflamme Hellinger distance } in the average-case setting. We prove that the corresponding sufficient conditions also guarantee the existence of a single recovery map that works uniformly for every $\calE$-controlled channel. In particular, the asymptotically good AQEC families constructed in \cite{Elimelech2026Theory} admit universal recovery maps. In the average-case setting, we prove a stronger statement: for every fixed code and error set, the optimal error obtained by allowing the recovery map to depend on the noise channel is attained by a single optimal universal recovery map.

Establishing the existence of universal recovery maps is only one part of the problem; it is equally important to understand how such maps can be constructed in a general error-set setting. A natural candidate is the Petz recovery map, also known in QEC as the transpose channel, which has played a central role in AQEC \cite{barnum2002reversing,ng2010simple,noh2018quantum,zheng2024near,li2025optimality,kim2026optimal}. Given a noise channel and a code, the corresponding Petz map is known to be near optimal under the average-case criterion \cite{barnum2002reversing,zheng2024near}. At the same time, the Petz construction is not restricted to trace-preserving maps \cite{barnum2002reversing,ng2010simple,afham2026projections}, and in the exact setting it may be associated directly with an error set rather than with a particular channel. Motivated by these observations, we study the Petz map constructed from a given error set $\calE$ and analyze its performance uniformly over the corresponding controlled family. We show that its average-case error is controlled by the Knill--Laflamme Hellinger distance, while its worst-case error is controlled by the environment-leakage distance under the stronger notion of $\calE$-control.

The remainder of the paper is organized as follows. In Section~\ref{sec:AQEC} we review the error-set formulation of AQEC used throughout the paper. 
Section~\ref{sec:universalRecovery} establishes the existence of universal error-set recovery maps, and Section~\ref{sec:errorSetPetz} develops the error-set Petz recovery map and its average- and worst-case performance guarantees. Section~\ref{sec:numericalExperiments} presents numerical examples illustrating these results in several regimes. In Section~\ref{sec:binomialExperiment} we study binomial codes under amplitude-damping noise, comparing channel-adapted Petz maps with a single error-set Petz map as the loss probability varies. Section~\ref{sec:NEPauli} considers Pauli error sets: square Shor codes are used to compare the average-case Petz guarantee with the optimal universal decoding error, while Haar-random codes under limited-weight Pauli errors are used to examine the tightness of the worst-case Petz guarantee and its dependence on the code dimension. 

\section{Approximate quantum error-correction}\label{sec:AQEC}
 We consider quantum codes, denoted by $Q$, on a finite-dimensional Hilbert space $\calH$. For an integer $M>0$, we denote by $\calH_{M}$ the $M$-dimensional Hilbert space $\C^M$. We denote the set of linear operators and density operators on $\calH$ by $L(\calH
)$ and $D(\calH)$ respectively. We commonly measure the closeness of quantum states $\rho,\tau\in D(\calH)$ by the fidelity defined by $\calF(\rho,\tau):=\norm{\sqrt{\rho}\sqrt{\tau}}_1$, where $\norm{\cdot}_1$ is the trace norm (see Definition~\ref{Def:StateNorms}). A quantum channel is a completely positive and trace-preserving map (CPTP) $L(\calH)\to L(\calH')$. Any completely positive (CP)  map is known to be characterized by a Kraus operator set, $\ppp{E_k}_{k\in [S]}$, where $[S]=\ppp{0,1,\dots, S-1}$ and $E_i:\calH\to \calH'$ such that $\calN(X)=\sum_{k\in[S]}E_k X E_k^\dag$. If $\calN$ is also trace preserving, we have $\sum_{k}E_k^\dag E_k=I$. Another equivalent characterization of CP maps is due to Stinespring~\cite{Stinespring1955Positive}: a map is CP if and only if there exist an ancilla space $\calH_M$ and an operator $U:\calH\to \calH'\otimes \calH_M$ such that $\calN(X)=\tr_{M}(UXU^\dag)$, where $\tr_M$ denotes the partial trace operation on the ancilla space $\calH_M$. We elaborate more on the Stinespring representation in Appendix~\ref{sec:Stinespring}. A map $\calR:L(\calH')\to L(\calH)$ is called a perfect recovery map for a channel $\calN$ on the code space if $\calR\circ\calN|_{L(Q)}=I_{L(Q)}$, and an approximate recovery map if   $\calR\circ\calN|_{L(Q)}$ and $I_{L(Q)}$ are close under an appropriate distance measure between quantum channels. 

 We focus on two prominent fidelity-based distinguishability measures for quantum channels. The first is the worst-case entanglement fidelity, which is used in the B\'eny--Oreshkov framework \cite{beny2010general}. For quantum channels $\calN,\calM:L(\calH)\to L(\calH')$, it is defined as
\[
\calF_{\mathrm{wc}}\p{\calN,\calM}
:=
\inf_{\rho\in D(\calH)}
\calF\p{
\p{I_{L(\calH)}\otimes\calN}\p{\ket{\psi_{\rho}}\bra{\psi_{\rho}}},
\p{I_{L(\calH)}\otimes\calM}\p{\ket{\psi_{\rho}}\bra{\psi_{\rho}}}
},
\]
where $\ket{\psi_{\rho}}$ is any purification of $\rho$. Thus, $\calF_{\mathrm{wc}}$ measures the minimum fidelity of the channels' outputs uniformly over all input states, including their correlations with a reference system.

Alongside the worst-case entanglement fidelity, a widely used weaker criterion is the channel (or process) fidelity \cite{audenaert2002optimizing,kosut2009quantum,noh2018quantum,zheng2024near}, frequently used in the analysis of Petz recovery. Let $\ket{\Phi_{\calH}}$ be a maximally entangled state on $\calH\otimes\calH$. The channel fidelity of $\calN$ and $\calM$ is
\[ \calF_{\mathrm{ch}}\p{\calN,\calM} :=\calF\p{\p{I_{L(\calH)}\otimes\calN}\p{\ket{\Phi_{\calH}}\bra{\Phi_{\calH}}}, \p{I_{L(\calH)}\otimes\calM}\p{\ket{\Phi_{\calH}}\bra{\Phi_{\calH}}}}^2. \]
Unlike $\calF_{\mathrm{wc}}$, the channel fidelity evaluates the entanglement fidelity only for the maximally mixed input rather than minimizing over all input states. It therefore gives a weaker criterion, although it remains a standard and meaningful measure of channel performance. Moreover, the channel fidelity is closely related to the Haar-average input-output fidelity \cite[Sec.~6.4]{khatri2020principles}, and therefore it can be thought of as an average-case criterion.

For both criteria, proximity of the fidelity to one captures closeness of the corresponding channels. 
In accordance with \cite{beny2010general} and \cite{Elimelech2026Theory}, we formulate this closeness in terms of the associated Bures distances, for which small distance corresponds to high fidelity. The corresponding worst-case Bures distance and channel-fidelity Bures distance are defined as
\[ d_{\mathrm{wc}}\p{\calN,\calM} :=\sqrt{1-\calF_{\mathrm{wc}}\p{\calN,\calM}}, \quad d_{\mathrm{ch}}\p{\calN,\calM}:=\sqrt{1-\calF_{\mathrm{ch}}\p{\calN,\calM}}. \] 

We next give the standard channel-based definition of approximate quantum error correction for a code. \begin{definition}\label{def:AQECchannel}
    A quantum code $Q\subseteq\calH$ is said to be an $\varepsilon$-$\s{wc}$-AQEC code for a channel $\calN:L(\calH)\to L(\calH')$ if there exists a CPTP recovery operation $\calR:L(\calH')\to L(Q)$ such that
\begin{equation}
d_{\mathrm{wc}}\p{\calR\circ\calN|_{L(Q)},I_{L(Q)}}\leq\varepsilon.
\label{eq:channel_AQEC}
\end{equation}
The code is said to be an $\varepsilon$-$\s{av}$-AQEC code for $\calN$ if \eqref{eq:channel_AQEC} holds for some $\calR$ with $d_{\mathrm{ch}}$ in place of $d_{\mathrm{wc}}$. 
\end{definition}

\subsection{The error-set model for AQEC} 
In \cite{Elimelech2026Theory}, the adversarial error-set model of exact QEC was generalized to the approximate setting. Although the full linearity property of exact QEC does not extend to AQEC, a restricted form of linearity was shown to survive. Namely, given an error set of interest $\calE$, one can obtain uniform guarantees on the AQEC performance of a code for a family of channels admitting a Kraus representation whose operators lie in $\Span(\calE)$ and whose expansion coefficient matrix is a contraction:
\begin{definition}[$\calE$-controlled channels and error-set AQEC]\label{def:errorSetAQEC} Let $\calE=\ppp{E_k}_{k\in[S]}$ be a finite set of operators. A quantum channel $\calN$ is said to be $\calE$-controlled if it admits a Kraus representation $\ppp{A_m}_{m}$ of the form 
$ A_m=\sum_{k\in[S]}c_{m,k}E_k,$ 
where the coefficient matrix $C=(c_{m,k})$ satisfies 
\[ \norm{C}_{\infty}=\sup_{\norm{x}_2=1}\norm{Cx}_{2}\leq 1. \] 
We denote the set of $\calE$-controlled channels by $\mathscr{N}(\calE)$.

A quantum code $Q$ is said to be an  $\varepsilon$-($\s{wc}$/$\s{av}$)-AQEC code for the error set $\calE$ if it is an $\varepsilon$-($\s{wc}$/$\s{av}$)-AQEC code, in the sense of Definition~\ref{def:AQECchannel}, for every $\calE$-controlled channel. We denote the optimal AQEC error for the average-case and worst-case criteria by $\varepsilon_{\s{opt}}^{\s{av}}(\calE,Q)$ and $\varepsilon_{\s{opt}}^{\s{wc}}(\calE,Q)$ respectively. That is, $\varepsilon_{\s{opt}}^{\s{av}}(\calE,Q)$ and $\varepsilon_{\s{opt}}^{\s{wc}}(\calE,Q)$ are the infimum values over all $\varepsilon>0$ such that $Q$ is $\varepsilon$-$\s{av}/\s{wc}$-AQEC for $\calE$.
\end{definition} 

The family of $\calE$-controlled channels  has several useful structural properties. In particular, it is convex, and for a linearly independent set $\calE$, the property of being $\calE$-controlled is independent of the choice of the channel's Kraus representation. This model also gives rise to a broad range of nontrivial error-set controlled families. The primary example arises when the error set $\calE_t$ is the set of all Pauli operators of weight at most $t$ on an $N$-qudit system (that is, errors that act as tensor products of Pauli operators on at most $t$ qudits and identity on the remaining $N-t$ qudits). In this case, the set of $\calE_t$-controlled channels is the entire linear span of $\calE_t$, i.e., the set of all channels acting on at most $t$ qudits. We refer to \cite[Section~3.1.1]{Elimelech2026Theory} for further discussion. The error-set AQEC guarantees established in \cite{Elimelech2026Theory} show that this model provides meaningful uniform control over the corresponding family of channels. These guarantees are expressed through two error-set/code parameters: the \emph{environment-leakage distance}, which controls worst-case AQEC performance, and the \emph{Knill--Laflamme Hellinger distance}, which characterizes the average-case performance of the Petz recovery map.

We begin with the environment-leakage distance. Let $\calE=\ppp{E_k}_{k=0}^{S-1}$ be a finite set of operators $E_k:\calH\to\calH'$, and let $Q\subseteq\calH$ be a code with orthogonal projector $P$. For every $S\times S$ matrix $\lambda$, the corresponding B\'eny--Oreshkov superoperator $\calB^{\calE}_{\lambda,Q}:L(\calH)\to L(\C^S)$ is defined by\footnote{{This restatement corrects a misprint in \cite[Corollary~2]{beny2010general}.}}
\begin{equation}
\calB^{\calE}_{\lambda,Q}(X)= \sum_{k,l}\tr\p{B_{k,l}X}\ket{k}\bra{l}, \quad B_{k,l} = PE_l^\dag E_k P-\lambda_{k,l}P.
\label{eq:BO_superoperator} 
\end{equation}

This superoperator measures the deviation of the complementary map associated with $\calE$ and $Q$ from the constant map determined by $\lambda$.

\begin{definition}[Environment-leakage distance {\cite[Sec. 3.1]{Elimelech2026Theory}}] \label{def:environment_leakage_distance}
The \emph{environment-leakage distance} between an error set $\calE$ and a code $Q$ is
\begin{equation}
\zeta\p{\calE,Q} := \inf_{\lambda}\norm{\calB^{\calE}_{\lambda,Q}}_{\diamond},
\label{eq:environment_leakage_distance}
\end{equation}
where $\norm{\cdot}_{\diamond}$ denotes the diamond norm  on the space of superoperators, given in Definition~\ref{Def:superOpNorms}; see p.~\pageref{Def:superOpNorms}.
\end{definition}

Next we introduce the Knill--Laflamme Hellinger distance. Let $Q$ be a $K$-dimensional code with orthonormal basis $\ppp{\ket{c_i}}_{i\in[K]}$, and let $\calE=\ppp{E_k}_{k\in[S]}$. The corresponding QEC matrix $A_{\s{QEC}}\in L(\C^K\otimes\C^S)$ is given by
\begin{equation}
\p{A_{\s{QEC}}}_{(i,k),(j,l)} := \bra{c_i}E_k^\dag E_l\ket{c_j}. \label{eq:QEC_matrix}
\end{equation}
The exact Knill--Laflamme conditions hold precisely when this matrix is contained in the Knill--Laflamme space
\begin{equation}
\calH_{\s{KL}} := \ppp{I_K\otimes\lambda:\lambda\in L(\C^S)} \subseteq L(\C^K\otimes\C^S). \label{eq:KL_space}
\end{equation}
The relevant approximate version of this condition is obtained by measuring the distance of $A_{\s{QEC}}$ from $\calH_{\s{KL}}$: in the Hellinger-distance sense, this proximity characterizes the average-case performance of the Petz recovery map and therefore provides an approximate counterpart of the Knill--Laflamme conditions. For positive semidefinite operators $A$ and $B$, their quantum Hellinger distance is $D_{\s{H}}\p{A,B} :=
\|\sqrt{A}-\sqrt{B}\|_2$,
and the distance of $A\succeq0$ from the positive cone of the Knill--Laflamme space is
\begin{equation}
D_{\s{H}}\p{A,\calH_{\s{KL}}}
:=
\min_{\substack{B\in\calH_{\s{KL}}\\B\succeq0}}
D_{\s{H}}\p{A,B}.
\label{eq:Hellinger_distance_KL_space}
\end{equation}

\begin{definition}[Knill--Laflamme Hellinger distance {\cite[Def.~9]{elimelech2026asymptotically}}] \label{def:KL_Hellinger_distance}
The \emph{Knill--Laflamme Hellinger distance} (KL--Hellinger) between a $K$-dimensional code $Q$ and an error set $\calE$ is
\begin{equation}
\zeta_{\s{H}}\p{\calE,Q}
:=
\frac{1}{\sqrt{K}}
D_{\s{H}}\p{A_{\s{QEC}},\calH_{\s{KL}}}
=
\frac{1}{\sqrt{K}}
\min_{\lambda\succeq0}
\norm{\sqrt{A_{\s{QEC}}}-I_K\otimes\sqrt{\lambda}}_2.
\label{eq:KL_Hellinger_distance}
\end{equation}
\end{definition}

As shown in \cite[Section~3.2]{Elimelech2026Theory}, this quantity is independent of the choice of the matrix representation of $A_{\s{QEC}}$ and is expressed in closed form as follows:
\begin{equation}
    \zeta_{\s{H}}(\calE,Q)=\sqrt{\frac{1}{K}\tr(A_{\s{QEC}})-\frac{1}{K^2}\norm{\tr_K(\sqrt{A_{\s{QEC}}})}_2^2}. \label{eq:KLHdistanceExpression}
\end{equation} Moreover, if $\calE_{\calN}$ is a Kraus set of a channel $\calN$, then \cite[Corollary~10]{Elimelech2026Theory} shows that $\zeta_{\s{H}}\p{\calE_{\calN},Q}$ is the average-case Bures distance attained by the Petz recovery map, which is known to be optimal \cite{barnum2002reversing} up to a factor of $\sqrt{2}$. It was shown in \cite{Elimelech2026Theory} that $\zeta$ and $\zeta_{\s{H}}$ control the optimal AQEC error of a code
as stated next.
\begin{theorem}[Sufficient AQEC conditions \cite{Elimelech2026Theory}]\label{th:sufficientAQEC}
Let $Q$ be a quantum code and let $\calE$ be an error set. Then 
\[\varepsilon_{\s{opt}}^{\s{wc}}(\calE,Q)\leq \sqrt{2\zeta(\calE,Q)} \quad \text{and}\quad \varepsilon_{\s{opt}}^{\s{av}}(\calE,Q)\leq \zeta_{\s{H}}(\calE,Q). \]
\end{theorem}

\section{Universal error-set recovery maps}\label{sec:universalRecovery}
In this section, we introduce universal recovery maps for AQEC, define the corresponding optimal decoding errors, and prove their existence under both the worst-case and average-case criteria. Recall that the definition of error-set AQEC requires that, for every $\calE$-controlled channel, there exist a recovery map whose decoding error is at most $\varepsilon$. The recovery map may therefore depend on the particular channel. We are interested instead in a single recovery map that achieves the same error guarantee uniformly over all $\calE$-controlled channels. 

\begin{definition}[Universal recovery maps and optimal decoding errors]\label{def:universal_decoder} Let $Q\subseteq\calH$ be a quantum code, and let $\calE$ be a finite set of error operators from $\calH$ to $\calH'$. A CPTP recovery operation $\calR:L(\calH')\to L(Q)$ is called an $\s{av}$-$\varepsilon$ universal recovery map, respectively, a $\s{wc}$-$\varepsilon$ universal recovery map, for $\calE$ and $Q$ if, for every $\calE$-controlled channel $\calN$, \begin{equation} d\p{\calR\circ\calN|_{L(Q)},I_{L(Q)}}\leq\varepsilon, \label{eq:universal_decoder} \end{equation} where $d=d_{\mathrm{ch}}$ in the average-case criterion and $d=d_{\mathrm{wc}}$ in the worst-case criterion. 
The optimal average-case and worst-case universal decoding error parameters of $Q$ for $\calE$ are defined as       \begin{equation} 
\varepsilon_{\s{dec}}^{\s{*}}\p{\calE,Q} := \inf\ppp{\varepsilon>0: \text{there exists an $\s{*}$-$\varepsilon$ universal recovery map for $\calE$ and $Q$}},\quad \s{*}\in \ppp{\s{wc},\s{av}}.\label{eq:optimal_universal_decoding_errors} 
    \end{equation} 
\end{definition} 

\subsection{Existence of universal error-set recovery maps}
In what follows, we derive bounds on the optimal decoding error under both the worst-case and average-case criteria. We begin with the average-case setting, where we show that relying on a single universal recovery map incurs no loss: the optimal average-case decoding error coincides with the optimal average-case AQEC error.

We will use the following result.
\begin{theorem*}[Sion's minimax theorem {\cite{Sion1958}}]
    Let $X$ and $Y$ be convex subsets of linear topological spaces such that $Y$ is compact, and assume that $f:X\times Y\to \R$ is a function such that $f(\cdot,y)$ is continuous and quasi-concave on $X$ for all $y$ and $f(x,\cdot)$ is  continuous and quasi-convex on $Y$ for all $x$. Then 
     \[
     \sup_{x\in X}\min_{y\in Y} f(x,y)=\min_{y\in Y}  \sup_{x\in X}f(x,y),
     \]
    where the supremum above is a maximum if $X$ is compact as well.
\end{theorem*}

\begin{theorem}\label{th:AvcaseOptimalDecoding}
     Let $Q\subset \calH_{\SIN}$ be a code in a finite-dimensional Hilbert space, and let $\calE=\ppp{E_i}_{i\in [S]}$ be an error set with operators $E_i:\calH_{\SIN}\to\calH_{\SOUT}$. Then  \[\varepsilon_{\s{dec}}^{\s{av}}(\calE,Q)= \varepsilon_{\s{opt}}^{\s{av}}(\calE,Q).\] 
     In particular, there exists an $\s{av}$-$\zeta_{\s{H}}(\calE,Q)$ universal recovery map $\calR$. 
\end{theorem}

\begin{proof}
Let $X$ be the set of recovery maps, which is a subset of the linear space of superoperators $L(\calH_{\SOUT})\to L(Q)$ equipped with the norm topology (recall that all norm topologies on a finite-dimensional vector space are equivalent), let $Y$ be the set of $\calE$-controlled channels, and let $f$ be the function defined as
    \[f(\calR,\calN)=\calF_{\mathrm{ch}}(  \calR\circ\calN|_{L(Q)},I_{L(Q)})=\bra{\Phi_Q} (I_{L(Q)}\otimes \calR\circ \calN) (\phi_Q) \ket{\Phi_{Q}}, \]
    where $\ket{\Phi_Q}\in Q\otimes Q$ is the maximally entangled state and $\phi_Q=\ket{\Phi_Q}\bra{\Phi_Q}$. 

    \begin{lemma}\label{lem:CompactEcontrolled}
 The spaces $X$ and $Y$ defined in the proof of Theorem \ref{th:AvcaseOptimalDecoding} are compact and convex subsets of topological vector spaces of superoperators $L(\calH_{\SOUT})\to L(Q)$ and $L(\calH_{\SIN})\to L(\calH_{\SOUT})$, respectively, equipped with the diamond norm topology.
    \end{lemma}
The proof of this lemma is given in Appendix~\ref{app:LemmaCompactProof}.  

For any $\calR\in X$ and $\calN\in Y$ the functions  $f(\calR,\cdot)$ and $f(\cdot,\calN)$ are affine and therefore  continuous and both quasi-convex and quasi-concave. Together with $X$ and $Y$ defined above, they satisfy the
assumptions of Sion's minimax theorem, and we therefore have: 
    \begin{align*}
       \max_{\calR\in X}\min_{\calN\in Y}  \calF_{\mathrm{ch}}(  \calR\circ\calN|_{L(Q)},I_{L(Q)})=\min_{\calN\in Y}\max_{\calR\in X}  \calF_{\mathrm{ch}}(   \calR\circ\calN|_{L(Q)},I_{L(Q)}).
    \end{align*}
    Using the above and continuity of the square root function we conclude: 
    \begin{align*}
        \varepsilon_{\s{dec}}^{\s{av}}(\calE,Q)&=\inf_{\calR\in X}\sup_{\calN\in Y}d_{\mathrm{ch}}( \calR\circ\calN|_{L(Q)},I_{L(Q)})\\
        &=\sqrt{1-\max_{\calR\in X}\min_{\calN\in Y} \calF_{\mathrm{ch}}(  \calR\circ\calN|_{L(Q)},I_{L(Q)})}\\
        &=\sqrt{1-\min_{\calN\in Y}\max_{\calR\in X} \calF_{\mathrm{ch}}(  \calR\circ\calN|_{L(Q)},I_{L(Q)})}
        \\&=\sup_{\calN\in Y}\inf_{\calR\in X}d_{\mathrm{ch}}( \calR\circ\calN|_{L(Q)},I_{L(Q)})\\
        &=\varepsilon_{\s{opt}}^{\s{av}}(\calE,Q).
    \end{align*}
    The existence of an $\s{av}$-$\zeta_{\s{H}}(\calE,Q)$ universal recovery map now follows immediately from Theorem~\ref{th:sufficientAQEC}.
\end{proof}
     
We next turn to the worst-case criterion. Here, we show that the environment-leakage distance controls not only the optimal AQEC error, as established in Theorem~\ref{th:sufficientAQEC} {(\cite[Theorem~3]{Elimelech2026Theory})}, 
but also the optimal decoding error. Thus, a given error-set condition guarantees the existence of a {\bf single recovery map} that supports $\varepsilon$-AQEC over all $\calE$-controlled channels, where $\varepsilon$ depends only on $Q$ and $\calE$.
\begin{theorem}\label{th:DecExistsWorst}
    Let $Q\subset \calH_{\SIN}$ be a code in a finite-dimensional Hilbert space and let $\calE=\ppp{E_i}_{i\in [S]}$ be an error set with operators $E_i:\calH_{\SIN}\to\calH_{\SOUT}$. Then 
    \[\varepsilon_{\s{dec}}^{\s{wc}}(\calE,Q)\leq \sqrt{\zeta(\calE,Q)}.\] 
        That is, there exists a universal recovery map $\calR_{\calE}$ such that for any $\calE$-controlled channel we have 
    \[d_{\mathrm{wc}}\p{\calR_{\calE}\circ \calN|_{L(Q)},I_{L(Q)}}\leq \sqrt{\zeta(\calE,Q)}\]
\end{theorem}
 This theorem represents an improvement over Theorem~\ref{th:sufficientAQEC} by both tightening the upper bound for the optimal $\varepsilon$ and showing that it holds under the universal recovery map. The proof is rather technically involved, so here we only sketch the main ideas and present a complete argument in    Appendix~\ref{app:PoorfBOdecoder}. 
\begin{proof}[Proof overview] We generalize the approach of \cite{Kretschmann2008Information} from the channel to the error-set setting. The key tool underlying this approach is the Stinespring dilation continuity theorem, which asserts that two CP maps are  close (under the worst-case criterion) if and only if they admit $\infty$-norm close Stinespring dilation operators in a common ancilla space. We apply this idea in both directions: First we use the closeness of the complementary map of the error-set noise superoperator to a constant map to obtain close Stinespring dilations of the error-set noise map and a (scaled) identity channel. These dilations give rise to a recovery map $\calR$. In the other direction, we fix an $\calE$-controlled channel $\calN$, and show that the recovered noisy state $\calR\circ\calN$  admits a Stinespring dilation  which is close to a Stinespring operator of the identity channel. This closeness, by the continuity theorem, implies that $\calR\circ\calN$ is close to the identity channel as required.

With this general theme in mind, let us sketch the main steps of the proof. The definition of $\zeta(\calE,Q)$ implies that for a canonical choice of the matrix $\lambda$, the corresponding  B\'eny-Oreshkov superoperator $\calB_{\lambda,Q}^{\calE}$  satisfies $\|\calB_{\lambda,Q}^{\calE} \|_{\diamond}\leq 2\zeta(\calE,Q)$.  We consider the error-set noise map on the code space, $\calM_{\calE}^Q$, whose Kraus representation is $\ppp{E_iP}_i$. We make two important observations. First,
\[\calB^{\calE}_{\lambda,Q}(X)= \sum_{k,l}\tr\p{( PE_l^\dag E_kP-\lambda_{k,l}P)X}\ket{k}\bra{l}=\widehat{\calM}_{\calE}^Q(X)-\calM_\lambda(X),\]
where $\widehat{\calM}_{\calE}^Q$ is the complementary map of $\calM_{\calE}^Q$ with respect to the canonical Stinespring dilation operator $\calV_{\calE}$  of $\calM_{\calE}^Q$ (see Definition~\ref{def:Stinespring}), and $\calM_{\lambda}$ is an operator on $L(Q)$ that admits a Kraus dilation operator $V_{\lambda}$. The second observation is that the complementary channel of $\calM_{\lambda}$ via the Stinespring operator $V_{\lambda}$ turns out to be a scaled identity channel $\alpha_\lambda  I_{L(Q)}$. 

The continuity theorem combined with the fact that any two different Stinespring dilation operators differ by a partial isometry on the ancilla spaces implies that there exist partial isometries $J_{\calE}$ and $J_{\lambda}$ such that 
\begin{equation}
    \norm{(I\otimes J_{\calE} )V_{\calE}-(I\otimes J_\lambda )V_{\lambda}}_{\infty}\leq \sqrt{\| \widehat{\calM}_{\calE}^Q-\calM_{\lambda}\|_{\diamond}}\leq \sqrt{2\zeta(\calE,Q)},\label{eq:ContinuityMotivation}
\end{equation}
where on the previous line and in the remainder of this discussion, the identity operators in the tensor product maps 
act on the appropriate spaces (which may be different). We omit these spaces from the notation for better readability.

The bound \eqref{eq:ContinuityMotivation} motivates a recovery map: assume that one wants to (approximately) reverse the operation of the error-set noise map $\calM_{\calE}^Q$. A natural choice for the decoding operation would be:
\[\calR(X)=\tr_{\s{env}}\p{R X R^\dag}+(\text{arbitrary completion for TP}), \quad R=J_\lambda^\dag  J_{\calE}.\]
To see that, we note that the decoded noise channel on the code space $\calR\circ \calM_{\calE}^Q$  turns out to have Kraus representation of the form $V_{\s{dec}}=(I\otimes J_{\lambda}^\dag J_{\calE})V_{\calE}$ (up to the trace preservation term), which satisfies 
\begin{align}
    \norm{V_{\s{dec}}-V_{\lambda}}_{\infty}&=\nonumber\norm{(I\otimes J_{\lambda}^\dag J_{\calE})V_{\calE}-(I\otimes J_{\lambda}^\dag J_{\lambda})V_{\lambda}}_{\infty}\\
    &\nonumber=\norm{(I\otimes J_{\lambda}^\dag)\Big((I\otimes  J_{\calE})V_{\calE}-(I\otimes  J_{\lambda})V_{\lambda}\Big)}_{\infty}\\
    &\leq\norm{(I\otimes  J_{\calE})V_{\calE}-(I\otimes  J_{\lambda})V_{\lambda}}_{\infty}\leq \sqrt{2\zeta  (\calE,Q)},\label{eq:decodedInfBound}
\end{align}
where the last step uses the fact that $J_{\lambda}^\dag$ is a partial isometry and therefore it is contractive. 
Since $V_{\lambda}$ dilates the scaled identity, 
the continuity theorem implies that $\calR\circ \calM_{\calE}^Q$ is close to $\alpha_{\lambda}I_{L(Q)}$ as desired. 

In the final step, we transition from recovering the error-set noise map (which generally is not a physical quantum channel) to recovering an actual $\calE$-controlled physical channel $\calN$ given by a coefficient matrix $C$ (see Definition~\ref{def:errorSetAQEC}). A straightforward calculation shows that $\calN$ admits a Stinespring dilation operator $V_{C}$ related to the $V_{\calE}$ by $V_{C}=(C\otimes I)V_{\calE}$. Accordingly, (up to trace preserving terms) the recovered noise channel $\calR\circ\calN|_{L(Q)}$ on the code space has the Stinespring representation 
       \begin{align*}
       V_{\s{dec},C}&=(I\otimes R)V_C=(I\otimes R)(C\otimes I)V_{\calE}=(C\otimes I)(I\otimes R)V_{\calE}\\
       &= (C\otimes R)V_\calE.
       \end{align*}
The spectral condition $\norm{C}_{\infty}\leq 1$  together with \eqref{eq:decodedInfBound} gives that  
    \begin{align*}
    \norm{V_{\s{dec},C}-(C\otimes I)V_{\lambda}}_{\infty}&=\norm{(C\otimes I)\Big(  (I\otimes R)V_{\calE}-V_{\lambda}\Big)}_\infty\\
    &\leq\norm{  (I\otimes R)V_{\calE}-V_{\lambda}}_\infty\leq \sqrt{2\zeta(\calE,Q)}. 
    \end{align*}
 By channel normalization, $V_{\lambda_C}:=(C\otimes I)V_{\lambda}$ is a Stinespring dilation for the identity channel. Thus, by the Stinespring continuity theorem we have 
     \[ d_{\mathrm{wc}}(\calR\circ\calN|_{L(Q)},I_{L(Q)})\leq \frac{1}{\sqrt{2}}\norm{V_{\s{dec},C}-V_{\lambda_C}}_{\infty}\leq \sqrt{\zeta(\calE,Q)}. \qedhere
     \]
\end{proof}

Theorem~\ref{th:DecExistsWorst} strengthens the sufficient AQEC condition of \cite[Theorem~3]{Elimelech2026Theory} by a factor of $\sqrt{2}$:
\begin{cor}\label{cor:ZetaBetterBound} For an error set $\calE $ and a quantum code $Q$ we have:
    \[ \varepsilon_{\s{opt}}^{\s{wc}}\p{\calE,Q}  \leq \sqrt{\zeta\p{\calE,Q}}. \]
\end{cor} Indeed, since every universal recovery map is, in particular, a valid recovery map for each individual $\calE$-controlled channel, using Theorem~\ref{th:DecExistsWorst} we have $\varepsilon_{\s{opt}}^{\s{wc}}\p{\calE,Q} \leq \varepsilon_{\s{dec}}^{\s{wc}}\p{\calE,Q} \leq \sqrt{\zeta\p{\calE,Q}}$. Hence, the environment-leakage distance provides a uniform worst-case guarantee through a single recovery map, while also strengthening  
the sufficient condition by a $\sqrt{2}$ factor.

\section{Error-set Petz recovery map: a polar approach}\label{sec:errorSetPetz}
A common approach to AQEC is based on the Petz recovery map, also known in the QEC literature as the transpose channel \cite{barnum2002reversing,ng2010simple,noh2018quantum,zheng2024near,li2025optimality,kim2026optimal}. In the coding setting, the Petz recovery map $\calR_{\calN,Q}$ is constructed from a noise channel $\calN$ and a code $Q$. It gives an exact recovery operation whenever $Q$ satisfies the Knill--Laflamme conditions, and remains near optimal in the approximate setting. In the Bures-distance normalization used here, its decoding error is within a factor of $\sqrt{2}$ of the optimal error \cite{barnum2002reversing,ng2010simple}. In \cite{Elimelech2026Theory}, the optimal average-case AQEC error in the error-set model was bounded by analyzing the Petz recovery map $\calR_{\calN,Q}$ separately for each $\calE$-controlled channel $\calN$. It was shown that the resulting average-case Bures distance is bounded by the KL--Hellinger distance $\zeta_{\s{H}}\p{\calE,Q}$. This yields a bound on the optimal AQEC error $\varepsilon_{\s{opt}}^{\s{av}}\p{\calE,Q}$, but does not directly control the optimal decoding error $\varepsilon_{\s{dec}}^{\s{av}}\p{\calE,Q}$, since the recovery map $\calR_{\calN,Q}$ depends on the particular channel $\calN$. Bounding the decoding error instead requires a single recovery map whose performance is controlled simultaneously for all $\calE$-controlled channels.

In this section, we bridge this gap by replacing the channel-dependent Petz map with a Petz map associated directly with the error set $\calE$. The Petz construction does not require the underlying CP map to be trace preserving, and may be applied without assuming that its Kraus operators satisfy a completeness relation \cite{barnum2002reversing,ng2010simple,afham2026projections}. After the standard completion outside the relevant support, it defines a CPTP recovery map. We apply this construction to the error-set CP map
\begin{equation}
    \calM_{\calE}(X):=\sum_{E\in\calE}EXE^\dag. \label{eq:ErrorMapSet}
\end{equation}
The goal of this section is to show that the average-case and worst-case performance of the resulting recovery map can be controlled uniformly over all $\calE$-controlled channels.

We first recall the conventional definition of the Petz recovery map adapted to our error-set setting.

\begin{definition}[Error-set Petz recovery map]\label{def:errorSetPetz}
Let $Q\subseteq\calH_{\SIN}$ be a $K$-dimensional quantum code with orthogonal projector $P$, and let $\calE=\ppp{E_i}_{i\in[S]}$ be an error set with operators $E_i:\calH_{\SIN}\to\calH_{\SOUT}$. Consider the CP map $\calM_{\calE}(X)$ defined in \eqref{eq:ErrorMapSet}, let $\sigma:=\frac{1}{K}P$ be the reference state, and denote $\Gamma=\calM_{\calE}(P)$. The error-set Petz recovery map associated with $\calM_{\calE}$ and $Q$ is
\begin{equation}
    \calR_{\calE,Q}^{\s{Petz}}(X) := \sigma^{1/2}
\calM_{\calE}^{\dag}\p{ \calM_{\calE}(\sigma)^{-1/2} X \calM_{\calE}(\sigma)^{-1/2}}\sigma^{1/2},\label{eq:PetzFormula}
\end{equation}
where the inverse is the Moore--Penrose pseudoinverse taken on the support of $\calM_{\calE}(\sigma)$. Equivalently,
\[ \calR_{\calE,Q}^{\s{Petz}}(X) = P\calM_{\calE}^{\dag}\p{ \Gamma^{-1/2} X \Gamma^{-1/2} }P. \]
On the support of $\Gamma$, this map admits a Kraus representation of the form
\begin{equation}
    \calR_{\calE,Q}^{\s{Petz}}(X) = \sum_{i\in[S]}R_iXR_i^\dag, \qquad R_i:=PE_i^\dag\Gamma^{-1/2}.\label{eq:KrausSetPetz}
\end{equation}
The map is trace preserving on this support, and we can view $\calR_{\calE,Q}^{\s{Petz}}$ as a CPTP map on $L(\calH_{\SOUT})$ upon completion on its orthogonal complement:
\[\calR_{\calE,Q}^{\s{Petz}}(X)= \sum_{i\in[S]}R_iXR_i^\dag+ \tr\p{(I-P_{\Gamma})X}\rho_Q,\]
where $P_{\Gamma}$ is the projection on $\supp(\Gamma)$ and $\rho_Q=\ket{\psi_{Q}}\bra{\psi_Q}\in L(Q)$ is an arbitrary pure state. 
\end{definition}

\subsection{Polar representation of the Petz map}
We adopt the polar decomposition approach to the Petz map used in \cite{afham2026projections}, where the Petz map is characterized as the Choi--Bures projection of an unnormalized reverse map onto the set of quantum channels. We adapt this viewpoint to the AQEC setting, taking the case of exact QEC as the source of intuition and motivation for the construction.  Let $Q\subseteq\calH_{\SIN}$ be a $K$-dimensional code with orthogonal projector $P$, and suppose that $Q$ satisfies the Knill--Laflamme conditions for the error set $\calE=\ppp{E_i}_{i\in[S]}$. Equivalently, its QEC matrix $A_{\s{QEC}}$ defined in \eqref{eq:QEC_matrix} belongs to the Knill--Laflamme space \eqref{eq:KL_space} and has the form $A_{\s{QEC}}=I_K\otimes\lambda$ for some positive semidefinite $S\times S$ matrix $\lambda$. Consider the error-synthesis operator $S_{\calE,Q}:Q\otimes\C^S\longrightarrow\calH_{\SOUT}$ given as $S_{\calE,Q}\p{ \ket{\psi}\otimes\ket{i}} :=E_iP\ket{\psi}$, or equivalently
\begin{equation}
     S_{\calE,Q}=\sum_{i\in [S]} E_i P \otimes \bra{i}. \label{eq:ErrorSynthDef}
\end{equation}
where we identify $Q$ with its image in $\calH_{\SIN}$ under the inclusion map $P$. We observe that
\begin{equation}
S_{\calE,Q}^{\dag}S_{\calE,Q} =\sum_{i,j\in[S]} PE_i^\dag E_jP\otimes\ket{i}\bra{j}=A_{\s{QEC}}, \quad \text{and}\quad S_{\calE,Q}S_{\calE,Q}^{\dag}=\sum_{i\in[S]}E_iPE_i^\dag=\calM_{\calE}(P)=\Gamma.
\label{eq:errorSynthesisDagers}
\end{equation}
Upon writing $P=\sum_{k\in[K]}\ket{c_k}\bra{c_k}$
for an orthonormal basis $\ppp{\ket{c_k}}_{k\in[K]}$ of $Q$, the operator $S_{\calE,Q}^\dag S_{\calE,Q}$ in \eqref{eq:errorSynthesisDagers} is precisely the operator represented by the QEC matrix $A_{\s{QEC}}$. 

Let $S_{\calE,Q}=UG$ be the polar decomposition of $S_{\calE,Q}$, where
\[G := \sqrt{S_{\calE,Q}^{\dag}S_{\calE,Q}} = \sqrt{A_{\s{QEC}}} = I_K\otimes\sqrt{\lambda}. \]
Writing $P_{\lambda}$ for the projection onto $\supp(\lambda)$, the projections 
of the partial isometry $U$ are
\[ U^\dag U=I_K\otimes P_{\lambda}, \qquad UU^\dag=P_{\Gamma}.\]
Note that for every codeword $\ket{\psi}\in Q$ and every $i\in[S]$, we have
\[E_i\ket{\psi}= S_{\calE,Q}\p{\ket{\psi}\otimes\ket{i}}=
U\p{I_K\otimes\sqrt{\lambda}}
\p{\ket{\psi}\otimes\ket{i}}=
U\p{\ket{\psi}\otimes\sqrt{\lambda}\ket{i}} \]
Since $\sqrt{\lambda}\ket{i}\in\supp(\lambda)$, applying $U^\dag$ gives $U^\dag E_i\ket{\psi} = \ket{\psi}\otimes\sqrt{\lambda}\ket{i}$.
Thus, $U^\dag$ 
separates the logical state from the error label. In particular, for every $\rho\in L(Q)$,
\[ \tr_S\p{U^\dag E_i\rho E_i^\dag U} = \lambda_{ii}\rho. \]
Hence, after normalization, the logical state is recovered perfectly for each individual error. More generally, trace preservation guarantees that the corresponding scalar factors sum to one for every channel whose Kraus operators lie in $\Span(\calE)$, and the same recovery map therefore exactly recovers every such channel.

This motivates the following definition of the polar decoder:
\begin{equation}
    \calR_{\calE,Q}^{\s{polar}}(X) := \tr_S\p{U^\dag XU} +\tr\p{\p{I-P_{\Gamma}}X}\rho_Q,\label{eq:DefPolar}
\end{equation}
where $\rho_Q\in L(Q)$ is an arbitrary code state. Since $P_{\Gamma}=UU^\dag$ is the projection onto the final support of $U$, the second term completes the map to a CPTP map on $L(\calH_{\SOUT})$. 
\begin{prop}\label{prop:polarPetz} For any code $Q$ and error set $\calE$ we have $ \calR_{\calE,Q}^{\s{polar}}= \calR_{\calE,Q}^{\s{Petz}}$. 
\end{prop}

\begin{proof}
    Since the $\supp(\Gamma)^\perp$-branch of both recovery maps is the same, by the definition of Petz recovery map and the polar recovery map, it is sufficient to show that for all $X$
    \[\tr_S\p{U^\dag XU}=\sum_{i\in [S]}PE_i^\dag\Gamma^{-1/2} X \Gamma^{-1/2}E_i P. \]
    Indeed, by the polar decomposition and \eqref{eq:errorSynthesisDagers}
    we have
\[U=\p{S_{\calE,Q}S_{\calE,Q}^\dag}^{-1/2}S_{\calE,Q}=\Gamma^{-1/2}S_{\calE,Q}=\sum_{i\in [S]}\Gamma^{-1/2}E_iP\otimes \bra{i}.\]
    The rest follows from a straightforward calculation:
    \begin{align}
        \tr_S\p{U^\dag XU}&=\tr_S\p{\sum_{i,j\in [S]} \p{P E_i^\dag \Gamma^{-1/2}X\Gamma^{-1/2} E_j P} \otimes \ket{i}\bra{j} }=\sum_{i\in [S]} P E_i^\dag \Gamma^{-1/2}X\Gamma^{-1/2} E_i P.\label{eq:Utrace}
    \end{align}
\end{proof}

The polar decoder of \eqref{eq:DefPolar} was introduced in \cite{ma2025haar} in the restricted setting of Haar random codes, or more generally, approximately nondegenerate codes and unitary Hilbert--Schmidt orthogonal errors. In particular \cite{ma2025haar} considered the SVD decomposition of the error synthesis operator $S_{\calE,Q}=\Sigma D \Theta$, where the diagonal matrix $D$ is close to the identity whenever $S_{\calE,Q}$ is an approximate isometry under an approximate non-degeneracy assumption. The suggested recovery map was then defined via a Stinespring dilation operator $U'$ obtained by rounding the diagonal elements of $D$ to $1$ and taking the adjoint, $U'=\Theta^\dag\Sigma^\dag$, which gives the partial isometry $U^\dag$ of the polar decoder. The equivalence of that polar decoder with the Petz recovery map, established in Proposition~\ref{prop:polarPetz} above, was not observed there. This equivalence allows us to combine the general theory of Petz recovery with the ideas of \cite{ma2025haar} to obtain both average-case and worst-case guarantees for general error sets, beyond both the channel setting usually considered for the Petz map and the restricted non-degenerate Pauli-like setting of \cite{ma2025haar}.

\subsection{Error-set performance of Petz recovery}
In this section we use the polar representation of the Petz map to provide uniform error-set guarantees on the AQEC capabilities of the Petz recovery map under both the worst- and average-case criteria. 
 \subsubsection*{Average-case analysis}
 We begin by analyzing the average-case error-set performance of the Petz map. In Theorem~\ref{th:AVpetz} we prove that the average-case error of the Petz error-set recovery map is uniformly bounded by the KL--Hellinger distance. 
 \begin{theorem}\label{th:AVpetz}
     Let $Q\subseteq\calH_{\SIN}$ be a $K$-dimensional quantum code, and let $\calE=\ppp{E_i}_{i\in[S]}$ be an error set with operators $E_i:\calH_{\SIN}\to\calH_{\SOUT}$. Then the error-set Petz recovery map $\calR_{\calE,Q}^{\s{Petz}}$ of Definition~\ref{def:errorSetPetz} satisfies 
     \[d_{\mathrm{ch}}\p{\calR_{\calE,Q}^{\s{\s{Petz}}}\circ \calN|_{L(Q)},I_{L(Q)}}\leq \zeta_{\s{H}}(\calE,Q),\]
    for any $\calE$-controlled channel $\calN$.
 \end{theorem}
  
Before proving Theorem~\ref{th:AVpetz}, we record a consequence concerning the suboptimality of the error-set Petz map. For a fixed channel, the channel-adapted Petz recovery map is known to be optimal up to a factor of $\sqrt{2}$ \cite{barnum2002reversing,li2025optimality}. This result does not directly extend to the error-set setting, where a single recovery map constructed from $\calE$ must perform uniformly over all $\calE$-controlled channels. By combining Theorem~\ref{th:AVpetz} with the necessary average-case error-set AQEC conditions of \cite{Elimelech2026Theory}, we nevertheless obtain a quantitative bound for Pauli-like error sets: when $\calE$ consists of Hilbert--Schmidt orthogonal unitary errors, the error achieved by the error-set Petz map differs from the optimal universal decoding error by a factor of at most order $\sqrt{|\calE|}$.
 \begin{cor}Let $\calE$ be an error set composed of Hilbert-Schmidt orthogonal unitary errors and let $Q$ be a quantum code. Then 
  \begin{align*} \frac{1}{\sqrt{2|\calE|}}\zeta_{\s{H}}(\calE,Q)&\leq \varepsilon_{\s{opt}}^{\s{av}}(\calE,Q)=\varepsilon_{\s{dec}}^{\s{av}}(\calE,Q)\\
  &\leq \sup_{\calN \in \mathscr{N}(\calE)}d_{\mathrm{ch}}\p{\calR_{\calE,Q}^{\s{\s{Petz}}}\circ \calN|_{L(Q)},I_{L(Q)}}\leq \zeta_{\s{H}}(\calE,Q).
  \end{align*}
 \end{cor}
 \begin{proof}
     \cite[Proposition 12]{Elimelech2026Theory} gives the inequality $\frac{1}{\sqrt{2|\calE|}}\zeta_{\s{H}}(\calE,Q)\leq \varepsilon_{\s{opt}}^{\s{av}}$.
     The equality follows from the equivalence of the AQEC decoding error and the optimal error in the average-case criterion, proved in Theorem~\ref{th:AvcaseOptimalDecoding}, and the rightmost inequality follows from Theorem~\ref{th:AVpetz}.
 \end{proof}

 \begin{proof}[Proof of Theorem~\ref{th:AVpetz}] 
     Our proof follows the approach of \cite{zheng2024near}. Let $\calN$ be an $\calE$-controlled channel with Kraus representation $\ppp{F_{\alpha}}_{\alpha\in [R]}$ such that $F_{\alpha}=\sum_{i}c_{\alpha,i}E_i$ and $\norm{C}_{\infty}\leq 1$. Let $A_{\calE}\in L(\C^K\otimes \C^S)$ be the QEC operator of the code associated with $\calE$. We begin with the simple observation that the trace preservation part in the Petz map vanishes on the code space.
     \begin{lemma}\label{lem:supportCodePetz}
         For any channel $\calN$ with Kraus representation $\ppp{F_\alpha}_\alpha\subseteq\Span(\calE)$, we have 
         \[\calR_{\calE,Q}^{\s{Petz}}\circ \calN|_{L(Q)}=\calR_{\Gamma}\circ \calN|_{L(Q)}, \]
         where $\calR_\Gamma$ has the Kraus representation $\ppp{R_{i}}_{i\in[S]}$ given in \eqref{eq:KrausSetPetz} and is defined as
         \[\calR_{\Gamma}(X)=\tr_S(U^\dag X U),\]
         where $U:Q\otimes\C^S\to\calH_{\SOUT}$ 
                is the partial isometry from the polar decomposition of the error-synthesis map $S_{\calE,Q}$.
     \end{lemma}
    \begin{proof}
         By definition of the Petz map, we can write $\calR_{\calE,Q}^{\s{Petz}}\circ \calN=\calR_{\Gamma}\circ \calN+\calR_{\Gamma^\perp}\circ \calN$, where 
     \[\calR_{\Gamma}(X)=\sum_{i\in[S]}R_iXR_i^\dag,\quad R_i=PE_i^\dag \Gamma^{-1/2}, \quad \text{and}\quad \calR_{\Gamma^{\perp}}(X)=\tr\p{(I-P_{\Gamma})X}\rho_Q.\]
     By \eqref{eq:errorSynthesisDagers} we have $\Gamma=S_{\calE,Q}S_{\calE,Q}^\dag$ and in particular \[\supp(\Gamma)=\ima(S_{\calE,Q})=\Span\ppp{E_i\ket{\psi}~:~ \psi\in Q, i\in [S]}.\]
     Therefore, for any basis operator $\ket{\psi_1}\bra{\psi_2}\in L(Q)$ we have 
     \[\calR_{\Gamma^\perp}\circ \calN(\ket{\psi_1}\bra{\psi_2})=\tr\p{(I-P_\Gamma)\sum_{\alpha}F_\alpha \ket{\psi_1}\bra{\psi_2} F_{\alpha}^\dag }\rho_{Q}=0,\]
     where the last equality follows since $F_{\alpha}\ket{\psi_1}=\sum_{i}c_{\alpha,i}E_i\ket{\psi_1}\in \supp\p{\Gamma}$ for all $\ket{\psi_1}\in Q$, and $I-P_{\Gamma}$ projects on $\supp(\Gamma)^{\perp}$. 
    \end{proof} 
    
     Lemma~\ref{lem:supportCodePetz} implies that on the code space the noise recovery map has Kraus representation $\ppp{R_i F_{\alpha}}_{i,\alpha}$ where $i\in [S]$ and $\alpha\in [R]$. Using the well-known channel fidelity formula of a channel with respect to the identity (see \cite[eq.~(8)]{barnum2002reversing}) we have 
     \begin{align}
     \calF_{\mathrm{ch}}\p{\calR_{\calE,Q}^{\s{\s{Petz}}}\circ \calN|_{L(Q)},I_{L(Q)}}&\nonumber=\sum_{\alpha=0}^{R-1}\sum_{i=0}^{S-1}\abs{\frac{1}{K}\tr\p{P R_i F_\alpha }}^2\\
     &\nonumber=\sum_{\alpha=0}^{R-1}\sum_{i=0}^{S-1}\abs{\frac{1}{K}\tr\p{P R_i F_\alpha P }}^2
     \\&=\sum_{\alpha=0}^{R-1}\sum_{i=0}^{S-1}\abs{\sum_{j=0}^{S-1}\frac{1}{K}c_{\alpha,j}\tr\p{P E_i^\dag \Gamma^{-1/2} E_j P }}^2.\label{eq:FidelcalcAV}
     \end{align}
  To proceed with the analysis of \eqref{eq:FidelcalcAV} we need the following lemma.
      \begin{lemma}
        For all $i,j\in [S]$, $\tr\p{P E_i^\dag \Gamma^{-1/2} E_j P }$ is exactly the $(i,j)$th entry of the QEC matrix $\tr_K\p{\sqrt{A_{\s{QEC}}}}$.
      \end{lemma}
       \begin{proof}
           Indeed, recall that by the polar decomposition,  $\sum_{j\in[S]}E_j P \otimes \bra{j}=S_{\calE,Q}=U G$, where 
     \begin{align*}
     G&=\sqrt{S_{\calE,Q}^\dag S_{\calE,Q} }=\sqrt{A_{\s{QEC}}}, \\ U&=\p{S_{\calE,Q}S_{\calE,Q}^\dag}^{-1/2}S_{\calE,Q}=\Gamma^{-1/2}S_{\calE,Q}=\sum_{i\in [S]}\Gamma^{-1/2}E_iP\otimes \bra{i}.
     \end{align*}
     In particular 
     \begin{equation}
         \sqrt{A_{\s{QEC}}}=G=U^\dag S_{\calE,Q}=\sum_{i,j\in [S]}P E_i^\dag \Gamma^{-1/2}E_j P\otimes \ket{i} \bra{j},\label{eq:SqrtQECmatrix}
     \end{equation}
     which implies that 
     \[\tr_K\p{\sqrt{A_{\s{QEC}}}}=\sum_{i,j\in[S]}\tr\p{P E_i^\dag \Gamma^{-1/2}E_j P}  \ket{i} \bra{j}. \qedhere
     \]
       \end{proof}

We continue with the proof of Theorem \ref{th:AVpetz}.     Considering the coefficient matrix $C=\sum_{\alpha,j}c_{\alpha,j}\ket{\alpha}\bra{j}$ as an operator $\C^S\to\C^R$ we obtain 
     \begin{equation*}
         \abs{\sum_{j=0}^{S-1}c_{\alpha,j}\tr\p{P E_i^\dag \Gamma^{-1/2} E_j P }}^2= \abs{\p{C\tr_K\p{\sqrt{ A_{\s{QEC}}}}^T}_{\alpha,i}}^2.
     \end{equation*}
Denote $A:=\tr_K(\sqrt{A_{\s{QEC}}})$. Combining the last equality with \eqref{eq:FidelcalcAV} we have
     \begin{align}
         \calF_{\mathrm{ch}}\p{\calR_{\calE,Q}^{\s{\s{Petz}}}\circ \calN|_{L(Q)},I_{L(Q)}}&\nonumber=\frac{1}{K^2}\sum_{i=0}^{S-1}\sum_{\alpha=0}^{R-1}|(CA^T)_{\alpha,i}|^2=\frac{1}{K^2}\norm{CA^T}_2^2\\
         &=\frac{1}{K^2}\tr\p{C A^T (A^T)^\dag C^\dag }=\frac{1}{K^2}\tr\p{C (A^T)^2 C^\dag }\label{eq:potivitysrt}\\
         &=\frac{1}{K^2}\tr\p{(C (A^T)^2 C^\dag)^T}=\frac{1}{K^2}\tr\p{\Bar{C} A^2 C^T}=\frac{1}{K^2}\tr\p{A^2\Bar{C}^\dag\Bar{C}}.\label{eq:traceProperuse}
     \end{align} 
     In \eqref{eq:potivitysrt} we used the fact that $\sqrt{A_{\s{QEC}}}\succeq 0$ and the partial trace is CP, and in \eqref{eq:traceProperuse} we used the commutativity and transpose invariance of the trace. 

     Our goal is to bound the quantity 
     \begin{equation}
         1-\calF_{\mathrm{ch}}\p{\calR_{\calE,Q}^{\s{\s{Petz}}}\circ \calN|_{L(Q)},I_{L(Q)}}=1-\frac{1}{K^2}\tr\p{\tr_K\p{\sqrt{A_{\s{QEC}}}}^2 \Bar{C}^\dag \Bar{C}}.\label{eq:Fidelitybound}
     \end{equation}
     To that end, consider the QEC matrix of the code $Q$ with respect to the Kraus set $\ppp{F_{\alpha}}_\alpha$ of the channel $\calN$, denoted by $A_{\s{QEC}}^{\calN}$. A straightforward calculation shows that when  $A_{\s{QEC}}$ and $A_{\s{QEC}}^{\calN}$ are computed with respect to the same basis of $Q$ (or when considered in operator form) we have the following relation: 
     \[ A_{\s{QEC}}^{\calN}=(I_K\otimes \Bar{C})A_{\s{QEC}}(I_K\otimes \Bar{C}^\dag).\]
     On the other hand, since $\calN$ is trace preserving, we have
     \begin{equation}
         \tr\p{A_{\s{QEC}}^{\calN}}=\tr\p{\sum_{\alpha,\beta \in [R]}(PF_\alpha^\dag F_{\beta}P)\otimes \ket{\alpha}\bra{\beta}}=\sum_{\alpha\in[R]}\tr(F_{\alpha} P F_{\alpha}^\dag)=\tr(\calN(P))=K.\label{eq:conjugationconjugation}
     \end{equation}
     Combining the above with \eqref{eq:Fidelitybound} and \eqref{eq:conjugationconjugation} we have:
     \begin{align}
          1-\calF_{\mathrm{ch}}\p{\calR_{\calE,Q}^{\s{\s{Petz}}}\circ \calN|_{L(Q)},I_{L(Q)}}&\nonumber=\frac{1}{K}\tr\p{(I_K\otimes \Bar{C})  A_{\s{QEC}}(I_K\otimes \Bar{C})^\dag}-\frac{1}{K^2}\tr\p{\tr_K\p{\sqrt{A_{\s{QEC}}}}^2 \Bar{C}^\dag \Bar{C}}\\
         &=\frac{1}{K}\tr\p{\Bar{C}\,\tr_K\p{A_{\s{QEC}}}\Bar{C}^\dag}-\frac{1}{K^2}\tr\p{\tr_K\p{\sqrt{A_{\s{QEC}}}}^2 \Bar{C}^\dag \Bar{C}}\label{eq:UsePartialTraceProp}\\
         &=\tr\p{\p{\frac{1}{K}\tr_K(A_{\s{QEC}})-\frac{1}{K^2}\tr_K\p{\sqrt{A_{\s{QEC}}}}^2}\Bar{C}^\dag \Bar{C}},\label{eq:AlmostFinalTrace}
     \end{align}
     where in \eqref{eq:UsePartialTraceProp} we used Lemma~\ref{lem:tracPropComposite}. Note that $I_S-\Bar{C}^\dag \Bar{C}\succeq 0$ by the assumption that $\norm{\Bar{C}}_{\infty}=\norm{C}_{\infty}\leq 1$. By Lemma~\ref{lem:partialTacesSqrts} we also have: 
     \[\frac{1}{K}\tr_K(A_{\s{QEC}})-\frac{1}{K^2}\tr_K\p{\sqrt{A_{\s{QEC}}}}^2\succeq 0.\]
     In particular 
          \begin{equation}
              \tr\p{\p{\frac{1}{K}\tr_K(A_{\s{QEC}})-\frac{1}{K^2}\tr_K\p{\sqrt{A_{\s{QEC}}}}^2}(I_S-\Bar{C}^\dag \Bar{C})}\geq 0,\label{eq:prodPOsitivetrace}
          \end{equation}
     as the trace of the composition of positive semidefinite operators is always non-negative. Combining \eqref{eq:prodPOsitivetrace} with \eqref{eq:AlmostFinalTrace} and \eqref{eq:KLHdistanceExpression} we obtain:
     \begin{align*}
          1-\calF_{\mathrm{ch}}\p{\calR_{\calE,Q}^{\s{\s{Petz}}}\circ \calN|_{L(Q)},I_{L(Q)}}&\leq \tr\p{\p{\frac{1}{K}\tr_K(A_{\s{QEC}})-\frac{1}{K^2}\tr_K\p{\sqrt{A_{\s{QEC}}}}^2}}\\
         &=\frac{1}{K}\tr\p{A_{\s{QEC}}}-\frac{1}{K^2}\norm{\tr_K\p{\sqrt{A_{\s{QEC}}}}}_2^2\\
         &=\zeta_{\s{H}}(\calE,Q)^2.
     \end{align*}
     We conclude by taking a square root on both sides of the above expression. 
 \end{proof}

\begin{remark}[Trace-nonincreasing CP maps]\label{rem:AVpetzCP}
The same argument applies when $\calN$ is a trace-nonincreasing CP map whose Kraus operators satisfy the coefficient condition of Definition~\ref{def:errorSetAQEC}. We call such a map an $\calE$-controlled CP map. For this map, $\calF_{\mathrm{ch}}\p{\calR_{\calE,Q}^{\s{Petz}}\circ\calN|_{L(Q)},I_{L(Q)}}$ denotes the unnormalized entanglement-fidelity contribution given by~\eqref{eq:FidelcalcAV}, without conditioning on the retained trace. Then
\begin{equation}\label{eq:AVpetzCP}
 0\leq\tr\p{\calN(P/K)}
 -\calF_{\mathrm{ch}}\p{\calR_{\calE,Q}^{\s{Petz}}\circ\calN|_{L(Q)},I_{L(Q)}}
 \leq\zeta_{\s{H}}(\calE,Q)^2.
\end{equation}
Indeed, the support argument in Lemma~\ref{lem:supportCodePetz} does not use trace preservation, and~\eqref{eq:conjugationconjugation} gives $\tr\p{A_{\s{QEC}}^{\calN}}=\tr\p{\calN(P)}$. Thus~\eqref{eq:AlmostFinalTrace} holds with $1$ on the l.-h.s replaced by $\tr\p{\calN(P/K)}$ and the same positivity and contraction arguments give~\eqref{eq:AVpetzCP}.
\end{remark}

 \subsubsection*{Worst-case analysis}
Our next goal is to analyze the worst-case AQEC of the Petz map in the error-set setting. While we do not provide AQEC guarantees against all $\calE$-controlled channels, we can formulate them for a subset of channels isolated by a stronger 
constraint.
\begin{definition}[Strongly $\calE$-controlled channels]\label{def:SronglyCont}
Let $\calE=\ppp{E_i}_i$ be an error set. A channel $\calN$ is called \emph{strongly $\calE$-controlled}, if it admits a Kraus representation $\ppp{F_\alpha}_{\alpha}$ such that $F_{\alpha}=\sum_{i}c_{\alpha,i}E_i$ and the coefficient matrix $C=(c_{\alpha,i})_{\alpha,i}$ satisfies $\norm{C}_{2}\leq 1$.
\end{definition}

Many natural error sets give rise to rich families of strongly $\calE$-controlled channels, often containing every channel whose Kraus operators lie in $\Span(\calE)$. The canonical example is an error set composed of unitary Hilbert--Schmidt orthogonal operators, which includes Pauli-type and Majorana errors, among others. For such an error set $\calE$, every channel admitting a Kraus representation with operators in $\Span(\calE)$ is automatically strongly $\calE$-controlled, since its coefficient matrix satisfies $\norm{C}_{2}=1$ \cite[Proposition~2.3]{ma2025haar}. In particular, the family of $t$-limited error channels is strongly controlled by the set of Pauli operators of weight at most $t$; see \cite[Example~1]{Elimelech2026Theory} for further discussion. Similarly to $\calE$-controlled channels, the family of strongly $\calE$-controlled channels is convex. This observation also implies that random erasure and deletion channels are strongly controlled by natural choices of error sets; see \cite[Sections~2.1.2, 2.1.3]{Elimelech2026Theory}.

\begin{theorem}\label{th:WCpetz}
     Let $Q\subseteq\calH_{\SIN}$ be a $K$-dimensional quantum code, and let $\calE=\ppp{E_i}_{i\in[S]}$ be an error set with operators $E_i:\calH_{\SIN}\to\calH_{\SOUT}$. Then  for any strongly $\calE$-controlled channel $\calN$, the error-set Petz recovery map $\calR_{\calE,Q}^{\s{Petz}}$ of Definition~\ref{def:errorSetPetz} satisfies 
     \begin{equation}
         d_{\mathrm{wc}}\p{\calR_{\calE,Q}^{\s{\s{Petz}}}\circ \calN|_{L(Q)},I_{L(Q)}}\leq \frac{1}{\sqrt{2}}\norm{\sqrt{A_{\s{QEC}}}-I_{K}\otimes \frac{1}{\sqrt{K}}\sqrt{\tr_K(A_{\s{QEC}})}}_{\infty}.\label{eq:PetzWCcond}
     \end{equation}
 \end{theorem}
\begin{proof}
    We adopt the strategy used in the proof of Theorem~\ref{th:DecExistsWorst} to provide worst-case guarantees for the proposed recovery map. Given the Petz recovery map and a strongly $\calE$-controlled channel specified by some 
    coefficient matrix $C$, we find a Stinespring representation of the decoded noisy output $\calR_{\calE,Q}^{\s{Petz}}\circ \calN|_{L(Q)}$ whose distance from some Stinespring dilation operator of the identity is bounded by the r.-h.s.
        of \eqref{eq:PetzWCcond}. A similar approach was also used in \cite{ma2025haar}. We then use the Stinespring dilation continuity theorem (see Theorem~\ref{th:SteinContChannels}) to conclude that $\calR_{\calE,Q}^{\s{Petz}}\circ \calN|_{L(Q)}$ is close to the identity with respect to the worst-case Bures distance $d_{\mathrm{wc}}$. 

    Let $P$ be the projector on the code space $Q$, and denote $\lambda=\frac{1}{K}\tr_K(A_{\s{QEC}})$. Observe that by \eqref{eq:errorSynthesisDagers}
    \begin{equation}
        \lambda=\frac{1}{K}\tr_K(A_{\s{QEC}})=\frac{1}{K}\sum_{i,j\in[S]}\tr(PE_i^\dag E_j P)  \ket{i}\bra{j}.\label{eq:lambdaDef1}
    \end{equation}
    where $(\ket i)_{i\in [S]}$ is a basis of $\C^S$.
    Let $\calN_C:L(\calH_{\SIN})\to L(\calH_{\SOUT})$ be a strongly $\calE$-controlled channel with Kraus operators $\ppp{F_\alpha}_{\alpha\in [R]}$, $R\in \N$, such that $C=(c_{\alpha,i})_{\alpha,i}$ is the linear coefficient matrix defined by the equation $F_\alpha=\sum_{i\in [S]}c_{\alpha,i}E_i$ that satisfies $\norm{C}_2\leq 1$. By Lemma~\ref{lem:supportCodePetz} we can write 
    \begin{equation}
        \calR_{\calE,Q}^{\s{Petz}}\circ \calN_C|_{L(Q)}=\calR_{\Gamma}\circ \calN_{C}|_{L(Q)}, \quad \calR_\Gamma(X)=\tr_S(U^\dag XU),\label{eq:OnTheCodeSPacePETZ}
    \end{equation}
    where $U:Q\otimes \calH_S\to \calH_{\SOUT}$ is the polar partial isometry in the polar decomposition of the error synthesis map $S_{\calE,Q}$ defined in \eqref{eq:ErrorSynthDef}, such that 
    \[S_{\calE,Q}= U \sqrt{S_{\calE,Q}^\dag S_{\calE,Q}}= U \sqrt{A_{\s{QEC}}}.\]
In particular, since $U:Q\otimes \calH_S\to \calH_{\SOUT} $ is invertible on 
    $\supp(S_{\calE,Q})=\supp (A_{\s{QEC}})$, we have 
    \begin{equation}
        U^\dag S_{\calE,Q}=\sqrt{A_{\s{QEC}}}, \label{eq:polarDager}
    \end{equation}

    Consider the noise operation on the code state $\calN_C|_{L(Q)}$, which can be formally described as $\calN_C \circ \calN_Q$, where $\calN_Q:L(Q)\to L(\calH_{\SIN})$ is simply the inclusion isometry $\calN_Q(X)=PXP$. It admits a Kraus representation  $\ppp{F_{\alpha}P}_{\alpha\in [R]}$ and has a canonical Stinespring dilation operator $V_C:Q \to \calH_{\SOUT}\otimes \calH_R$ of the form
    \begin{equation}\label{eq: VC}
        V_{C}=\sum_{\alpha\in[R]}F_{\alpha}P\otimes \ket{\alpha}=\sum_{i\in [S]}\sum_{\alpha\in[R]}c_{\alpha,i}E_iP \otimes \ket{\alpha}=\sum_{\alpha\in[R]} \p{\sum_{i\in [S]}c_{\alpha,i}E_iP }\otimes \ket{\alpha}.
    \end{equation}
    By  the definition of $S_{\calE,Q}$ in \eqref{eq:ErrorSynthDef}, for any $\ket{\psi}$ we have $E_i P\ket{\psi}=S_{\calE,Q}(\ket{\psi}\otimes \ket{i})$, so \eqref{eq: VC} gives rise to
    \begin{align}
        V_{C}\ket{\psi}\nonumber&=\sum_{\alpha\in[R]} \p{\sum_{i\in [S]}c_{\alpha,i}S_{\calE,Q}(\ket{\psi}\otimes \ket{i}) }\otimes \ket{\alpha}=\sum_{\alpha\in[R]} \sum_{i\in [S]}c_{\alpha,i}(S_{\calE,Q}\otimes I_R)(\ket{\psi}\otimes \ket{i}\otimes \ket{\alpha})\\
        &=(S_{\calE,Q}\otimes I_R)\Bigg(\ket{\psi}\otimes \Big(\sum_{\alpha\in[R]} \sum_{i\in [S]}c_{\alpha,i}\ket{i}\otimes \ket{\alpha}\Big)\Bigg)=(S_{\calE,Q}\otimes I_R)\p{\ket{\psi}\otimes \ket{\s{vec}(C)}},\label{eq:VecRepV_C}
    \end{align}
    where, $\ket{\s{vec}(C)}=\sum_{\alpha,i}c_{\alpha,i}\ket{i}\otimes \ket{\alpha}\in \calH_S\otimes \calH_R$ is the vectorization of $C$. 
    Next, we find a Stinespring dilation operator $V_{\s{dec},C}$ for the decoded noisy state on the code space $\calR_{\calE,Q}^{\s{Petz}}\circ \calN_{C}|_{L(Q)}$.  Let $V_{\s{dec},C}: Q\to Q \otimes \calH_S\otimes \calH_R$ be given as $V_{\s{dec,C}}=(U^\dag \otimes I_R )V_C$ and note that for any $X\in L(Q)$ we have 
    \begin{align*}
        \tr_{R,S}\p{V_{\s{dec},C}X V_{\s{dec},C}^\dag}&=\tr_{R,S}\p{(U^\dag \otimes I_R)(V_C X V_C^\dag)(U \otimes I_R) }\\
        &=\tr_S\p{U^\dag \tr_R(V_C X V_C^\dag)U }=\calR_{\calE,Q}^{\s{Petz}}\circ \calN_C(X).
    \end{align*}
     where the second equality follows from Lemma~\ref{lem:tracPropComposite} and the last equality follows from \eqref{eq:OnTheCodeSPacePETZ}. 

     Our goal is to show now that $V_{\s{dec},C}$ is close to some Stinespring dilation of the identity channel on $L(Q)$. Consider the operator $V_{\lambda_C}:Q\to Q\otimes \calH_S\otimes \calH_R$ given by 
     \[V_{\lambda_C}:=I_Q\otimes \ket{\psi_C}, \quad \ket{\psi_C}=(\sqrt{\lambda}\otimes I_R)\ket{\s{vec}(C)}. \]
     We show that $V_{\lambda_C}$ is indeed a dilation for $I_{L(Q)}$. Observe that
     \[\tr_{R,S}(V_{\lambda_C} XV_{\lambda_C}^\dag)=\braket{\psi_C|\psi_C}X, \]
    so it is sufficient to show that  $\ket{\psi_C}$ is normalized. This follows from channel normalization and from \eqref{eq:lambdaDef1}:
    \begin{align*}
        \braket{\psi_C|\psi_C}&=\bra{{\s{vec}(C)}}({\lambda}\otimes I_R)\ket{\s{vec}(C)}=\sum_{i,j\in [S]}\sum_{\alpha,\beta\in [R]}c_{\alpha,i}^* c_{\beta,j}\bra{i}\lambda \ket{j}\braket{\alpha |\beta}\\
        &=\sum_{\alpha\in[R]}\sum_{i,j\in [S]}c_{\alpha,i} ^*c_{\alpha,j}\frac{1}{K}\tr\p{P E_i^\dag E_j P}=\tr\p{\frac{1}{K}\sum_{\alpha \in[R]} F_\alpha PF_{\alpha}^\dag }=\tr\p{\calN_C\big(\frac{P}{K}\big)}=1.
    \end{align*}

    We now bound the $\infty$-norm distance between $V_{\s{dec},C}$ and $V_{\lambda_C}$. For any normalized $\ket{\psi}\in Q$, by \eqref{eq:polarDager} and \eqref{eq:VecRepV_C} we have
    \begin{align*}
        \norm{V_{\s{dec},C}\ket{\psi}-V_{\lambda_C}\ket{\psi}}&=\norm{(U^\dag \otimes I_R)V_C\ket{\psi}-\ket{\psi}\otimes \ket{\psi_C}} \\
        &=\norm{\p{(U^\dag S_{\calE,Q} \otimes I_R)-I_K\otimes\sqrt{\lambda}\otimes I_R}(\ket{\psi}\otimes \ket{\s{vec}(C)} )}\\
        &\leq \norm{(U^\dag S_{\calE,Q} \otimes I_R)-I_K\otimes\sqrt{\lambda}\otimes I_R}_\infty \norm{\ket{\psi}}\norm{\ket{\s{vec}(C)}}\\
        &=\norm{\sqrt{A_{\s{QEC}}}-I_K \otimes \sqrt{\lambda}}_\infty\norm{C}_2
        \\&\leq \norm{\sqrt{A_{\s{QEC}}}-I_K \otimes \sqrt{\lambda}}_\infty,
    \end{align*}
    where we used the fact that $\norm{\ket{\s{vec}(C)}}=(\sum_{\alpha,i}|c_{\alpha,i}|^2)^{1/2}=\norm{C}_2$. Combining the above with Stinespring dilation continuity (Theorem~\ref{th:SteinContChannels}) we have 
    \[d_{\mathrm{wc}}\p{\calR_{\calE,Q}^{\s{Petz}}\circ \calN_C|_{L(Q)},I_{L(Q)}}\leq \frac{1}{\sqrt{2}}\norm{V_{\s{dec},C}-V_{\lambda_C}}_{\infty}\leq \frac{1}{\sqrt{2}}\norm{\sqrt{A_{\s{QEC}}}-I_K \otimes \sqrt{\lambda}}_\infty.\]
\end{proof}

\begin{remark} Theorem~\ref{th:WCpetz} recovers the guarantees established for the polar decoder in \cite{ma2025haar}. The $\delta$-approximate nondegeneracy condition considered there requires the error-synthesis operator $S_{\calE,Q}$ to be a $\delta$-approximate isometry, in the sense that each of its singular values differs from $1$ by at most $\delta$. Since the singular values of $S_{\calE,Q}$ are precisely the eigenvalues of $\sqrt{S_{\calE,Q}^{\dag}S_{\calE,Q}} = \sqrt{A_{\s{QEC}}}$, this assumption is equivalent to $\norm{\sqrt{A_{\s{QEC}}}-I_{KS}}_{\infty} \leq \delta.$ Theorem~\ref{th:WCpetz} then shows that $\calR_{\calE,Q}^{\s{Petz}}$ is an $\s{wc}$-$\delta$ universal decoder for all $\calE$-controlled channels. When $\calE$ consists of unitary Hilbert--Schmidt orthogonal errors, this family contains every channel whose Kraus operators lie in $\Span(\calE)$. \end{remark}

\begin{prop}
    Let $Q\subseteq\calH_{\SIN}$ be a $K$-dimensional quantum code, and let $\calE=\ppp{E_i}_{i\in[S]}$ be an error set with operators $E_i:\calH_{\SIN}\to\calH_{\SOUT}$. Then
    \begin{equation}
        \norm{\sqrt{A_{\s{QEC}}}-I_{K}\otimes \frac{1}{\sqrt{K}}\sqrt{\tr_K(A_{\s{QEC}})}}_{\infty}\leq  {\sqrt{2S\zeta(\calE,Q)}}.\label{eq:BOinequality}
    \end{equation}
    In particular, the error-set Petz recovery map $\calR_{\calE,Q}^{\s{Petz}}$ of Definition~\ref{def:errorSetPetz} satisfies 
     \begin{equation}
        d_{\mathrm{wc}}\p{\calR_{\calE,Q}^{\s{\s{Petz}}}\circ \calN|_{L(Q)},I_{L(Q)}}\leq  {\sqrt{S\zeta(\calE,Q)}},\label{eq:BOboundPetzWC}
     \end{equation}
     for any strongly $\calE$-controlled channel $\calN$.
\end{prop}
\begin{proof}
    Let us denote 
    \[\lambda=\frac{1}{K} \tr_K\p{A_{\s{QEC}}}=\sum_{i,j\in [S]}\frac{1}{K}\tr(PE_i^\dag E_jP) \ket{i}\bra{j},\]
    where the last equality follows from \eqref{eq:errorSynthesisDagers}. We first show that 
    \begin{equation}
        \norm{\sqrt{A_{\s{QEC}}}-I_{K}\otimes \frac{1}{\sqrt{K}}\sqrt{\tr_K(A_{\s{QEC}})}}_{\infty}\leq \sqrt{\norm{A_{\s{QEC}}-I_{K}\otimes \frac{1}{K}\tr_K(A_{\s{QEC}})}_{\infty}}.\label{eq:UsingHolder}
    \end{equation}
    Indeed, it is well-known that the function $f(t)=\sqrt{t}$ is operator-monotone on $[0,\infty)$ in the sense of Definition~\ref{eq:operatorMonotone} (see  \cite[Proposition V.1.8]{bhatia2013matrix}). Thus, \eqref{eq:UsingHolder}  follows directly by applying the  H\"older type inequality of Lemma~\ref{lem:HolderTypef} on $f$. To finish the proof, it is sufficient to show that the r.h.s. of \eqref{eq:UsingHolder} is bounded by $ {\sqrt{2S\zeta(\calE,Q)}}$. 

    To that end, consider the B\'eny-Oreshkov superoperator corresponding to $ {\lambda^T}$:
    \[\calB^{\calE}_{ {\lambda^T},Q}(X)= \sum_{i,j}\tr\p{B_{i,j}X}\ket{i}\bra{j}, \quad B_{i,j} = P {E_j^\dag E_i}P- {\lambda_{j,i}}P.\]
    Note that $\calB_{ {\lambda^T},Q}^{\calE}$ is Hermitian preserving as $\lambda^T$ is a Hermitian matrix (see \eqref{eq:HermitianPereserving} for the details). By \cite[Remark 8]{Elimelech2026Theory}, Lemma~\ref{lem:superoperatorNorms}, and the duality principle of Lemma~\ref{lem:Holder duality} we have 
    \[\norm{(\calB^{\calE}_{ {\lambda^T},Q})^\dag}_{\s{cb},\infty}=\norm{\calB^{\calE}_{ {\lambda^T},Q}}_{\s{cb},1}=\norm{\calB^{\calE}_{ {\lambda^T},Q}}_{\diamond}\leq 2\zeta(\calE,Q),\]
    where $(\calB^{\calE}_{ {\lambda^T},Q})^\dag:L(\calH_{S})\to L(Q)$ is the Hilbert-Schmidt dual operator of $\calB^{\calE}_{ {\lambda^T},Q}$. We claim that $(\calB^{\calE}_{ {\lambda^T},Q})^\dag$ is given as:
    \[(\calB^{\calE}_{ {\lambda^T},Q})^\dag(X)=\sum_{i,j\in[S]}\bra{i}X\ket{j}B_{j,i}.\]
    Indeed, for any $Y\in L(Q)$ and  $X\in L(\calH_{S})$ we have
    \begin{align*}
        \innerP{\calB^{\calE}_{ {\lambda^T},Q}(Y),X}_{\s{HS}}&=\tr\p{\sum_{i,j}\tr(B_{i,j}Y)^*\ket{j}\bra{i}X}=\sum_{i,j}\tr(Y^\dag B_{i,j}^\dag )\bra{i}X\ket{j}\\
        &=\tr\p{\sum_{i,j}\bra{i}X\ket{j} B_{i,j}^\dag Y^\dag}=\tr\p{Y^\dag \bigg( \sum_{i,j}\bra{i}X\ket{j} B_{j,i} \bigg)  }=\innerP{Y,(\calB^{\calE}_{ {\lambda^T},Q})^\dag (X)}_{\s{HS}},
    \end{align*}
    where we used the equality $B_{i,j}=B_{j,i}^\dag$ which follows since $\lambda$ is Hermitian. Consider the matrix $X\in \calH_S\otimes \calH_S$ given by $X=\sum_{i,j\in[S]}\ket{i}\bra{j}\otimes  {\ket{i}\bra{j}}$ where $\ppp{\ket{i}}_i$ is an orthonormal basis, and note that $ {\norm{X}_{\infty}=S}$. By the definition of $\norm{\cdot}_{\s{cb},\infty}$ we have 
    \begin{equation}
        \norm{\p{(\calB_{ {\lambda^T},Q}^{\calE})^\dag\otimes I_S}(X)}_{\infty}\leq  {S}\norm{(\calB_{ {\lambda^T},Q}^{\calE})^\dag}_{\s{cb},\infty}.\label{eq:cbINFTYbound}
    \end{equation}
    On the other hand 
    \begin{align}
        \p{(\calB_{ {\lambda^T},Q}^{\calE})^\dag\otimes I_S}(X)&\nonumber=\sum_{i,j\in[S]}B_{j,i}\otimes\ket{i}\bra{j}\\
        &=\sum_{i,j\in[S]}P {E_i^\dag E_j}P\otimes\ket{i}\bra{j}- \sum_{i,j\in[S]}{\lambda_{i,j}} P\otimes\ket{i}\bra{j},\label{eq:tensorProdIk1}\\
        &= A_{\s{QEC}}-I_K\otimes \lambda \label{eq:tensorProdIk2}
    \end{align}
    where the first equality of \eqref{eq:tensorProdIk1} follows by \eqref{eq:errorSynthesisDagers} and \eqref{eq:tensorProdIk2} follows from the definition of $\lambda$ and since when restricted to the $K$-dimensional codespace $Q$, the projector $P$ acts as the identity $I_K$. Combining \eqref{eq:UsingHolder}, \eqref{eq:cbINFTYbound} and \eqref{eq:tensorProdIk2} we conclude \eqref{eq:BOinequality}. Finally, \eqref{eq:BOboundPetzWC} follows from \eqref{eq:BOinequality} and Theorem~\ref{th:WCpetz}.
\end{proof}

\subsection{Channel performance of Petz recovery} The performance of the Petz recovery map for a fixed channel is largely understood under the channel-fidelity criterion. Its near-optimality and exact optimality properties, as well as explicit expressions for its performance in terms of the QEC matrix, have been studied extensively \cite{barnum2002reversing,ng2010simple,zheng2024near,li2025optimality,kim2026optimal}. Existing worst-case results are more limited and generally address different criteria or rely on additional structure. In \cite{ng2010simple}, near-optimality was established with respect to the minimum pure-state input-output fidelity, rather than the worst-case entanglement fidelity. A general worst-case trace-norm guarantee for the Petz map, with a code-dimension-dependent loss, was obtained in the subspace and operator-algebra settings in \cite{chen2020entanglement}. More recently, \cite{elovenkova2026covariant} analyzed the worst-case entanglement fidelity of Petz recovery for covariant codes under flagged single-qudit erasure. Their analysis relies on covariance and irreducibility, which imply that the worst-case input is the maximally mixed logical state and thereby reduce the problem to a channel-fidelity calculation. 

Beyond these settings, the literature remains predominantly based on channel fidelity. One reason is the tractable Kraus-operator expression for channel fidelity, due to Schumacher \cite{schumacher1996sending}, which gives a workable expression for the channel fidelity. By contrast, worst-case entanglement fidelity requires minimization over all logical inputs, including states entangled with a reference system, and its general characterization leads to the complementary-channel formalism of B\'eny and Oreshkov \cite{beny2010general}. Building on the polar decomposition analysis of Theorem~\ref{th:WCpetz}, we derive worst-case entanglement-fidelity guarantees for the Petz recovery map for general finite-dimensional channels, without covariance assumptions.

We now adapt the error-set machinery to the case where the error set $\calE_{\calN}=\ppp{E_i}_{i\in[S]}$ is a Kraus set of a quantum channel $\calN:L(\calH_{\SIN})\to L(\calH_{\SOUT})$. The channel $\calN$ is automatically $\calE_{\calN}$-controlled: using its given Kraus representation, the coefficient matrix is $C=I_S$, and hence $\norm{C}_{\infty}=1$. By contrast, this representation does not satisfy the strong control condition, since $ \norm{C}_{2} = \norm{I_S}_{2} = \sqrt{S}$. To apply Theorem~\ref{th:WCpetz}, consider instead the rescaled error set $\calE_{\calN}' := \sqrt{S}\calE_{\calN} = \ppp{\sqrt{S}E_i}_{i\in[S]}$. The channel $\calN$ is strongly $\calE_{\calN}'$-controlled, since its Kraus operators admit the representation $E_i = \frac{1}{\sqrt{S}}\p{\sqrt{S}E_i}$, whose coefficient matrix is $I_S/\sqrt{S}$ and therefore has norm one. A common rescaling of the error operators does not change the corresponding Petz recovery map. On the other hand, the QEC matrix associated with $\calE_{\calN}'$ is $S A_{\s{QEC}}$. Applying Theorem~\ref{th:WCpetz} therefore gives the following channel-level bound. 

\begin{prop}\label{prop:channelWCPetz} Let $\calN:L(\calH_{\SIN})\to L(\calH_{\SOUT})$ be a quantum channel with Kraus set $\calE_{\calN}=\ppp{E_i}_{i\in[S]}$, and let $Q\subseteq\calH_{\SIN}$ be a $K$-dimensional quantum code. Then 
\[d_{\mathrm{wc}}\p{ \calR_{\calN,Q}^{\s{Petz}} \circ \calN|_{L(Q)}, I_{L(Q)} } \leq \sqrt{\frac{S}{2}} \norm{ \sqrt{A_{\s{QEC}}} - I_K\otimes \frac{1}{\sqrt{K}} \sqrt{\tr_K\p{A_{\s{QEC}}}} }_{\infty}. \]\end{prop} Indeed, the QEC matrix of the rescaled error set is $A_{\s{QEC}}' = S A_{\s{QEC}} $, and hence \begin{align*}
     \norm{ \sqrt{A_{\s{QEC}}'} - I_K\otimes \frac{1}{\sqrt{K}} \sqrt{\tr_K\p{A_{\s{QEC}}'}} }_{\infty} = \sqrt{S} \norm{ \sqrt{A_{\s{QEC}}} - I_K\otimes \frac{1}{\sqrt{K}} \sqrt{\tr_K\p{A_{\s{QEC}}}} }_{\infty}. 
\end{align*} The claim now follows from Theorem~\ref{th:WCpetz}.

\providecommand{\wt}{\operatorname{wt}}
\definecolor{NEpetzblue}{HTML}{245C9F}
\definecolor{NEpetzorange}{HTML}{CF6E12}
\definecolor{NEpetzgreen}{HTML}{248354}
\section{Numerical examples}\label{sec:numericalExperiments}
We give numerical experiments that highlight the universal error-set recovery guarantees developed in this paper. We first examine amplitude damping on binomial codes, where a Petz map adapted to one loss probability can perform poorly at smaller loss probabilities, while a single error-set Petz map remains effective. We then examine how closely the average-case and worst-case bounds capture universal recovery under Pauli error sets, using square Shor codes and Haar-random codes. 

\subsection{Universal recovery for amplitude damping}\label{sec:binomialExperiment}
A Petz map adapted to a given channel is near optimal in the average-case criterion~\cite{barnum2002reversing}. This guarantee concerns the channel used to construct the recovery. We examine what happens when the actual noise is weaker than the noise for which the recovery was designed. We fix a design loss probability $\delta$, construct the Petz map for the corresponding amplitude damping channel, and apply it at loss probabilities $\gamma\leq\delta$. We compare its performance with a single error-set Petz map constructed from a photon-loss budget. For that comparison we use the bosonic fock space binomial codes of \cite{michael2016new}, whose photon-loss error correcting are well-known, which allows us to examine the differences between non-attenuated photon-loss errors and random photon-loss errors.

The single-mode binomial construction of~\cite[Sec.~II, Eq.~(7)]{michael2016new} has orthonormal basis
\begin{equation}\label{eq:NEbinomial}
 \ket{c_i}=2^{-(n-1)/2}\sum_{\substack{0\leq r\leq n\\r\equiv i\ ({\rm mod}\ 2)}}\sqrt{\binom nr}\ket{gr},
 \qquad i\in[2],
\end{equation}
where $\ket m$ is the Fock state with $m$ photons, $g$ is the spacing between occupied levels, and $n\geq2$. We denote this code by $Q_{g,n}$. Its largest occupied photon number is $D=gn$, and every encoded state has expected photon number  $\overline N=gn/2$. The code exactly corrects up to $L\leq\min\{n-1,g-1\}$ photon-loss errors~\cite[Sec.~II and Appendix~B]{michael2016new}, represented by the operators $a_\ell, \ell=1,\dots,L$, where
\[
 a_\ell\ket m=
 \begin{cases}
 \binom{m}{l}^{\frac{1}{2}} \ket{m-\ell},&m\geq\ell,\\
 0,&m<\ell.
 \end{cases}
\]
In our numerical examples below we use $n=32$ and $g=32,64,128$, so $\overline N=512,1024,2048$. 

For a loss probability $\gamma$, the bosonic amplitude damping channel on states supported on $\ppp{\ket0,\ldots,\ket D}$ is
\begin{equation}\label{eq:NEAD}
 \calN_\gamma(X)=\sum_{\ell=0}^{D}A_\ell(\gamma)X A_\ell(\gamma)^\dag,
 \qquad A_\ell(\gamma)\ket m=
 \sqrt{\binom m\ell}\,\gamma^{\ell/2}(1-\gamma)^{(m-\ell)/2}\ket{m-\ell},
\end{equation}
with $A_\ell(\gamma)\ket m=0$ for $m<\ell$. In addition to photon loss, these operators include the photon-number-dependent attenuation factor $\gamma^{\ell/2}(1-\gamma)^{(m-\ell)/2}$. We fix $\delta=0.02$ and write $\calR_\delta$ for the Petz map constructed from the full channel $\calN_\delta$ and the code. This recovery is held fixed as the actual loss probability $\gamma$ varies between $0$ and $\delta$.

We now consider the {\em error-set  Petz map} constructed as follows: let $\calH_D^{\mathrm{even}}=\Span\ppp{\ket0,\ket2,\ldots,\ket D}$, which applies to  all the codes  used here once we choose $D$ to match tbeir construction. We choose a photon loss budget $t<D$ such that, with hogh probability, $\calN_{\gamma}$ introduces at most $t$ photon losses. We then find an error set $\calE_{t,D}$ that controls all CP maps $\calN_{\gamma,\leq t}$ obtained by truncating the sum defining the full amplitude damping channel \eqref{eq:NEAD} for any loss probability $\gamma\in (0,1)$. By Remark~\ref{rem:AVpetzCP}, the recovery error of the error-set Petz map is bounded by the KL--Hellinger distance and the probability of lossing more than $t$ photons. The details are as follows.

 For a loss budget $t<D$, consider the error set
\begin{equation}\label{eq:NEerrorSet}
 \begin{aligned}
 \calE_{D,t}&=\ppp{E_{\ell,k}:0\leq\ell\leq t,\ \ell\leq k\leq D,\ k-\ell\text{ even}},\\
 E_{\ell,k}\ket m&=
 \sqrt{\frac{\binom m\ell}{\binom Dk}}\binom{(D-m)/2}{(k-\ell)/2}\ket{m-\ell},
 \qquad \ket m\in\calH_D^{\mathrm{even}}.
 \end{aligned}
\end{equation}
This construction extends the normalized photon-loss error set of~\cite[Example~3]{Elimelech2026Theory} from constant excitation to the maximal excitation $D$, restricted to the even count of photons. The parameter $\ell$ is the number of photons lost, and it is limited to $\ell\leq t$. We include a formal derivation of these error operations in  Appendix~\ref{app:binomialAD}.

For every $\gamma\in[0,1]$, consider the truncated CP map
\[
 \calN_{\gamma,\leq t}(X)=\sum_{\ell=0}^{t}A_\ell(\gamma)X A_\ell(\gamma)^\dag,
 \qquad X\in L(\calH_D^{\mathrm{even}}).
\]
Each Kraus operator $A_\ell(\gamma)$ is a linear combination of the operators $E_{\ell,k}$. Its coefficients are square roots of binomial probabilities, as given in~\eqref{eq:NEcontrolIdentity}. The squared coefficients for a fixed $\ell$ sum to at most one, and different loss counts use disjoint groups of error operators. Consequently, the full Kraus coefficient matrix has operator norm at most one. Lemma~\ref{lem:NEADcontrol} in Appendix~\ref{app:binomialAD} proves this contraction bound. Thus, for every loss probability, $\calN_{\gamma,\leq t}$ is an $\calE_{D,t}$-controlled CP map in the sense that it satisfies the coefficient condition of Definition~\ref{def:errorSetAQEC}. It is generally trace non-increasing rather than trace preserving, so we do not identify it with a channel in $\mathscr{N}(\calE_{D,t})$. We also note that the full amplitude damping channel $\calN_{\gamma}$ is not controlled by this set unless $t=D$.

We fix $t=3g/4$, giving $(D,t)=(1024,24),(2048,48),(4096,96)$, and construct $\calR_{\calE_{D,t},Q}^{\s{Petz}}$. This recovery is also held fixed throughout the experiment. With this choice, we have $\delta \overline N< t$, which implies that the probability of more than $t$ losses (see Eq.~\eqref{eq:NEtail}) is small for all $\gamma\in[0,\delta]$.  For the three codes, this probability is below $5\times10^{-4}$ throughout this interval. For the first of the three codes, the value $t=24$ lies below the exact correction capability of $31$ photon losses, whereas for the remaining two codes, the budgets $t=48,96$ exceed this capability.  

For our choice of the system parameters, the numerical experiments provide strong evidence of the capabilities of the error-set Petz map, one may wish to go beyond the empirical evidence. To this end, we bound the average-case error of Petz recovery against the full amplitude damping channel $\calN_{\gamma}$, using its guarantee on the truncated CP map and the probability of more than $t$ photon losses. Let $P$ be the code projector and write
\begin{equation}\label{eq:NEtail}
 p_Q(t,\gamma):=\tr\p{\calN_{\gamma,>t}(P/2)}
 =2^{-n}\sum_{r=0}^{n}\binom nr\Pr\ppp{\operatorname{Bin}(gr,\gamma)>t},
\end{equation}
where $\calN_{\gamma,>t}=\calN_\gamma-\calN_{\gamma,\leq t}$. This is the probability of more than $t$ losses for the maximally mixed code state $P/2$. A straightforward calculation shows that for a fixed recovery, the channel fidelity under the full AD channel is the sum of the contributions from $\calN_{\gamma,\leq t}$ and $\calN_{\gamma,>t}$. The latter contribution lies between zero and $p_Q(t,\gamma)$. Applying Eq.~\eqref{eq:AVpetzCP} to the truncated map gives
\begin{equation}\label{eq:NEretainedPetzGuarantee}
 0\leq\tr\p{\calN_{\gamma,\leq t}(P/2)}
 -\calF_{\mathrm{ch}}\p{\calR_{\calE_{D,t},Q}^{\s{Petz}}\circ\calN_{\gamma,\leq t}|_{L(Q)},I_{L(Q)}}
 \leq\zeta_{\s{H}}(\calE_{D,t},Q)^2.
\end{equation}
Since $\tr\p{\calN_{\gamma,\leq t}(P/2)}=1-p_Q(t,\gamma)$ the upper-tail CP map $\calN_{\gamma,>t}$ contributes nonnegative fidelity, we obtain
\begin{equation}\label{eq:NEADPetzGuarantee}
 d_{\mathrm{ch}}\p{\calR_{\calE_{D,t},Q}^{\s{Petz}}\circ\calN_\gamma|_{L(Q)},I_{L(Q)}}
 \leq\sqrt{\zeta_{\s{H}}(\calE_{D,t},Q)^2+p_Q(t,\gamma)}.
\end{equation}

\begin{figure}[t]
\begin{tikzpicture}
\node[anchor=south west,inner sep=0] (img) at (0,0)
 {\includegraphics[width=.80\linewidth]{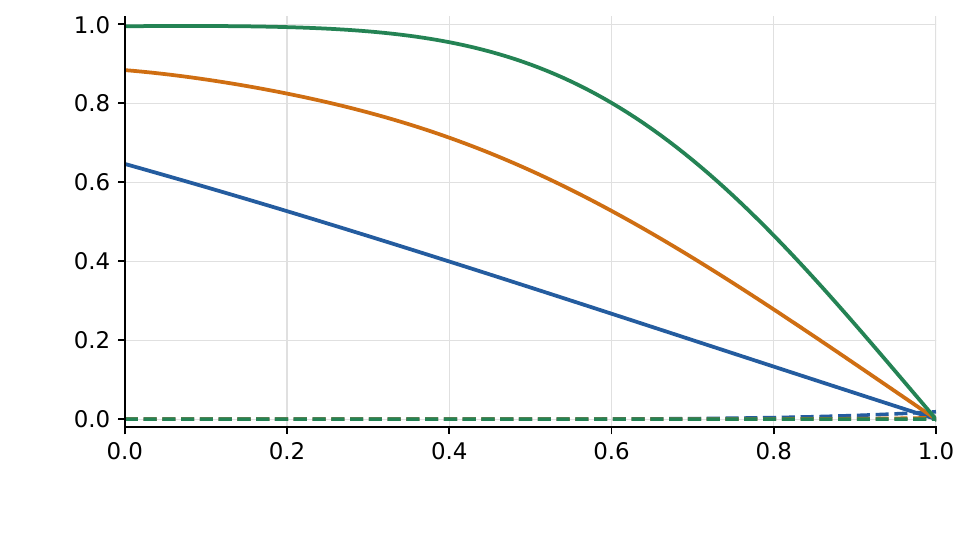}};
\begin{scope}[x={(img.south east)},y={(img.north west)}]
\node[blue,font=\small] at (.56,.035) {Actual loss probability relative to design, $\gamma/\delta$};
\node[blue,rotate=90,font=\small] at (.025,.565) {Average-case error, $d_{\mathrm{ch}}$};
\end{scope}
\end{tikzpicture}
\par\smallskip
{\small\begin{tabular}{@{}ll@{\hspace*{.3in}}ll@{\hspace*{.3in}}ll@{}}
\tikz[baseline=-.5ex]{\draw[NEpetzblue,line width=.9pt] (0,0)--(.65,0);}
& $\overline N=512$, $t=24$
&\tikz[baseline=-.5ex]{\draw[NEpetzorange,line width=.9pt] (0,0)--(.65,0);}
& $\overline N=1024$, $t=48$
&\tikz[baseline=-.5ex]{\draw[NEpetzgreen,line width=.9pt] (0,0)--(.65,0);}
& $\overline N=2048$, $t=96$
\end{tabular}
\\[.1in]
Solid curves: channel-adapted Petz map constructed at $\delta$\\
 Dashed lines: Error-set Petz map constructed from $\calE_{D,t}$ 
}
\caption{Average-case recovery for the binomial codes~\eqref{eq:NEbinomial} under the full amplitude damping channel, with $n=32$, $g=32,64,128$, and design probability $\delta=0.02$. The horizontal axis gives the actual loss probability relative to the design probability, and the vertical axis gives the average case distance $d_{\mathrm{ch}}=\sqrt{1-\calF_{\mathrm{ch}}}$. For each code length, the curves represent the average-case performance for the channel-adapted Petz map designed for loss probability $\delta$ (dashed) and the error-set Petz map (solid), designed for a maximal loss of $1.171\delta$ fraction of the total number of photons throughout the plotted interval.}
\label{fig:NEmismatch}
\end{figure}

Figure~\ref{fig:NEmismatch} shows that the channel-adapted Petz map performs well at its design probability, with average-case error at most $2.30\times10^{-4}$ for all three codes, but  performs poorly when the loss probability is smaller.  In contrast, the fixed error-set performs well through the entire range $\gamma<\delta$. We note that at the edge of the range,  where $\gamma=\delta$, the error-set Petz map errors are $0.0189$, $0.00362$, and $0.000418$, for $\Bar{N}=512, 1024 $ and $2048$ respectively. These errors are larger than those of the corresponding channel-adapted Petz maps designed for error probability $\delta$, but the gap decreases as $\overline N$ increases, showing that the error-set Petz map is typically advantageous, although truncation comes with its price.

These results show that designing the Petz map for a larger loss probability does not always yield a recovery that performs well at smaller loss probabilities. In contrast, the error-set construction gives a single recovery with small error throughout the tested interval for each code, and it is therefore robust to uncertainty in the photon loss probability.

\subsubsection*{Importance of error-set choice} 
In the following experiment we show that the performance of the Petz map depends on the correct choice of the error set used to construct it, and that a misaligned choice can yield poor results. We examine the Petz map constructed from the photon-loss operators $\calE_{t}:=\ppp{a_\ell:0\leq\ell\leq t}$, using the same codes and loss budgets. This choice is motivated by the exact photon-loss correction capability of the binomial code. However, the AD Kraus operators also contain photon-number-dependent attenuation. Even the no-loss operator $A_0(\gamma)$ is not in the span of these photon-loss operators for $0<\gamma<1$, since $a_0=I$ while $A_0(\gamma)$ attenuates different occupied levels by different factors. Thus the truncated AD map is not controlled by this error set.

For each code, we construct this Petz map once and complete it outside the span of its error images by preparing $\ket{c_0}\bra{c_0}$, as for the $\calE_{t,D}$-based Petz recovery used above. We calculate its channel fidelity under the full AD channel for each $\gamma$. Figure~\ref{fig:NEbareLossPetz} shows that this choice gives substantially larger errors than the controlling-set recovery in Figure~\ref{fig:NEmismatch}. For example, for $\overline N=2048$ at $\gamma=0.1\delta$, its error is approximately $0.288$, although the probability of more than $t$ losses is below $10^{-75}$. The negligible loss tail shows that the large recovery error cannot be explained by losses above the budget $t$. Thus, constructing the Petz map from photon-loss operators alone can give poor AD performance, even when almost all losses lie within the chosen budget.

\begin{figure}[t]
\definecolor{NEbareblue}{HTML}{245C9F}
\definecolor{NEbareorange}{HTML}{CF6E12}
\definecolor{NEbaregreen}{HTML}{248354}
\begin{tikzpicture}
\node[anchor=south west,inner sep=0] (img) at (0,0)
 {\includegraphics[width=.85\linewidth]{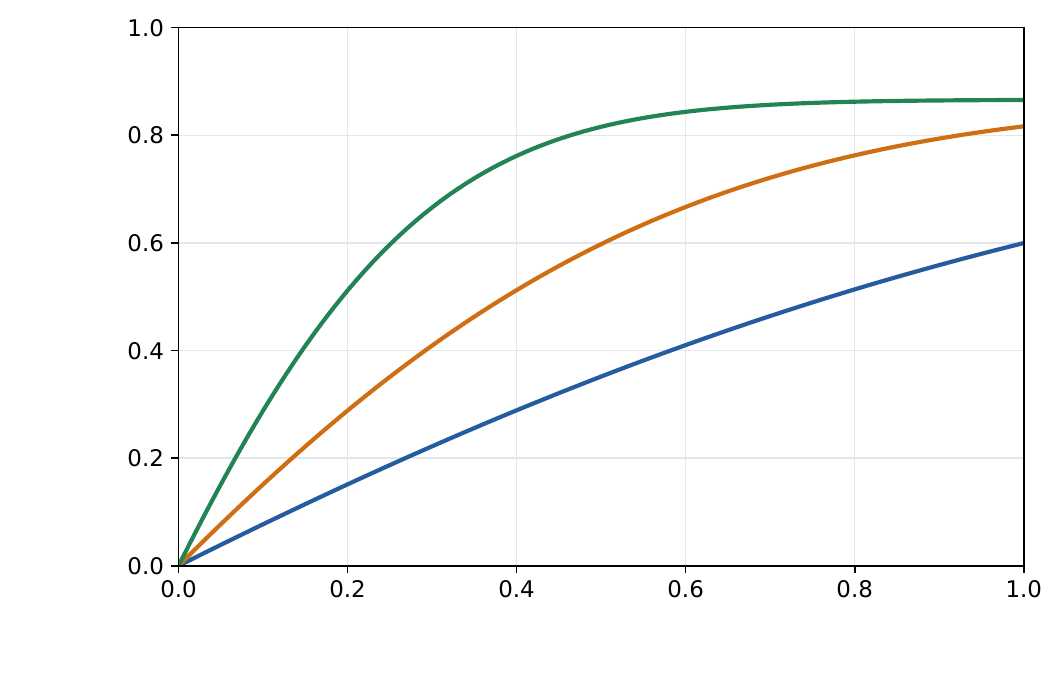}};
\begin{scope}[x={(img.south east)},y={(img.north west)}]
\node[blue,font=\small] at (.57,.035)
 {Actual loss probability relative to design, $\gamma/\delta$};
\node[blue,rotate=90,font=\small] at (.025,.57)
 {Average-case error, $d_{\mathrm{ch}}$};
\end{scope}
\end{tikzpicture}
\par\smallskip
{\small\begin{tabular}{@{}ll@{\hspace*{.3in}}ll@{\hspace*{.3in}}ll@{}}
\tikz[baseline=-.5ex]{\draw[NEbareblue,line width=.9pt] (0,0)--(.65,0);}
& $\overline N=512$, $t=24$&
\tikz[baseline=-.5ex]{\draw[NEbareorange,line width=.9pt] (0,0)--(.65,0);}
& $\overline N=1024$, $t=48$&
\tikz[baseline=-.5ex]{\draw[NEbaregreen,line width=.9pt] (0,0)--(.65,0);}
& $\overline N=2048$, $t=96$
\end{tabular}}
\caption{Average-case errors of the Petz maps constructed from
$\ppp{a_\ell:0\leq\ell\leq t}$ for the binomial codes~\eqref{eq:NEbinomial},
with $n=32$, $g=32,64,128$, and $t=3g/4$. Both axes are linear, and
$\delta=0.02$ is the same reference probability as above. Each curve
uses a fixed recovery, completed by preparing $\ket{c_0}$ outside its
support, and gives its achieved distance $d_{\mathrm{ch}}=
\sqrt{1-\calF_{\mathrm{ch}}}$ under the full AD channel.
}\label{fig:NEbareLossPetz}
\end{figure}

\subsection{Universal recovery under Pauli error sets}\label{sec:NEPauli}
In this section, we examine the error-set Petz map guarantees of Theorems~\ref{th:AVpetz} and \ref{th:WCpetz} against the actual optimal approximation error associated with Pauli error-sets and two well-known quantum codes.  We consider two examples: square Shor codes under a weighted Pauli error set, and Haar-random codes under the complete sets of Pauli errors of weight at most one or two. The first example examines the average-case guarantee of Theorem~\ref{th:AVpetz}, while the second examines the worst-case guarantee of Theorem~\ref{th:WCpetz}. In both examples, the Petz recovery map is constructed from the code and error set. 

\subsubsection*{Square Shor codes under stochastic Pauli noise}

Let $L\geq3$ be odd and set $N=L^2$. The $\qc{L^2,1,L}$ square Shor code $Q_L$ is the two-dimensional subspace spanned by
\begin{equation}\label{eq:NEshorCode}
 \ket{c_i}=\p{\frac{\ket{0}^{\otimes L}+(-1)^i\ket{1}^{\otimes L}}{\sqrt2}}^{\otimes L},
 \qquad i\in[2].
\end{equation}
We consider the error set formed of Pauli operators of every weight. Write $\calP_N=\ppp{I,X,Y,Z}^{\otimes N}$, with one Hermitian representative for each Pauli string, and let $\wt(W)$ denote the number of its nonidentity factors. For $0<p<1$, set $a=p/3$ and define
\begin{equation}\label{eq:NEshorErrorSet}
 \calE_p=\ppp{a^{\wt(W)/2}W:W\in\calP_N}
 =\ppp{I,\sqrt aX,\sqrt aY,\sqrt aZ}^{\otimes N}.
\end{equation}
This error set, considered in~\cite[Sec.~6.1]{Elimelech2026Theory}, controls Pauli channels assigning probability at most $a^{\wt(W)}$ to each $W\in\calP_N$. This includes independent Pauli noise with each nonidentity single-qubit probability at most $a$, as well as correlated Pauli channels satisfying these inequalities. It also contains arbitrary channels acting on a random set of qubits when the probability of touching any prescribed $k$ qubits is at most $a^k$, as in the spatial local-stochastic noise models used in fault-tolerance analyses~\cite[Sec.~III.D, Definition~1]{bombin2016resilience}. 

For each code and $p$, we construct the associated error-set Petz map, $\calR_{\calE_p,Q_L}^{\s{Petz}}$. Our goal it to compare the largest average-case error of this map among all controlled channels, defined as
\[\varepsilon_{\s{av}}^{\s{Petz}}(p):= \sup_{\calN\in \mathscr{N}(\calE_p)}d_{\mathrm{ch}}\p{\calR_{\calE_p,Q_L}^{\s{Petz}}\circ\calN|_{L(Q_L)},I_{L(Q_L)}},\]
to the optimal average-case decoding error of the code atained by any universal recovery map, $\varepsilon_{\s{dec}}^{\s{av}}(\calE_p,Q_L)$. To this end, we fix a specific  $\calE_p$-controlled stochastic Pauli channel $\calN_*$ that minimizes the channel fidelity of this specific error-set Petz map. Writing the channel in the generic form $\sum_{P\in\calP_N}\pi_P PXP^\dag$, one can show that the fidelity of the Petz map is a linear function of the weights, and the $\calE_p$-controllability condition imposes the constraint $\pi_P\leq a^{\wt(P)}$, so the weights for the channel that minimizes fidelity are found by solving a linear program. Combined with Theorem~\ref{th:AVpetz} we obtain the bound  
\begin{equation}
    d_{\mathrm{ch}}\p{\calR_{\calE_p,Q_L}^{\s{Petz}}\circ\calN|_{L(Q_L)},I_{L(Q_L)}}\leq \varepsilon_p^{\s{Petz}}\leq \zeta_{\calH}(\calE_p,Q_L).\label{eq:PetzStochasticBest}
\end{equation}

On the other hand, the same channel provides a lower bound for the optimal decoding error by optimizing the channel fidelity over all possible recovery map. Indeed, Theorem~\ref{th:AvcaseOptimalDecoding} gives
\begin{equation}\label{eq:NEshorBracket}
 \inf_{\calR}d_{\mathrm{ch}}\p{\calR\circ\calN_*|_{L(Q_L)},I_{L(Q_L)}}
 \leq \varepsilon_{\s{opt}}^{\s{av}}(\calE_p,Q_L)=\varepsilon_{\s{dec}}^{\s{av}}(\calE_p,Q_L)
 \leq\zeta_{\s{H}}(\calE_p,Q_L),
\end{equation}
where the infimum is taken over all CPTP recovery maps. To calculate the optimal average-case error for $\calN_{*}$ in the lower bound in \eqref{eq:NEshorBracket}, we calculate the channel fidelity obtained by the optimal recovery map for stochastic Pauli channels and stabilizer codes presented in~\cite[Theorem~3]{fletcher2008structured}. As it turns out, the performance of the optimal recovery for this channel is very to the the error-set Petz recovery map. 

We evaluate the lower and upper bounds in \eqref{eq:PetzStochasticBest} and \eqref{eq:NEshorBracket} for $L=5,7,9$ at $65$ values of $p$ between $0.001$ and $0.03$, and plot the obtained values in Figure~\ref{fig:NEshorPetz}. The results gives rise to two conclusions. First, the KL--Hellinger bound closely captures the error-set performance of the Petz map: its ratio to the error attained on $\calN_*$ is at most $1.00102$ throughout the tested range. Additionally, the ratio between the Knill-Laflamme bound and the optimal error for $\calN_*$ is at most $1.175$. This suggest the second conclusion: the average-case performance guarantee given by the KL--Hellinger distance can be nearly attained and the error-set Petz map can be nearly as optimal as a universal average-case decoder.

\begin{figure}[t]
\begin{tikzpicture}
\node[anchor=south west,inner sep=0] (img) at (0,0)
 {\includegraphics[width=.99\linewidth]{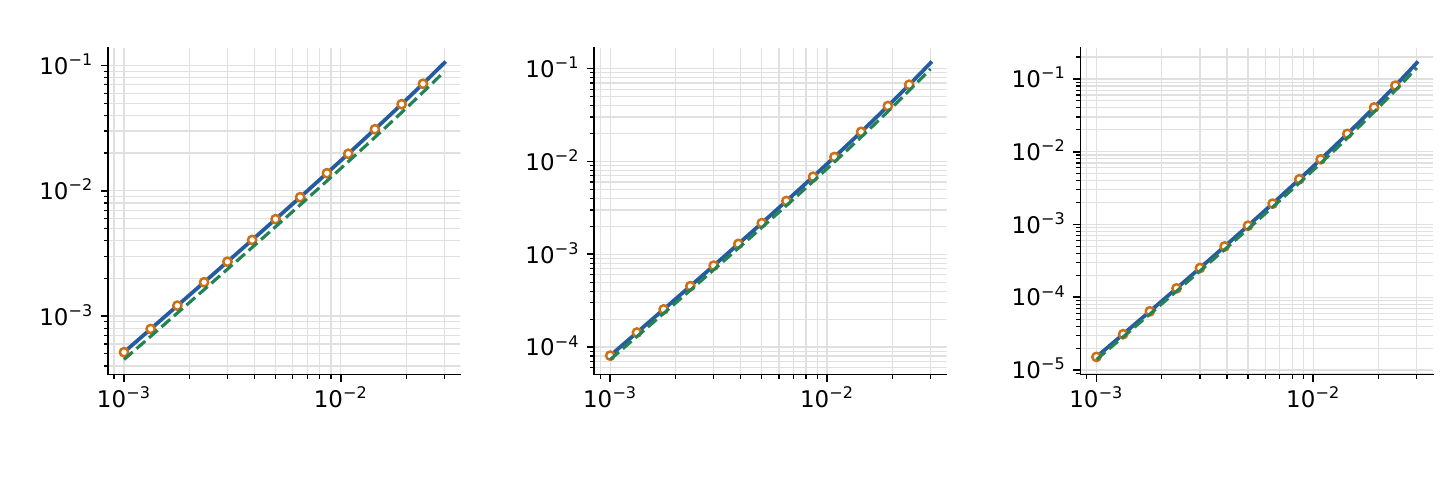}};
\begin{scope}[x={(img.south east)},y={(img.north west)}]
\node[blue,font=\small] at (.215,.975) {$\qc{25,1,5}$};
\node[blue,font=\small] at (.537,.975) {$\qc{49,1,7}$};
\node[blue,font=\small] at (.858,.975) {$\qc{81,1,9}$};
\node[blue,font=\small] at (.535,.035) {Error-set parameter, $p$};
\node[blue,rotate=90,font=\small] at (.012,.55) {Average-case error, $d_{\mathrm{ch}}$};
\end{scope}
\end{tikzpicture}
\par\smallskip
{\small\begin{tabular}{@{}ll@{}}
\tikz[baseline=-.5ex]{\draw[NEpetzblue,line width=.9pt] (0,0)--(.65,0);}
& Upper bound: $\zeta_{\s{H}}(\calE_p,Q_L)$\\
\tikz[baseline=-.5ex]{
  \draw[NEpetzorange,line width=.9pt] (0,0)--(.65,0);
  \draw[draw=NEpetzorange,fill=white,line width=.9pt]
    (.325,0) circle[radius=2pt];
}
& Error-set Petz recovery for $\calN_*$\\
\tikz[baseline=-.5ex]{\draw[NEpetzgreen,dashed,line width=.9pt] (0,0)--(.65,0);}
& Lower bound: optimal recovery for $\calN_*$
\end{tabular}}
\caption{Average-case universal recovery for the square Shor codes with $L=5,7,9$ and the error set~\eqref{eq:NEshorErrorSet}. Both axes are logarithmic. The upper and lower curves bound the optimal universal average-case decoding error in~\eqref{eq:NEshorBracket}. The markers give the error attained by the error-set Petz map on the selected controlled channel and nearly coincide with the upper curve. At each $p$, the recovery is constructed from $\calE_p$.}
\label{fig:NEshorPetz}
\end{figure}

\subsubsection*{Haar codes under limited-weight Pauli errors}
Haar-random codes provide a natural example for studying approximate correction of general errors. Their AQEC performance on quantum channels was studied in \cite{klesse2007approximate} and more recently in \cite{venkatesh2026thermalization}, and they were also shown to approximately attain the quantum Hamming and Singleton bounds \cite{ma2025haar,sommers2026spectral}. We use this ensemble to examine how closely the worst-case guarantee of Theorem~\ref{th:WCpetz} captures the performance of the error-set Petz map at finite lengths, and how this comparison depends on the code dimension.

For each length $N$ and dimension $K=2,8$, we sample a $2^N\times K$ matrix with independent complex Gaussian entries and use QR decomposition to make the columns form an orthonormal set. The resulting orthonormal vectors $\ket{c_i}\in\calH_2^{\otimes N}$, $i\in[K]$ span a $K$-dimensional quantum code $Q$. We consider $N=12,14,\ldots,24$ and sample one code for each pair $(N,K)$. Thus, the two choices of $K$ encode one and three logical qubits, respectively. Each realization is held fixed throughout the experiment, without resampling or optimizing the basis codewords.

We consider the complete limited-weight Pauli error set
\begin{equation}\label{eq:NElimitedPauli}
 \calE_t=\ppp{W\in\calP_N:\wt(W)\leq t},
 \qquad S=|\calE_t|=\sum_{j=0}^{t}3^j\binom Nj,
\end{equation}
with no rescaling of its operators, for $t=1$ and $t=2$. Every channel whose Kraus operators lie in $\Span(\calE_t)$ is strongly $\calE_t$-controlled, as discussed before Theorem~\ref{th:WCpetz}. The Petz map worst-case guarantee of Theorem~\ref{th:WCpetz} therefore applies to general Kraus combinations of the limited-weight errors. For both dimensions, we test $t=1$ at all seven lengths and $t=2$ at $N=12,14,\ldots,20$, using the same code at each common length.

For each code $Q$ and error set $\calE_t$, we are interested in bounding the worst-case optimal error of the associated Petz map,  given as  
\[\varepsilon_{\s{wc}}^{\s{Petz}}(t):= \sup_{\calN\in \mathscr{N}(\calE_t)}d_{\mathrm{wc}}\p{\calR_{\calE_t,Q}^{\s{Petz}}\circ\calN|_{L(Q)},I_{L(Q)}}.\]
For the upper bound, we use Theorem~\ref{th:WCpetz}, where we directly calculate the r.-h.s. of \eqref{eq:PetzWCcond} from the full QEC matrix~\eqref{eq:QEC_matrix} of dimension $K|\calE_t|$. For the lower bound, we select the uniform limited-weight Pauli channel:
\begin{equation}\label{eq:NEHaarChannel}
 \calN_*(X)=\frac1S\sum_{W\in\calE_t}WXW^\dag.
\end{equation}
This channel applies a uniformly chosen Pauli from $\calE_t$. It is strongly $\calE_t$-controlled, as $\calE_t$ is composed of Hilbert--Schmidt orthogonal unitaries. Combining the error of the error-set Petz map on $\calN_*$ with the upper bound of Theorem~\ref{th:WCpetz}, we obtain
\begin{equation}
    d_{\mathrm{wc}}\p{\calR_{\calE_t,Q}^{\s{Petz}}\circ\calN_*|_{L(Q)},I_{L(Q)}}\leq \varepsilon_{\s{wc}}^{\s{Petz}}(t)\leq \frac{1}{\sqrt{2}}\norm{\sqrt{A_{\s{QEC}}}-I_{K}\otimes \frac{1}{\sqrt{K}}\sqrt{\tr_K(A_{\s{QEC}})}}_{\infty}.\label{eq:PetzHaarLower}
\end{equation}

We further evaluate the worst-case error of the error-set Petz map on $\calN_*$ by minimizing the squared entanglement fidelity over all logical input density matrices, followed by the conversion to $d_{\mathrm{wc}}$. For $K=2$, since density matrices are in 1-1 correspondece to points on the Bloch ball, the worst-case Bures distance optimization becomes a convex quadratic minimization on the Bloch ball. For $K=8$, we minimize the corresponding convex quadratic over $8\times8$ density matrices. For this uniform channel, the error-set Petz map also coincides with the channel-adapted Petz map, since the common factor $1/\sqrt S$ in the Kraus operators cancels in the channel-adapted Petz map expression \eqref{eq:PetzFormula}. 

To bound below the optimal worst-case error associated with $\calE_t$ and $Q$, we use the optimal worst-case error for the controlled channel $\calN_*$. Since the recovery is optimized over all CPTP maps, we have
\begin{equation}
    \varepsilon_{\s{opt}}^{\s{wc}}(\calN_*,Q)
 \leq \varepsilon_{\s{opt}}^{\s{wc}}(\calE_t,Q)
 \leq \varepsilon_{\s{dec}}^{\s{wc}}(\calE_t,Q).\label{eq:HaarLowerOpt}
\end{equation}
For any fixed recovery, evaluating the entanglement fidelity for the maximally mixed input gives an upper bound on its worst-case entanglement fidelity. We then bound above the fidelity achievable at this input by any CPTP recovery, using a feasible dual point of the recovery semidefinite program of~\cite{fletcher2007optimum} on the span of the noisy code. Converting this fidelity bound to a distance gives the desired lower bound on $\varepsilon_{\s{opt}}^{\s{wc}}(\calN_*,Q)$, and hence on both error-set quantities in \eqref{eq:HaarLowerOpt} above.

\begin{figure}[!htbp]
\definecolor{NEpetzblue}{HTML}{245C9F}
\definecolor{NEpetzorange}{HTML}{CF6E12}
\definecolor{NEpetzgreen}{HTML}{248354}
\begin{tikzpicture}
\node[anchor=south west,inner sep=0] (img) at (0,0)
 {\includegraphics[width=.99\linewidth]{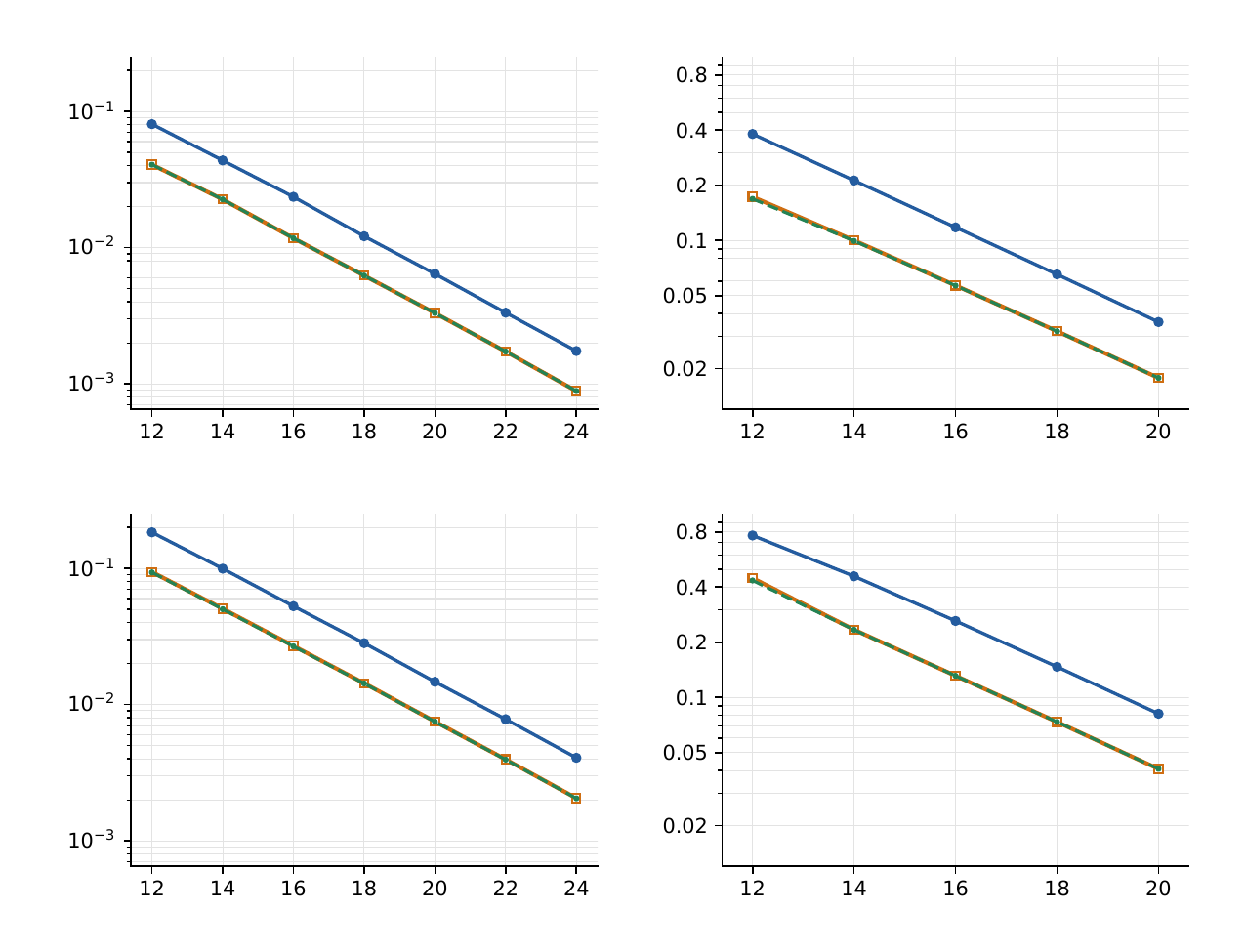}};
\begin{scope}[x={(img.south east)},y={(img.north west)}]
\node[blue,font=\small] at (.293,.974) {$K=2$, $t=1$};
\node[blue,font=\small] at (.768,.974) {$K=2$, $t=2$};
\node[blue,font=\small] at (.293,.494) {$K=8$, $t=1$};
\node[blue,font=\small] at (.768,.494) {$K=8$, $t=2$};
\node[blue,font=\small] at (.545,.019) {Number of physical qubits, $N$};
\node[blue,rotate=90,font=\small] at (.014,.525) {Worst-case error, $d_{\mathrm{wc}}$};
\end{scope}
\end{tikzpicture}
\par\smallskip
{\small\begin{tabular}{@{}ll@{}}
\tikz[baseline=-.5ex]{\draw[NEpetzblue,line width=.9pt] (0,0)--(.65,0);\fill[NEpetzblue] (.325,0) circle (1.7pt);}
& Upper bound: Theorem~\ref{th:WCpetz}\\
\tikz[baseline=-.5ex]{\draw[NEpetzorange,line width=.9pt] (0,0)--(.65,0);\draw[NEpetzorange,fill=white] (.275,-.05) rectangle (.375,.05);}
& Error-set Petz recovery for $\calN_*$\\
\tikz[baseline=-.5ex]{\draw[NEpetzgreen,dashed,line width=.9pt] (0,0)--(.65,0);\fill[NEpetzgreen] (.325,0) circle (1.2pt);}
& Lower bound for arbitrary recovery of $\calN_*$
\end{tabular}}

The parameter pairs $(K,t)$ are, respectively, $\begin{array}{c@{\hspace*{0.05in}}c}(2,1)&(2,2)\\(8,1)&(8,2)\end{array}$. 

\caption{{Worst-case universal recovery for Haar-random codes and the complete Pauli error set~\eqref{eq:NElimitedPauli}. 
The horizontal axes are linear and the vertical axes are logarithmic, with the same ranges within each of the two columns. The upper and lower curves bound $\varepsilon_{\s{dec}}^{\s{wc}}(\calE_t,Q)$. The orange curve gives the error of the error-set Petz map on the uniform channel~\eqref{eq:NEHaarChannel}. One fixed code is used for each pair $(N,K)$, with the same realization sampled from the Haar ensemble. }}
\label{fig:NErandomPetz}
\end{figure}

Figure~\ref{fig:NErandomPetz} shows a comparison between the upper and lower bounds of \eqref{eq:PetzHaarLower} and \eqref{eq:HaarLowerOpt}, for both code dimensions. For $K=2$, the upper bound in~\eqref{eq:PetzHaarLower} is within $2.012$ and $2.259$ times the obtained lower bound on the optimal decoding error \eqref{eq:HaarLowerOpt},  for $t=1$ and $t=2$, respectively. For $K=8$, the corresponding factors are $1.987$ and $1.999$. 
We note that the factor close to two persists when the code dimension increases from $K=2$ to $K=8$ and when the maximal error weight increases from $1$ to $2$.  This leads one to hypothesize that for Haar-random codes, the bound of Theorem~\ref{th:WCpetz} approximates the optimal universal decoding error within a factor close to two with high probability irrespective of the code dimension. 

It can also seen from Figure~\ref{fig:NErandomPetz} that the error-set Petz map is nearly optimal on the selected channel: its error is within a factor of $1.030$ of the lower bound allowing arbitrary recovery, throughout both experiments. At $N=20$ and $t=2$, this factor is below $1.0002$ for both dimensions. Thus, in these instances, allowing a recovery adapted to $\calN_*$ gives little improvement over the error-set Petz map.

\subsection*{Acknowledgments}
\addcontentsline{toc}{section}{Acknowledgments}
  D.E. acknowledges support from the Yad Hanadiv Foundation through the Rothschild Fellowship.   
  V.V.A. acknowledges NSF grant OMA2120757 (QLCI).   
  The research of A.B. was partially supported by NSF grants CIF-2330909 and CIF-2526035.

During the preparation of this manuscript, the authors used Claude (Anthropic) and GPT (OpenAI) as assistive tools for discussion of ideas, preliminary checks, and editing. All mathematical arguments were independently written by the authors, who take full responsibility for the content of the paper. Product names are given solely to identify the tools used; their inclusion does not constitute an endorsement or recommendation by NIST.

\subsection*{Data availability}
This is a purely mathematical work and no data was created or analyzed in this study. All figures can be reproduced directly from the presented equations.

\appendix
\makeatletter
\section*{Appendices}

 \section{Technicalities and auxiliary results}
 \subsection{States and channels in Hilbert spaces: norms and properties}
 Let us briefly present the state and superoperator norms used in this work.
 \begin{definition}[State norms]\label{Def:StateNorms}
Let $\calH$, $\calH'$ be finite-dimensional spaces and $p\in  [1,\infty]$. The $p$-norm of an operator $X:\calH\to \calH'$ is defined as 
\[\norm{X}_p=\p{\sum_{\lambda \in \s{S}(X)}\lambda^p}^{\frac{1}{p}
},\]
where $\s{S}(X)$ is the multiset of singular values of $X$, and for $p=\infty$ (by taking a limit) $\norm{X}_{\infty}$ is the maximal singular value of $X$.
\end{definition} 
\begin{definition}\label{eq:operatorMonotone}
    A real-valued function $f:[0,\infty)\to \R$ is called operator-monotone if for any finite-dimensional Hilbert space $\calH$  and for any positive-semidefinite operators  $A $ and  $B$ such that $A\succeq B$ we have $f(A)-f(B)\geq 0$, where $f$ acts on $X=\sum_{i}\lambda_i\ket{i}\bra{i}$ by $f(X)=\sum_i f(\lambda_i) \ket{i}\bra{i}$ 
    (and $\ppp{\ket{i}}$ are orthonormal vectors).
\end{definition}
\begin{lemma}\label{lem:HolderTypef} \cite[Theorem X.1.1]{bhatia2013matrix}
Let $\calH$ be  a finite-dimensional Hilbert space and let $A,B\in L(\calH)$ be positive-semidefinite operators. For any  operator-monotone function  $f:\R_+\to \R$ with $f(0)=0$: 
\[ \norm{f(A)-f(B)}_{\infty}\leq f(\norm{A-B}_{\infty}).\]
\end{lemma}

 \begin{definition}[Superoperator norms]\label{Def:superOpNorms}
Let $\calH$, $\calH'$ be finite-dimensional spaces, and $\calN:L(\calH)\to L(\calH')$ be a superoperator. 
\begin{enumerate}
    \item For $p\in  [1,\infty]$, the $p$-completely bounded (cb) norm of $\calN$ is defined to be 
    \[\norm{\calN}_{\s{cb},p}=\sup_{M\in \N}\sup_{\norm{X}_p\leq 1}\norm{(\calN\otimes I_{M})(X)}_p,\]
    where $I_{M}$ is the identity on $L(\calH_M)$, and $X$ ranges over operators in $L(\calH\otimes \calH_M)$.
    \item The diamond norm of $\calN$ is defined similarly to the $\norm{\cdot}_{\s{cb},1}$,  with the difference that maximization runs over density operators:
    \[\norm{\calN}_{\diamond}:=\sup_{M\in \N}\sup_{\rho\in D(\calH\otimes\calH_M)}\norm{(\calN\otimes I_{M})(\rho)}_1.
    \]
\end{enumerate}
 \end{definition}
\begin{lemma}[Superoperator norm properties {\rm (\cite[Theorem 3.51]{watrous2018theory})}]\label{lem:superoperatorNorms} If $\calN$ is a Hermitian preserving superoperator $L(\calH)\to L(\calH')$, then $\norm{\calN}_{\diamond}=\norm{\calN}_{\s{cb},1}$.
\end{lemma}

 We shall require the following well-known duality principle of completely bounded norms (see \cite[Theorem 1]{Johnston2009Computing}).
     \begin{lemma}[H\"older duality lemma] \label{lem:Holder duality}
         Let $\calH$, $\calH'$ be finite-dimensional spaces, and $\calN:L(\calH)\to L(\calH')$ be a superoperator. Then 
         \[\norm{\calN}_{\s{cb},1}=\norm{\calN^\dag}_{\s{cb},\infty},\]
         where $\calN^\dag$ is the adjoint of $\calN$ with respect to the Hilbert-Schmidt inner product. 
     \end{lemma}

\begin{lemma}[Partial traces in composite systems property]\label{lem:tracPropComposite}
Let $\calH=\calH_A\otimes \calH_{B}\otimes \calH_C$ be a finite-dimensional system. Let $Y,W\in L(\calH_A\otimes \calH_B)$ and $X\in L(\calH)$  be linear operators. Then
\[\tr_{BC}\p{(Y\otimes I_{C})X(W^\dag\otimes I_{C})}=\tr_{B}\p{Y\tr_{C}(X) W^\dag}.\]
\end{lemma}
\begin{proof}
    The proof follows immediately from the definition of the partial trace: for any $Z\in L(\calH)$ we have $\tr_{BC}(Z)=\tr_B(\tr_C(Z))$. Thus, it is sufficient to show that 
    \[\tr_C\p{(Y\otimes I_{C})X(W^\dag\otimes I_{C})}=Y\tr_{C}(X) W^\dag.\]
    Indeed, let $\ppp{\ket{c_i}}_i$ be an orthonormal basis for $\calH_C$. We have 
    \begin{align*}
        \tr_C\p{(Y\otimes I_{C})X(W^\dag\otimes I_{C})}&=\sum_{i}I_{AB}\otimes  \bra{c_i}  \bigg((Y\otimes I_{C})X(W^\dag\otimes I_{C})\bigg)I_{AB}\otimes\ket{c_i}\\&
        =\sum_{i}(Y\otimes  \bra{c_i})  X (W^\dag\otimes\ket{c_i})\\&
        =Y\p{\sum_{i}(I_{AB}\otimes  \bra{c_i})  X (I_{AB}\otimes\ket{c_i}) }W^\dag\\
        &=Y\tr_C(X) W^\dag.   \qedhere
    \end{align*}
\end{proof}
\begin{lemma}\label{lem:partialTacesSqrts}
    Let $X\in L(\C^{K}\otimes \C^S)$ be a positive semidefinite operator. Then 
    \[\tr_K(X)-\frac{1}{K}\p{\tr_K\sqrt{X}}^2\succeq 0.\]
\end{lemma}
\begin{proof}
    The proof follows a simple variance-like 
        argument. Let $X':=\frac{1}{K}I_{K}\otimes \tr_K(\sqrt{X})$. Note that $X'\succeq 0$ hence it is Hermitian, and therefore so is $\sqrt{X}-X'$. In particular 
    \[(\sqrt{X}-X')^2=(\sqrt{X}-X')(\sqrt{X}-X')^\dag\succeq 0.\]
    Since the partial trace operation is CPTP, we have 
    \begin{align}
        \tr_K\p{(\sqrt{X}-X')^2}\succeq 0. \label{eq:partialTracePositiveSquare}
    \end{align}
    On the other hand, 
\begin{align*}
    \tr_K\p{(\sqrt{X}-X')^2}&=\tr_K\p{X- \frac{1}{K}\sqrt{X}(I_{K}\otimes \tr_K(\sqrt{X})) -\frac{1}{K}(I_{K}\otimes \tr_K(\sqrt{X}))\sqrt{X} +\frac{1}{K^2}(I_{K}\otimes \tr_K(\sqrt{X})^2)}\\
    &=\tr_{K}(X)-\frac{1}{K}\tr_K\p{\sqrt{X}}^2,
\end{align*}
where in the last equality we used the equality 
\[\tr_K\p{\sqrt{X}(I_{K}\otimes \tr_K(\sqrt{X}))}=\tr_K\p{(I_{K}\otimes \tr_K(\sqrt{X}))\sqrt{X}}=\tr_K\p{\sqrt{X}}^2,\]
which follows from Lemma~\ref{lem:tracPropComposite}.
\end{proof}

 \section{Stinespring dilation and the continuity theorem}\label{sec:Stinespring}
 Stinespring theorem~\cite{Stinespring1955Positive} gives a characterization of completely positive superoperators via an auxiliary system and a dilation map.  Let $\calH$ and $\calH'$ be finite-dimensional Hilbert spaces and  let $\calN:L(\calH)\to L(\calH')$ be a CP map. In the Schr\"{o}dinger picture, Stinespring theorem implies that $\calN$ is of the form 
 \begin{equation}
     \calN(\rho)=\tr_{M}(V\rho V^\dag),\label{eq:StinepringSchrod}
 \end{equation}
 where $\calH_M$ is a finite-dimensional reference system and $V:\calH\to \calH'\otimes \calH_M $ is called a dilation operator, which happens to be an isometry when $\calN$ is a quantum channel (namely, it is also trace preserving). In the Heisenberg picture, the same Stinespring dilation operator also characterizes the action of the Hilbert-Schmidt adjoint of $\calN$:
 \begin{equation}
     \calN^{\dag}(\rho)=V^\dag (\rho\otimes I_M) V.\label{eq:StinespringHeis}
 \end{equation}
 Let us define Stinespring's dilations.
  \begin{definition}[Stinespring dilation] \label{def:Stinespring} For a CP map $\calN:L(\calH)\to L(\calH')$,  a Stinespring representation is a pair $(\calH_M,V)$ giving rise to \eqref{eq:StinepringSchrod} and \eqref{eq:StinespringHeis}. 
  A Stinespring dilation $(\calH_M,V)$ is called \emph{minimal} if it satisfies
 \begin{equation}
     \Span\ppp{(X\otimes I_M)V \ket{\psi} ~:~ \ket{\psi}\in \calH, X\in L(\calH')}=\calH'\otimes \calH_M,\label{eq:minStinespring}. 
 \end{equation}
  \end{definition}
 For a CP map with Kraus operators $\calE=\ppp{E_i}_{i\in [S]}$ the corresponding \textit{canonical Stinespring dilation} $V:\calH\to \calH'\otimes \calH_{S}$ is given by
 \[V\ket{\psi}=\sum_{i\in[S]}E_i\ket{\psi}  \otimes \ket{i}, \]
 where $\ppp{\ket{i}}_{i\in [S]}$ is an orthonormal basis for $\calH_{S}$. We refer to \cite{watrous2018theory,wilde2017quantum} for  further information on finite-dimensional Stinespring representations.

 Stinespring dilation continuity theorem for channels \cite[Theorem~1]{Kretschmann2008Information} asserts that two quantum channels are close if and only if they admit close dilation operators. While the continuity theorem was originally stated in the Heisenberg picture, we formulate it in the Schr\"{o}dinger picture, using the duality lemma (Lemma~\ref{lem:Holder duality}).
 \begin{theorem}[Stinespring dilation continuity theorem for channels \cite{Kretschmann2008Information}]\label{th:SteinContChannels}
     Let $\calH$ and $\calH'$ be finite-dimensional Hilbert spaces and   $\calN,\calM:L(\calH)\to L(\calH')$ be  CPTP maps. Then
     \[\inf_{V_{\calN},V_{\calM}}\norm{V_{\calN}-V_{\calM}}_\infty^2\leq \norm{\calN -\calM}_{\s{cb},1}=\norm{\calN^\dag -\calM^\dag}_{\s{cb},\infty}\leq \inf_{V_{\calN},V_{\calM}}2 \norm{V_{\calN}-V_{\calM}}_\infty, \]
     where the minimization runs over all dilation operators of $\calN$ and $\calM$ in a joint reference system.  
     More generally, it was proved that the minimal distance between Stinespring dilations is in fact the Bures distance:
     \begin{equation}
        d_{\mathrm{wc}}(\calN,\calM)= \frac{1}{\sqrt{2} }\inf_{V_{\calN},V_{\calM}}\norm{V_{\calN}-V_{\calM}}_\infty.\label{eq:BuresDistanceStinespring}
     \end{equation}
 \end{theorem}
 
In \cite{Kretschmann2008Continuity}, the authors extended the continuity theorem from the finite-dimensional channel scenario to general CP maps on C$^*$-algebras. If $\mathscr{A}$ is a C$^*$-algebra, $\calH$ is a Hilbert space, and $\calN:\mathscr{A}\to  L(\calH)$ is a CP map, then it admits a Stinespring representation: $(\pi,\calK, V)$  where $\calK$ is a Hilbert space, $\pi:\mathscr{A}\to L(\calK)$ is a $^*$-homomorphism and $V:\calH\to \calK$ is an operator such that for any $X\in \mathscr{A}$
\[\calN(X)=V^{\dag}\pi(X)V.\]
 \begin{theorem}[Stinespring dilation continuity theorem for CP maps]\label{th:SteinContCP}
    Suppose that $\mathscr{A}$ is a C$^*$-algebra, $\calH'$ is a Hilbert space, and $\calN,\calM:\mathscr{A}\to  L(\calH')$ are CP maps. Then
     \begin{equation}
         \frac{\norm{\calN -\calM}_{\s{cb},\infty}}{\sqrt{\norm{\calN}_{\s{cb},\infty}}+\sqrt{\norm{\calM}_{\s{cb},\infty}}} \leq \inf_{V_{\calN},V_{\calM}}\norm{V_{\calN}-V_{\calM}}_\infty\leq \sqrt{\norm{\calN -\calM}_{\s{cb},\infty}},\label{eq:StinespringContCP}
     \end{equation}
     where the minimization runs over all dilations $(\pi,\calK,V_{\calM}),(\pi,\calK,V_{\calN})$  of $\calM$ and $\calN$ via the same $^*$-homomorphism. Furthermore, there exists a Hilbert space $\calK$, a $^*$-homomorphism $\pi$ and dilation operators $V_{\calN}$,$V_{\calM}$ that realize the minimization of \eqref{eq:StinespringContCP}.
 \end{theorem}

 Let us translate the CP map version of Stinespring continuity to our finite-dimensional setting. Let $\calH$ and $\calH'$ be finite-dimensional Hilbert spaces and let $\calN,\calM:L(\calH)\to L(\calH')$ be CP maps. Note that also $\calN^\dag,\calM^\dag: L(\calH')\to L(\calH)$ are CP maps (see \cite[Proposition 2.18]{watrous2018theory}).  Applying Theorem~\ref{th:SteinContCP} for $\calN^\dag$ and $\calM^\dag$ with $\mathscr{A}=L(\calH')$, we conclude that there exists a Hilbert space $\calK$, a $^*$-homomorphism $\pi:L(\calH')\to L(\calK)$ and dilation operators $\hat{V}_{\calN},\hat{V}_{\calM}:\calH\to \calK$ such that
 \[\calN^{\dag}(X)=\hat{V}_{\calN}^{\dag}\pi(X)\hat{V}_{\calN}, \quad  \calM^\dag(X) = \hat{V}_{\calM}^{\dag} \pi(X)\hat{V}_{\calM},\]
 and
 \begin{equation}
     \frac{\norm{\calN^\dag -\calM^\dag}_{\s{cb},\infty}}{\sqrt{\norm{\calN^\dag}_{\s{cb},\infty}}+\sqrt{\norm{\calM^\dag}_{\s{cb},\infty}}} \leq \inf_{V_{\calN},V_{\calM}}\norm{V_{\calN}-V_{\calM}}_\infty=\norm{\hat{V}_{\calN}-\hat{V}_{\calM}}_\infty\leq \sqrt{\norm{\calN^\dag -\calM^\dag}_{\s{cb},\infty}}.\label{eq:ContinuityCP}
 \end{equation}
 We recall that any $*$-homomorphism $\pi: L(\calH')\to L(\calK)$ from a finite-dimensional full operator algebra is unitarily equivalent to a $*$-homomorphism $\calI^M :L(\calH')\to L(\calH'\otimes \calH_M)$, $\calH_M=\C^M$ of  the form $\calI^M(X)=X\otimes I_M$ (see \cite[Corollary III.1.2]{Davidson1996CStar}). That is, there exists a unitary $U:\calH'\otimes \calH_M\to \calK$ such that $\pi$ factors as
 \[\pi(X)=U(X\otimes I_M)U^\dag.\]
 In particular, without loss of generality, we can restrict ourselves to representations of the form $(\calI^M,\calH'\otimes \calH_M,V)$.

 Stinespring representation is clearly not unique. However, it is well known that any two Stinespring representations of a CP map  $(\pi,\calK,V_{\calN})$, and $(\hat{\pi},\hat{K},\hat{V}_{\calN})$ are related by a partial isometry (see \cite[Eq. (2)]{Kretschmann2008Continuity}). Namely, there exists a partial isometry $W:\calK\to\hat{K}$ such that 
 \begin{equation}
     WV_{\calN}=\hat{V}_{\calN}, \quad W^\dag \hat{V}_{\calN}=V_{\calN}, \quad  W\pi(X)=\hat{\pi}(X)W.\label{eq:StineEquiv}
 \end{equation}
 In the finite-dimensional case where $(\pi,\calK)=(\calI^M,\calH'\otimes \calH_M)$ and $(\hat{\pi},\hat{\calK})=(\calI^{\hat{M}},\calH'\otimes \calH_{\hat{M}})$, \eqref{eq:StineEquiv} implies that $W$ is of the form $W=I_{\calH'}\otimes J$, where $J:\calH_M\to \calH_{\hat{M}} $  is a partial isometry. Furthermore, if $V_{\calN}$ is minimal in the sense of \eqref{eq:minStinespring}, then $J$ is an actual isometry. We refer to \cite[Section~III.B]{Kretschmann2008Information} for a complete derivation of the above statement. 

 We now combine the above conclusion with \eqref{eq:StinepringSchrod}, \eqref{eq:StinespringHeis}, \eqref{eq:ContinuityCP}, \eqref{eq:StineEquiv}, and Lemma~\ref{lem:Holder duality} to obtain the following finite-dimensional version of Stinespring continuity theorem for CP maps: 
 \begin{lemma}[Stinespring dilation continuity for finite-dimensional CP maps]\label{lem:StinespringCont}
     Let $\calH$ and $\calH'$ be finite-dimensional Hilbert spaces and $\calN,\calM:L(\calH)\to L(\calH')$ be CP maps with Stinespring dilation operators $V_{\calN}:\calH\to \calH' \otimes \calH_{M_{\calN}}$ and $V_{\calM}:\calH\to \calH' \otimes \calH_{M_{\calM}}$ respectively. Then there exists $M\in \N$ and partial isometries $J_{\calN}:\calH_{M_{\calN}}\to \calH_M$, $J_{\calM}:\calH_{M_{\calM}}\to \calH_M$ such that 
      \begin{equation*}
     \frac{\norm{\calN  -\calM }_{\s{cb},1}}{\sqrt{\norm{\calN }_{\s{cb},1}}+\sqrt{\norm{\calM }_{\s{cb},1 }}} \leq \inf_{\Tilde{V}_{\calN},\Tilde{V}_{\calM}}\norm{\Tilde{V}_{\calN}-\Tilde{V}_{\calM}}_\infty=\norm{(I_{\calH'}\otimes J_{\calN})V_{\calN}-(I_{\calH'}\otimes J_{\calM}){V}_{\calM}}_\infty\leq \sqrt{\norm{\calN  -\calM }_{\s{cb}, 1}}.
 \end{equation*}
 and 
 \begin{align*}
      (I_{\calH'}\otimes J_{\calN}^\dag J_{\calN})V_{\calN}=V_{\calN}, \quad (I_{\calH'}\otimes J_{\calM}^\dag J_{\calM})V_{\calM}=V_{\calM}.
 \end{align*}
 If $V_{\calN}$ (or $V_{\calM}$) is minimal in the sense of \eqref{eq:minStinespring}, then $J_{\calN}$ (or $J_{\calM}$) is an isometry.
 \end{lemma}

\section{Proofs}
\subsection{Proof of Theorem~\ref{th:DecExistsWorst}}\label{app:PoorfBOdecoder}

Denote $K=\dim(Q)$, let $P$ be the projector on the code $Q$ and set $K_i:=E_i P$. We begin by translating the B\'eny-Oreshkov type assumption to the Stinespring dilation framework, which is the centerpiece of our proof.\vspace{0.2cm}

\noindent\textbf{Step 1: Stinespring dilations from B\'eny-Oreshkov superoperators.}
Consider the map $\calM^Q_{\calE}:L(\calH_{\SIN})\to L(\calH_{\SOUT})$ defined by $\calE$: 
\[\calM^Q_{\calE}(X)=\sum_{i=0}^{S-1}K_i X K_i^\dag.\]
Note that $\calM^Q_{\calE}$ is a CP map which in general is not trace preserving. Observe that $\calM_{\calE}^Q$ can also be seen as a map $L(Q)\to L(\calH_{\SOUT})$ when considering the restriction of $P$ to $Q$, which is the inclusion operator $Q\to \calH_{\SIN}$ and $P^\dag:\calH_{\SIN}\to Q$ remains the usual projection on $Q$. We note that $\calM^Q_{\calE}$ has the canonical Stinespring representation
\begin{equation}
    \calM^Q_{\calE}(X)=\tr_{S}\p{V_{\calE} X V_{\calE}^\dag}, \quad V_{\calE}\ket{\psi}=\sum_{i=0}^{S-1}K_i\ket{\psi} \otimes \ket{i}=\sum_{i=0}^{S-1}E_iP\ket{\psi} \otimes \ket{i},\label{eq:VcalEdef}
\end{equation}
where $V_{\calE}:\calH_{\SIN}\to \calH_{\SOUT}\otimes \calH_S$ is the canonical dilation operator
and  $\ppp{\ket{i}}_{i\in [S]}$ is an orthonormal basis for an $S$-dimensional reference system  $\calH_S$. Let $\widehat{\calM}^Q_{\calE}:L(\calH_{\SIN})\to L(\calH_{S})$ be the complementary CP map of $\calM^Q_{\calE}$ defined by tracing out the output space:
\[\widehat{\calM}^Q_{\calE}(X):= \tr_{{\calH_{\SOUT}}}\p{V_{\calE} X V_{\calE}^\dag}=\sum_{i,j=0}^{S-1}\tr(K_j^\dag K_i X) \ket{i}\bra{j}.\]
Note that for a fixed $S\times S$ matrix $\lambda$, $\calB_{\lambda,Q}^{\calE}$ equals the difference  between $\widehat{\calM}^Q_{\calE}$ and the constant map on the code space, which is defined as $\calM_{\lambda}(X):=\tr(XP)\lambda$ and satisfies $\calM_{\lambda}(X)=\tr(X)\lambda$ for $X\in L(Q)$. Indeed
\begin{align*}
    \calB_{\lambda,Q}^{\calE}(X)&=\sum_{i,j=0}^{S-1}\tr\p{(PE_j^\dag E_iP -\lambda_{i,j}P)X}\ket{i}\bra{j}
    \\
    &=\sum_{i,j=0}^{S-1}\tr(K_j^\dag K_i X) \ket{i}\bra{j}-\tr(P X) \sum_{i,j=0}^{S-1}\lambda_{i,j}\ket{i}\bra{j}=\widehat{\calM}^Q_{\calE}(X)-\calM_{\lambda}(X),
\end{align*}
and thus $\norm{\calB_{\lambda,Q}^{\calE}}_{\diamond}$ measures the diamond norm distance $\|\widehat{\calM}^Q_{\calE}-\calM_{\lambda}\|_\diamond$. Now let us choose a positive semidefinite matrix $\lambda$, given by 
\begin{equation}
    \lambda_{i,j}=\frac{1}{K}\tr(PE_j^\dag E_i)=\frac{1}{K}\tr(PE_j^\dag E_iP)=\frac{1}{K}\tr(K_j^\dag K_i).\label{eq:lambdaDef}
\end{equation}
Note that $\calB_{\lambda,Q}^\calE$ is Hermitian preserving because for any Hermitian operator $X$,
\begin{align}
        \nonumber\calB_{\lambda,Q}^{\calE}(X)^\dag&= \sum_{i,j=0}^{S-1}\tr\p{(PE_j^\dag E_iP -\lambda_{i,j}P)X}^*\ket{j}\bra{i}
        \\\nonumber&=\sum_{i,j=0}^{S-1}\tr\p{X^\dag(PE_i^\dag E_jP -\lambda_{i,j}^*P)}\ket{j}\bra{i}\\
        &=\sum_{i,j=0}^{S-1}\tr\p{X(PE_i^\dag E_jP -\lambda_{j,i}P)}\ket{j}\bra{i}= \calB_{\lambda,Q}^{\calE}(X).\label{eq:HermitianPereserving}
    \end{align}
By \cite[Remark 8]{Elimelech2026Theory} and Lemma~\ref{lem:superoperatorNorms}  we have 
\begin{equation}
\norm{\widehat{\calM}^Q_{\calE}-\calM_{\lambda}}_{\s{cb},1}=\norm{\calB_{\lambda,Q}^{\calE}}_{\s{cb},1}=\norm{\calB_{\lambda,Q}^{\calE}}_{\diamond}\leq 2\zeta(\calE,Q).\label{eq:BOoperatorBoundByzeta}
\end{equation}
where the first equality follows since $\norm{\cdot}_{\s{cb},1}$ and $\norm{\cdot}_{\diamond}$ agree for Hermitian preserving operators (see Lemma~\ref{lem:superoperatorNorms}).

Consider the operator $V_{\lambda}:\calH_{\SIN}\to Q \otimes \calH_T\otimes \calH_{T'}$, where $\calH_T$ and $\calH_{T'}$ are two isomorphic copies of the space $\supp(\lambda)\subseteq \calH_S$, defined as 
\begin{equation}
    V_{\lambda}\ket{\psi}=P\ket{\psi}\otimes \ket{\psi_\lambda},\quad \text{where} \quad \ket{\psi_{\lambda}}=\sum_{i\in [T]} \sqrt{p_i}\ket{\lambda_i}\otimes\ket{\lambda'_{i}},\label{eq:VlambdaDef}
\end{equation}
and $\ppp{\ket{\lambda_i}}_{i\in[T]}$ is an orthonormal basis such that $\lambda=\sum_i p_i \ket{\lambda_i}\bra{\lambda_i}$ (which exists by the spectral theorem).
A straightforward calculation shows that $V_{\lambda}$ dilates $\calM_{\lambda}$. That is, for all $X\in L(\calH_{\SIN})$
\[\calM_{\lambda}(X)=\tr_{\calH_{\SIN}\otimes \calH_{T'}}\p{V_{\lambda} X V_{\lambda}^\dag}=\tr(XP)\lambda,\]
and in particular, if $X\in L(Q)$
\[\calM_{\lambda}(X)=\tr_{Q\otimes \calH_{T'}}\p{V_{\lambda} X V_{\lambda}^\dag}=\tr(X)\lambda.
\]

Using the Stinespring dilations continuity theorem (see Lemma~\ref{lem:StinespringCont}, Appendix~\ref{sec:Stinespring}) combined with \eqref{eq:BOoperatorBoundByzeta}, we conclude that there exists a finite-dimensional Hilbert space $\calH_M$ and partial isometries $J_{\calE}:\calH_{\SOUT}\to \calH_M$ and $J_{\lambda}:Q\otimes \calH_{T'}\to \calH_M$ such that 
 \begin{equation}
     \norm{(I_{\calH_S}\otimes J_{\calE})V_{\calE}-(I_{\calH_S}\otimes J_{\lambda}){V}_{\lambda}}_\infty\leq \sqrt{\norm{\widehat{\calM}^Q_{\calE}  -\calM_{\lambda} }_{\s{cb}, 1}}\leq \sqrt{2\zeta(\calE,Q)}
     \label{eq:ContinuityCPfirst}
 \end{equation}
and 
 \begin{align}
      (I_{\calH_{S}}\otimes J_{\calE}^\dag J_{\calE})V_{\calE}=V_{\calE}, \quad (I_{\calH_{S}}\otimes J_{\lambda}^\dag J_{\lambda})V_{\lambda}=V_{\lambda}.\label{eq:commute}
 \end{align}
 Further, denote 
 \[
 A=\ppp{(X\otimes I_{Q\otimes \calH_{T'}})V_{\lambda} \ket{\psi} ~:~ \ket{\psi}\in \calH_{\SIN}, X\in L(\calH_T)}
 \] 
 and note that $\Span(A)=Q\otimes \calH_T\otimes \calH_{T'}$. To see this, observe that for all $i$ and $j$, 
 $X_{i,j}=\ket{\lambda_j}\bra{\lambda_i}$ and all $\ket{\psi}$,
 \begin{equation}
     (X_{i,j}\otimes I_{Q\otimes \calH_{T'}})V_{\lambda} \ket{\psi}= \sqrt{p_i}P\ket{\psi}\otimes \ket{\lambda_j}\otimes \ket{\lambda_i'}.\label{eq:minimalS}
 \end{equation}
 Since $p_i>0$, varying \eqref{eq:minimalS} over all $\ket{\psi}\in Q$ and $i,j\in [T]$ we obtain a spanning set for $Q\otimes \calH_{T}\otimes \calH_{T'}$. In particular, $V_{\lambda}$ is minimal in the sense of Definition~\ref{def:Stinespring} and thus, by Lemma~\ref{lem:StinespringCont}, $J_{\lambda}$ is an isometry. 
 
 We are now ready to define the universal recovery map. \vspace{0.2cm}

\noindent\textbf{Step 2: Defining the recovery map.} We now use $J_{\calE}$ and $J_{\lambda}$ to define the recovery operation $\calR:L(\calH_{\SOUT})\to L(\calH_{\SIN})$. Denote the image of $J_{\lambda}$ by $\calH_{J_{\lambda}}\subseteq \calH_M$ and let $\calH_{J_{\lambda}}^\perp$ be its orthogonal complement. Let $P_{J_{\lambda}}$ and $P_{J_{\lambda}^{\perp}}$ be projections on $\calH_{J_{\lambda}}$ and $\calH_{J_{\lambda}}^\perp$, respectively. Define the operators $R_1$, $R_2$ and $R_3$ as follows:
\begin{align}
    \nonumber&R_1:\calH_{\SOUT}\to Q\otimes \calH_{T'}, && R_1:=J_{\lambda}^{\dag} J_{\calE},\\
    \nonumber&R_{2}:\calH_{\SOUT}\to \calH_{J_{\lambda}^{\perp}}, &&R_2:=P_{J_{\lambda}^\perp} J_{\calE}, \\ &R_3:\calH_{\SOUT}\to \calH_{\SOUT}, && R_3:= I_{\calH_{\SOUT}}-J_{\calE}^{\dag}J_{\calE}.\label{eq:R_iDef}
\end{align}
 We define the recovery map $\calR:L(\calH_{\SOUT} )\to L(Q)$ as 
\begin{equation}
    \calR(X)=\tr_{{T'}}\p{R_1 X R_1^\dag}+\tr\p{R_2 X R_2^\dag} \rho_Q +\tr\p{R_3 X R_3^\dag}\rho_Q,\label{eq:RecoveryBO}
\end{equation}
for some arbitrary fixed pure state $\rho_Q=\ket{\psi_Q}\bra{\psi_Q}\in D(Q)$. Note that $\calR$ is clearly a CP map since each of the terms on the right in \eqref{eq:RecoveryBO} is a composition of a map of the form $X\to RXR^\dag$ and a partial trace, which are both CP maps. It remains to check that $\calR$ is trace preserving: 
\begin{align}
    \tr(\calR(X))&\nonumber= \tr\p{\tr_{{T'}}\p{R_1 X R_1^\dag}}+\p{\tr\p{R_2 X R_2^\dag}  +\tr\p{R_3 X R_3^\dag}}\tr(\rho_Q)\\
    &\nonumber=\tr\p{R_1 X R_1^\dag} +\tr\p{R_2 XR_2^\dag}+\tr\p{R_3 X R_3^\dag}\\
      &\nonumber=\tr\p{  R_1^\dag R_1 X} +\tr\p{R_2^\dag R_2 X}+\tr\p{ R_3^\dag R_3 X}\\
     &\label{eq:projJ}=\tr\p{(J_\calE^\dag J_{\lambda}J_{\lambda}^\dag J_{\calE} +J_\calE^\dag P_{J_\lambda^\perp} J_{\calE}+I_{\calH_{\SOUT}}-J_\calE^\dag  J_{\calE})X}
     \\&\label{eq:projJ_lambda}=\tr\p{(J_\calE^\dag P_{J_{\lambda}} J_{\calE} +J_\calE^\dag P_{J_\lambda^\perp} J_{\calE}+I_{\calH_{\SOUT}}-J_\calE^\dag  J_{\calE})X}\\
     &\label{eq:SumProj}=\tr\p{(J_\calE^\dag J_{\calE}+I_{\calH_{\SOUT}}-J_\calE^\dag  J_{\calE})X}=\tr(X).
\end{align}
In \eqref{eq:projJ} we used the fact that $R_3=I_{\calH_{\SOUT}}-J_{\calE}^\dag J_{\calE}$ is a projection operator (on the orthogonal complement of the support of $J_{\calE}$), which implies $R_3^\dag R_3=R_3$. Eq. \eqref{eq:projJ_lambda} 
follows since $J_{\lambda}$ is an isometry and therefore $J_{\lambda}J_{\lambda}^\dag=P_{J_{\lambda}}$, and in \eqref{eq:SumProj} we use $P_{J_{\lambda}}+P_{J_{\lambda}^{\perp}}=I_{M}$.\vspace{0.2cm}

\noindent\textbf{Step 3: Stinespring dilation of the action of an $\calE$-controlled channel on the code space.} We now consider the operation of an $\calE$-controlled channel on a code space, and express its canonical Stinespring dilation in terms of $V_{\calE}$. Let $\calN_C:L(\calH_{\SIN})\to L(\calH_{\SOUT})$ be an $\calE$-controlled channel with Kraus operators $\ppp{F_k}_{k\in [R]}$, $R\in \N$ such that $C=(c_{k,i})_{k,i}$ is the linear coefficient matrix defined by the equation
    \[
    F_k=\sum_{i\in [S]}c_{k,i}E_i.
    \]
We think of $C$ as an operator $\calH_{S}\to \calH_{R}$ given by
\begin{equation}
    C=\sum_{k\in [R]}\sum_{i\in [S]}c_{k,i}\ket{k}\bra{i}, 
    \label{eq:Coperator}
\end{equation}
where $\ppp{\ket{k}}_{k \in [R]} $ and $\ppp{\ket{i}}_{i \in [S]} $ are orthonormal bases for $\calH_{R}$ and $\calH_{S}$ respectively. The $\calE$-controlled assumption implies that $C$ can be chosen such that $\norm{C}_\infty\leq 1$. Consider the operator $V_{C}:\calH_{\SIN}\to \calH_{\SOUT}\otimes \calH_R  $ defined as $V_C=(I_{\calH_{\SOUT}}\otimes C )V_{\calE}$, 
where $V_\calE$ was defined in \eqref{eq:VcalEdef}.

\begin{lemma}\label{lem:DilationNoisyStates}
    The operator $V_C$ is a dilation operator of the channel $\calN_C$ on the code space $Q$ (which is the channel $\calN_C\circ \calN_Q$, where $\calN_Q:L(Q)\to L(\calH_{\SIN})$ is the inclusion encoding superoperator $\calN_Q(X)=PXP$). That is
    \begin{equation}
        \tr_{{R}}\p{V_{C}X V_{C}^\dag}=\calN_C\circ\calN_Q(X)\label{eq:Ncdilation}=\calN_C|_{L(Q)}(X).
            \end{equation}
\end{lemma}
\begin{proof}
    Note that $V_C$ acts as
\begin{align}
    V_C\ket{\psi}\nonumber&=(I_{\calH_{\SOUT}}\otimes C) V_{\calE}\ket{\psi}=I_{\calH_{\SOUT}}\otimes C \p{\sum_{i=0}^{S-1}K_i\ket{\psi} \otimes \ket{i}}\\
    &=\sum_{k=0}^{R-1} \sum_{i=0}^{S-1} c_{k,i}K_i\ket{\psi} \otimes \ket{k}=\sum_{k=0}^{R-1} F_k P\ket{\psi} \otimes \ket{k}.\label{eq:ChannelVCcanonical}
\end{align}
The above expression is the canonical Stinespring dilation for $\calN_C|_{L(Q)}=\calN_C\circ \calN_Q$ which has the Kraus representation $\ppp{F_kP}_k$ (and \eqref{eq:Ncdilation} can be verified by a straightforward calculation). 
\end{proof}\vspace{0.2cm}

\noindent\textbf{Step 4: Stinespring dilation of the decoded output for a noisy state}. Our next goal is to find a Stinespring dilation operator to obtain the decoded output $\calR\circ\calN_C|_{L(Q)}$ from its original representation \eqref{eq:RecoveryBO}. We begin by embedding images of the operators $R_1,R_2,R_3, V_{C}$ in a common reference system. Consider the space
\[\calH_{\s{Unif}}:=Q\otimes \calH_{R}\otimes \calH_{T'}\otimes \calH_{J_{\lambda}^\perp}\otimes \calH_{\SOUT}\otimes \calH_{\s{Flag}}, \quad \calH_{\s{Flag}}=\Span\ppp{\ket{1},\ket{2},\ket{3}},\]
and the environment 
\begin{equation}
    \calH_{\s{env}}:= \calH_{R}\otimes \calH_{T'}\otimes \calH_{J_{\lambda}^\perp}\otimes \calH_{\SOUT}\otimes \calH_{\s{Flag}}, \quad  \calH_{\s{Unif}}=Q\otimes \calH_{\s{env}}.\label{eq:Henv}
\end{equation}
If $\calH_{J_{\lambda}^{\perp}}$ is the zero space we replace it with $\C$. We define the embedding operators as follows: 
    \begin{equation}
\left.
\begin{aligned}
&U_1:Q\otimes \calH_{T'}\otimes \calH_R\to \calH_{\s{Unif}},
\qquad&&U_1= (I_{Q\otimes \calH_{T'}\otimes \calH_R})\otimes \ket{\psi_{J_{\lambda}^{\perp}}}
\otimes \ket{\psi_{\SOUT}}\otimes \ket{1}\\
&U_2:\calH_{J_{\lambda}^\perp}\otimes \calH_R\to \calH_{\s{Unif}},
&&U_2= (I_{\calH_{J_{\lambda}^\perp}\otimes \calH_R})\otimes \ket{\psi_Q}
\otimes \ket{\psi_{T'}}\otimes \ket{\psi_{\SOUT}}\otimes \ket{2}\\
&U_3:\calH_{\SOUT}\otimes \calH_R\to \calH_{\s{Unif}},
&&U_3= (I_{\calH_{\SOUT}\otimes \calH_R})\otimes \ket{\psi_Q}
\otimes \ket{\psi_{J_{\lambda}^{\perp}}}\otimes \ket{\psi_{T'}}\otimes \ket{3},
\end{aligned}\label{eq:UiDef}
\right\}
\end{equation}
    where $ \ket{\psi_{J_{\lambda}^{\perp}}}$, $  \ket{\psi_{T'}}$ and $\ket{\psi_{\SOUT}}$ are arbitrary unit vectors in $\calH_{J_{\lambda}^\perp}$, $\calH_{T'}$ and $\calH_{\SOUT}$, respectively, and $\ket{\psi_{Q}} $ is a vector such that $\ket{\psi_{Q}} \bra{\psi_{Q}} =\rho_Q$, which is the state from the definition of $\calR$ \eqref{eq:RecoveryBO}. 
    Note that $U_1,U_2$ and $U_3$ are isometries. Consider the operator $V_{\s{dec},C}:\calH_{\SIN} \to \calH_{\s{Unif}}$ defined as
    \[V_{\s{dec},C}:=U_1(I_{R}\otimes R_1)V_C+U_2(I_{R}\otimes R_2)V_C+ U_3(I_{R}\otimes R_3)V_C.\]
    \begin{lemma}\label{lem:dilationDecodedState}
        The operator $V_{\s{dec},C}$ is a Stinespring dilation for $\calR \circ\calN_C|_{L(Q)}=\calR \circ\calN_C\circ \calN_Q$, where $\calN_Q:L(Q)\to L(\calH_{\SIN})$ is the inclusion encoding map $X\to PXP$. That is, for any $X\in L(Q)$
        \[\tr_{{\s{env}}}\p{V_{\s{dec},C}X V_{\s{dec},C}^\dag}=\calR \circ\calN_C\circ\calN_Q(X).
        \]
    \end{lemma}

    \begin{proof}
        We begin by observing that $U_1$, $U_2$ and $U_3$ are supported on orthogonal spaces (since they are 
        associated with distinct orthogonal flag states $\ket{1}\!,\ket{2}\!,$ and $\ket{3}$). In particular, when tracing out the environment, cross terms vanish and we have:
        \begin{equation}
            \tr_{{\s{env}}}\p{V_{\s{dec},C}X V_{\s{dec},C}^\dag}=\sum_{i=1,2,3}\tr_{{\s{env}}}\p{V_{\s{dec},i}X V_{\s{dec},i}^\dag},\quad V_{\s{dec},i}:=U_i(I_{R}\otimes R_i)V_C.\label{eq:KeyForDilation0}
        \end{equation}
       We proceed by showing that for any $X$ we have
    \begin{equation}
        \tr_{\s{env}}\p{V_{\s{dec},i}XV_{\s{dec},i}^\dag}=\begin{cases}
        \tr_{{T'}}\p{R_1 \tr_{R}\p{V_C X V_C^\dag} R_1^\dag} & i=1,\\
        \tr\p{R_i \tr_{R}\p{V_C X V_C^\dag} R_i^\dag}\rho_Q & i=2,3.
    \end{cases} \label{eq:KeyForDilation}
    \end{equation}
    To this end, let 
    \[\calH_{\s{env}_1}:=\calH_{J_{\lambda}^\perp}\otimes \calH_{\SOUT}\otimes \calH_{\s{Flag}}\quad \calH_{\s{env}_2}:=\calH_{T'}\otimes \calH_{\SOUT}\otimes \calH_{\s{Flag}}, \quad \calH_{\s{env}_3}:=\calH_{J_{\lambda}^\perp}\otimes \calH_{T'}\otimes \calH_{\s{Flag}}.
    \]
By the definition of $U_1$ \eqref{eq:UiDef}
    \begin{align}
        \tr_{\s{env}}\p{V_{\s{dec},1}XV_{\s{dec},1}^\dag}\nonumber&=\tr_{\s{env}}\p{(I_{R}\otimes R_1)V_C X V_C^\dag(I_{R}\otimes R_1^\dag)\otimes \ket{\psi_{J_{\lambda}^{\perp}}} \bra{\psi_{J_{\lambda}^{\perp}}} \otimes \ket{\psi_{\SOUT}}\bra{\psi_{\SOUT}}\otimes \ket{1}\bra{1} }\\
        &\label{eq:Compotrace}=\tr_{RT'}\p{(I_{R}\otimes R_1)V_C X V_C^\dag(I_{R}\otimes R_1^\dag)},
        \\&=\tr_{{T'}}\p{R_1 \tr_{R}\p{V_C X V_C^\dag} R_1^\dag},\label{eq:compoLemmaConj}
    \end{align}
    where in \eqref{eq:Compotrace} we used the fact that the partial trace has the property that, for any operator $X$ acting on a composite system $A\otimes B\otimes C$, we have $\tr_{BC}(X)=\tr_B(\tr_C(X))$, and in \eqref{eq:compoLemmaConj} we used Lemma~\ref{lem:tracPropComposite} with $\calH_A=Q$, $\calH_B=\calH_{T'}$ and $\calH_{C}=\calH_R$.
    For $i=2$, using the same arguments, we obtain:  
    \begin{align}
        \tr_{\s{env}}\p{V_{\s{dec},2}XV_{\s{dec},2}^\dag}\nonumber&\nonumber=\tr_{\s{env}}\p{\ket{\psi_{Q}} \bra{\psi_{Q}} \otimes (I_{R}\otimes R_2)V_C X V_C^\dag(I_{R}\otimes R_2^\dag)\otimes  \ket{\psi_{T'}} \bra{\psi_{T'}} \otimes \ket{\psi_{\SOUT}}\bra{\psi_{\SOUT}}\otimes \ket{2}\bra{2} }\\
        &\label{eq:rho_Q}= \tr\p{(I_{R}\otimes R_2)V_C X V_C^\dag(I_{R}\otimes R_2^\dag)}\rho_Q\\
        &=\tr\p{R_2 \tr_{R}\p{V_C X V_C^\dag} R_2^\dag}\rho_Q,\label{eq:theSameLemmaaGAIN}
    \end{align}
    where \eqref{eq:rho_Q} uses that $\rho_{Q}=\ket{\psi_Q}\bra{\psi_Q}$ by definition, and \eqref{eq:theSameLemmaaGAIN} uses Lemma~\ref{lem:tracPropComposite} again where we identify $\calH_{A}=\C$, $\calH_B=\calH_{J_{\lambda}^{\perp}}$ and $\calH_C=\calH_R$. The calculation for $i=3$ follows the same sequence of steps as in the case $i=2$, and we leave the details to the reader. 
    
    We now conclude the proof by combining \eqref{eq:KeyForDilation0}, \eqref{eq:KeyForDilation} with   the noisy channel output expression of Lemma~\ref{lem:DilationNoisyStates} and the definition of $\calR$ given in \eqref{eq:RecoveryBO}:
    \begin{align*}
        \tr_{{\s{env}}}\p{V_{\s{dec},C}X V_{\s{dec},C}^\dag}&=\sum_{i=1,2,3}\tr_{{\s{env}}}\p{V_{\s{dec},i}X V_{\s{dec},i}^\dag}\\&
        =\tr_{{T'}}\p{R_1 \tr_{R}\p{V_C X V_C^\dag} R_1^\dag}  +\sum_{i=2,3}
        \tr\p{R_i \tr_{R}\p{V_C X V_C^\dag} R_i^\dag}\rho_Q  \\
        &=\tr_{{T'}}\p{R_1  \p{\calN_{C}\circ \calN_{Q}(X)} R_1^\dag}  +\sum_{i=2,3}
        \tr\p{R_i  \p{\calN_{C}\circ \calN_{Q}(X)} R_i^\dag}\rho_Q \\
       &= \calR\circ \calN_{C}\circ \calN_{Q}(X).
    \end{align*}
      \end{proof}
    \noindent\textbf{Step 5: Closeness of the decoded output channel to the identity channel.} We now make the final step of the proof and show that the worst-case distance between $\calR \circ \calN_C\circ \calN_Q$ and the identity channel $I_{L(Q)}$ is bounded by $\sqrt{ \zeta(\calE,Q)}$. To this end, by Eq.~\eqref{eq:BuresDistanceStinespring} and Lemma~\ref{lem:dilationDecodedState}, it is sufficient to find a dilation operator $V_{\lambda_C}:Q\to \calH_{\s{Unif}}$ that dilates the identity channel $I_{L(Q)}$, such that $\norm{V_{\s{dec},C}-V_{\lambda_C}}_{\infty}\leq \sqrt{2\zeta(\calE,Q)}$. The Stinespring dilation operator of the constant channel $V_{\lambda}$, defined in \eqref{eq:VlambdaDef}, and the channel coefficient operator $C$, defined in \eqref{eq:Coperator}, yield a natural candidate: $(I_{\calH_{T'}\otimes Q}\otimes C) V_{\lambda}$. Since the image of this map lies in $Q\otimes \calH_R\otimes \calH_{T'}$, for it to be comparable with $V_{\s{dec},C}$, we have to embed it in the combined space using $U_1$ \eqref{eq:UiDef}. Formally, define
    \[V_{\lambda_{C}}:=U_1(I_{\calH_{T'}\otimes Q}\otimes C) V_{\lambda}.\]
    \begin{lemma} The operator $V_{\lambda_C}$ dilates the identity. That is, for all $X\in L(Q)$ we have  
    \[\tr_{\s{env}}\p{V_{\lambda_C}XV_{\lambda_C}^\dag}=X,\]
    where $\calH_{\s{env}}$ is the environment system defined in \eqref{eq:Henv}. 
    \end{lemma}
    \begin{proof}
        Using the same argument as in the proof of Lemma~\ref{lem:dilationDecodedState} we have 
       \begin{align}
           \tr_{\s{env}}\p{V_{\lambda_C}XV_{\lambda_C}^\dag}&\nonumber=\tr_{\s{env}}\p{U_1(I_{\calH_{T'}\otimes Q}\otimes C) V_{\lambda}XV_{\lambda}^\dag (I_{\calH_{T'}\otimes Q}\otimes C^\dag)U_1^\dag} \\&=\tr_{T'R}\p{(I_{\calH_{T'}\otimes Q}\otimes C) V_{\lambda}XV_{\lambda}^\dag (I_{\calH_{T'}\otimes Q}\otimes C^\dag)}.\label{eq:envAgain}
       \end{align}
        Recall that $V_{\lambda}= I_Q\otimes \ket{\psi_{\lambda}}$ where $\ket{\psi_{\lambda}}=\sum_{i\in[T]}\sqrt{p_i}\ket{\lambda_{i}}\otimes \ket{\lambda_{i}'}\in \calH_T\otimes \calH_{T'}$ and therefore,
        \begin{align*}
            (I_{\calH_{T'}\otimes Q}\otimes C) V_{\lambda}=I_{Q}\otimes \p{ (C\otimes I_{\calH_{T'}})\ket{\psi_{\lambda}}}:=I_Q\otimes \ket{\psi_{\lambda_{C}}},
        \end{align*}
        where 
        \begin{equation}
            \ket{\psi_{\lambda_{C}}}=\sum_{i\in[T]}\sqrt{p_i} C\ket{\lambda_i}\otimes \ket{\lambda_{i}'}\in \calH_{R}\otimes \calH_{T'}.\label{eq:psiLanbdaCdef}
        \end{equation}
    Combining this with \eqref{eq:envAgain}, we have 
    \[\tr_{\s{env}}\p{V_{\lambda_C}XV_{\lambda_C}^\dag}=\tr_{T'R}\p{X\otimes \ket{\psi_{\lambda_C}}\bra{\psi_{\lambda_C}}}=\braket{\psi_{\lambda_C} | \psi_{\lambda_C}} X.\]
Thus, to prove the lemma it suffices to show that $\braket{\psi_{\lambda_C} | \psi_{\lambda_C}}=1$. 
    
    Recalling \eqref{eq:lambdaDef} and \eqref{eq:VlambdaDef}, we have
    \[\lambda=\sum_{i\in[T]}p_i\ket{\lambda_i}\bra{\lambda_i}=\sum_{i,j\in [S]}\tr\p{\hat{P}E_j^\dag E_i} \ket{i}\bra{j},\]
    where $\hat{P}=\frac{1}{K}P\in D(Q)$. By \eqref{eq:psiLanbdaCdef} we have 
    \begin{align*}
        \braket{\psi_{\lambda_C} | \psi_{\lambda_C}} &=\sum_{i,j\in [T]}\sqrt{p_ip_j}\bra{\lambda_i}C^\dag C \ket{ \lambda_j }   \braket{\lambda_i'| \lambda_j'}=\sum_{i\in [T]}p_i\bra{\lambda_i}C^\dag C \ket{ \lambda_i } \\
        &=\sum_{i\in[T]}p_i\tr(C^\dag C \ket{\lambda_i}\bra{\lambda_i})=\tr\p{C^\dag C\sum_{i\in [T]}p_i\ket{\lambda_i}\bra{\lambda_i}}\\
        &=\tr(C\lambda C^\dag)=\tr\p{C \bigg(\sum_{i,j\in [S]}\tr\p{\hat{P} {E_j^\dag E_i}}\ket{i}\bra{j}\bigg) C^\dag}\\
        &=\tr\p{\sum_{k,l\in [R]}\sum_{i,j\in [S]}  {c_{k,i}c_{l,j}^*} \bigg(\tr\p{\hat{P} {E_j^\dag E_i}}\bigg)\ket{k}\bra{l} }\\
        &=\sum_{k\in [R]}\sum_{i,j\in [S]}  {c_{k,i}c_{k,j}^*} \bigg(\tr\p{\hat{P} {E_j^\dag E_i}}\bigg)\\
        &=\sum_{k\in [R]}\sum_{i,j\in [S]}  {c_{k,i}c_{k,j}^*} \bigg(\tr\p{ {E_i}\hat{P} {E_j^\dag} }\bigg)\\
        &=\sum_{k\in [R]} \tr\p{\bigg( {\sum_{i\in [S]}c_{k,i}E_i}\bigg)\hat{P}\bigg( {\sum_{j\in [S]}c_{k,j}^*E_j^\dag}\bigg)}\\&=\tr\p{\sum_{k\in [R]}F_k \hat{P}F_k^\dag}=\tr\p{\calN_C(\hat{P})}=1,
    \end{align*}
    where the last equality follows since $\tr(\hat{P})=1$ and $\calN_C$ is a quantum channel and therefore is trace preserving.
    
    \end{proof}

    We now make the final step, and prove closeness of $V_{\s{dec},C}$ and $V_{\lambda_C}$. 
    \begin{lemma}
        \[\norm{V_{\s{dec},C}-V_{\lambda_C}}_{\infty}\leq \norm{(I_{\s{\calH_{S}}}\otimes J_{\calE})V_{\calE}-(I_{\calH_S}\otimes J_{\lambda})V_{\lambda}}_{\infty}.\]
        where the $\norm{\cdot}_{\infty}$ is with respect to the operator norm on the code space $Q$.
    \end{lemma}
    \begin{proof}
         Throughout this proof we will repeatedly use the following simple fact: For Hilbert spaces $\calH_A,\calH_B,\calH_C,\calH_D$ and operators $T_{A\to C}:\calH_{A}\to \calH_{C}$ and $T_{B\to D}:\calH_{B}\to \calH_{D}$ we have 
         \[(T_{B\to D}\otimes I_{\calH_C} )(I_{\calH_B}\otimes T_{A\to C})=(I_{\calH_D}\otimes T_{A\to C})(T_{B\to D}\otimes I_{\calH_A} ).\]
        Let $\ket{\psi}\in Q$ be a code state. We begin by observing that $\ket{\psi}$ vanishes on the branch $U_3(I_{\calH_R}\otimes R_3)V_{C}$. Indeed:
        \begin{align*}
            (I_{\calH_R}\otimes R_3)V_C\ket{\psi}&=\big(I_{\calH_R}\otimes (I_{\calH_{\SOUT}}-J_{\calE}^\dag J_{\calE})\big)(C \otimes I_{\calH_{\SOUT}})V_{\calE}\ket{\psi}\\
            &=(C \otimes I_{\calH_{\SOUT}})\big(I_{\calH_S}\otimes (I_{\calH_{\SOUT}}-J_{\calE}^\dag J_{\calE})\big)V_{\calE}\ket{\psi}
            \\
            &=(C \otimes I_{\calH_{\SOUT}})\big((I_{\calH_S}\otimes I_{\calH_{\SOUT}})V_{\calE}\ket{\psi} -(I_{\calH_S}\otimes J_{\calE}^\dag J_{\calE})V_{\calE}\ket{\psi}\big)=0,
        \end{align*}
        where the last equality follows from the first relation in \eqref{eq:commute}. This derivation implies that
        \begin{equation}
            V_{\s{dec},C}\ket{\psi}=U_1(I_{\calH_R}\otimes R_1)V_{C}\ket{\psi}+U_2(I_{\calH_R}\otimes R_2)V_C\ket{\psi}.\label{eq:eqmanek}
        \end{equation} Let $\norm{\cdot}$ denote the usual inner product norm on $\calH_{\s{Unif}}$. By \eqref{eq:eqmanek} we have:
        \begin{align}
            \norm{V_{\s{dec},C}\ket{\psi}-V_{\lambda_C}\ket{\psi}}^2&\nonumber=\norm{U_1(I_{\calH_R}\otimes R_1)V_{C}\ket{\psi}+U_2(I_{\calH_R}\otimes R_2)V_C\ket{\psi}-U_1(I_{\calH_{T'}\otimes Q}\otimes C) V_{\lambda}\ket{\psi}}^2\\
            &=\norm{U_1(I_{\calH_R}\otimes R_1)V_{C}\ket{\psi}-U_1(I_{\calH_{T'}\otimes Q}\otimes C) V_{\lambda}\ket{\psi}}^2+\norm{U_2(I_{\calH_R}\otimes R_2)V_C\ket{\psi}}^2\label{eq:onOthogonalStates}\\
            &=\norm{(I_{\calH_R}\otimes R_1)V_{C}\ket{\psi}-(I_{\calH_{T'}\otimes Q}\otimes C) V_{\lambda}\ket{\psi}}^2+\norm{(I_{\calH_R}\otimes R_2)V_C\ket{\psi}}^2\label{eq:U1U2isometry}.
        \end{align}
        where \eqref{eq:onOthogonalStates} follows since the images of $U_1$ and $U_2$ are orthogonal (since they
        include orthogonal flag states) and \eqref{eq:U1U2isometry} follows since $U_1$ and $U_2$ are isometries. The first term on the right in \eqref{eq:U1U2isometry} can be written as
        \begin{align}
            &\nonumber\norm{(I_{\calH_R}\otimes R_1)V_{C}\ket{\psi}-(I_{\calH_{T'}\otimes Q}\otimes C) V_{\lambda}\ket{\psi}}\\
            &\nonumber\quad =\norm{(I_{\calH_R}\otimes J_{\lambda}^\dag J_{\calE})(I_{\calH_{\SOUT}}\otimes C)V_{\calE}\ket{\psi}-(I_{\calH_{T'}\otimes Q}\otimes C) V_{\lambda}\ket{\psi}}\nonumber\\&
            \quad =\norm{(I_{\calH_R}\otimes J_{\lambda}^\dag )(I_{\calH_R}\otimes  J_{\calE})(I_{\calH_{\SOUT}}\otimes C)V_{\calE}\ket{\psi}-(I_{\calH_{T'}\otimes Q}\otimes C) (I_{\calH_{S}}\otimes J_{\lambda}^\dag J_{\lambda})V_{\lambda}\ket{\psi}}\label{eq:ProjecJ}\\&
            \quad =\norm{(I_{\calH_R}\otimes J_{\lambda}^\dag )(I_{\calH_{M}}\otimes C)(I_{\calH_S}\otimes  J_{\calE})V_{\calE}\ket{\psi}  -(I_{\calH_{S}}\otimes J_{\lambda}^\dag )(I_{\calH_{M}}\otimes C) (I_{\calH_{S}}\otimes J_{\lambda})V_{\lambda}\ket{\psi}}\nonumber\\
            &\quad =\norm{(I_{\calH_R}\otimes J_{\lambda} J_{\lambda}^\dag )(I_{\calH_{M}}\otimes C)(I_{\calH_S}\otimes  J_{\calE})V_{\calE}\ket{\psi}  -(I_{\calH_{S}}\otimes J_{\lambda} J_{\lambda}^\dag )(I_{\calH_{M}}\otimes C) (I_{\calH_{S}}\otimes J_{\lambda})V_{\lambda}\ket{\psi}}\label{eq:Jisometry1}\\
            &\quad =\norm{(I_{\calH_R}\otimes P_{J_{\lambda}})\Big((I_{\calH_{M}}\otimes C)(I_{\calH_S}\otimes  J_{\calE})V_{\calE}\ket{\psi}  -(I_{\calH_{M}}\otimes C) (I_{\calH_{S}}\otimes J_{\lambda})V_{\lambda}\ket{\psi}\Big)}\label{eq:Jisometry2}            
        \end{align}
        where \eqref{eq:ProjecJ} follows from \eqref{eq:commute}, and \eqref{eq:Jisometry1} follows since $J_{\lambda}$ is an isometry and therefore, so is $I_{\calH_R}\otimes J_{\lambda}$. Before analyzing the second term in \eqref{eq:U1U2isometry} we note that 
        \begin{align}
            (I_{\calH_R}\otimes P_{J_{\lambda}^{\perp}})(I_{\calH_{M}}\otimes C) (I_{\calH_{S}}\otimes J_{\lambda})V_{\lambda}\ket{\psi}&\nonumber=(I_{\calH_M}\otimes C)(I_{\calH_{S}}\otimes P_{J_{\lambda}^{\perp}}) (I_{\calH_{S}}\otimes J_{\lambda})V_{\lambda}\ket{\psi}\\
            &=(I_{\calH_M}\otimes C) (I_{\calH_{S}}\otimes P_{J_{\lambda}^{\perp}}J_{\lambda})V_{\lambda}\ket{\psi}=0,\label{eq:PjZerosJ}
        \end{align}
        where the last equality follows since $P_{J_{\lambda}^\perp}$ is the projector on the orthogonal complement of the image of $J_{\lambda}$, which implies that $P_{J_{\lambda}^\perp}J_{\lambda}=0$. Using $\eqref{eq:PjZerosJ}$, we obtain
        \begin{align}
            \norm{(I_{\calH_R}\otimes R_2)V_C\ket{\psi}}&\nonumber=\norm{(I_{\calH_R}\otimes P_{J_{\lambda}^\perp}J_{\calE})(I_{\calH_{\SOUT}}\otimes C)V_\calE\ket{\psi}}\\
            &\nonumber=\norm{(I_{\calH_R}\otimes P_{J_{\lambda}^\perp})(I_{\calH_{M}}\otimes C)(I_{\calH_S}\otimes J_{\calE})V_\calE\ket{\psi}}\\
            &=\norm{(I_{\calH_R}\otimes P_{J_{\lambda}^\perp})\Big( (I_{\calH_{M}}\otimes C)(I_{\calH_S}\otimes J_{\calE})V_\calE\ket{\psi}- (I_{\calH_{M}}\otimes C) (I_{\calH_{S}}\otimes J_{\lambda})V_{\lambda}\ket{\psi}\Big)}.\label{eq:R2part} 
        \end{align} Denote 
        \[\ket{\hat{\psi}}:=(I_{\calH_{M}}\otimes C)(I_{\calH_S}\otimes J_{\calE})V_\calE\ket{\psi}- (I_{\calH_{M}}\otimes C) (I_{\calH_{S}}\otimes J_{\lambda})V_{\lambda}\ket{\psi},\]
        and note that the vectors $(I_{\calH_R}\otimes P_{J_{\lambda}^{\perp}})\ket{\hat{\psi}}$ and $(I_{\calH_R}\otimes P_{J_{\lambda}})\ket{\hat{\psi}}$ are orthogonal  as $P_{J_{\lambda}^\perp}$ and $P_{J_{\lambda}}$ project on  complementary orthogonal spaces. Therefore,
        \begin{equation}
             \norm{(I_{\calH_R}\otimes P_{J_{\lambda}^{\perp}})\ket{\hat{\psi}}}^2+\norm{(I_{\calH_R}\otimes P_{J_{\lambda}})\ket{\hat{\psi}}}^2=\norm{(I_{\calH_R}\otimes P_{J_{\lambda}^{\perp}})\ket{\hat{\psi}}+(I_{\calH_R}\otimes P_{J_{\lambda}})\ket{\hat{\psi}}}^2=\norm{\ket{\hat{\psi}}}^2.\label{eq:hatpsi}
        \end{equation}
        Combining \eqref{eq:U1U2isometry}, \eqref{eq:Jisometry2}, \eqref{eq:R2part}, and \eqref{eq:hatpsi}, we conclude that
        \begin{align*}
            \norm{V_{\s{dec},C}\ket{\psi}-V_{\lambda_C}\ket{\psi}}^2&=\norm{\ket{\hat{\psi}}}^2=\norm{(I_{\calH_{M}}\otimes C)\Big((I_{\calH_S}\otimes J_{\calE})V_\calE\ket{\psi}- (I_{\calH_{S}}\otimes J_{\lambda})V_{\lambda}\ket{\psi}\Big)}^2\\
            &\leq \norm{(I_{\calH_S}\otimes J_{\calE})V_\calE\ket{\psi}- (I_{\calH_{S}}\otimes J_{\lambda})V_{\lambda}\ket{\psi}}^2,
        \end{align*}
        where the last inequality follows since $I_{\calH_M}\otimes C$ is a contraction since $\norm{C}_{\infty}\leq 1$. 
Since this relation holds for all $\ket{\psi}\in Q$, this implies the desired norm inequality. 
    \end{proof}
    \subsection{Proof of Lemma~\ref{lem:CompactEcontrolled}}\label{app:LemmaCompactProof}
    We recall that the set of quantum channels between finite-dimensional spaces is convex and compact (see \cite[Proposition 2.28]{watrous2018theory}), which implies convexity and compactness of $X$. The set $Y$ of $\calE$-controlled channels is also convex; see \cite[Observation~2]{Elimelech2026Theory}. It remains to show that  $Y$ is also compact. Write $R=|\calE|$. 
    
    Let $\s{M}(\calE)$ be the space of all $R\times R$ matrices $A\succeq 0$ with complex entries such that $\norm{A}_{\infty}\leq 1$, equipped with the $\infty$-norm topology. First note that $\s{M}(\calE)$ is a compact space: It is bounded (by the constraint $\norm{A}_{\infty}\leq 1$), and it is closed as the intersection of closed sets (the set of positive semidefinite matrices is a closed set, and the set of matrices satisfying $\norm{A}_{\infty}\leq 1$ is closed by continuity of the norm). 
    
    Consider the function $g$ from $\s{M}(\calE)$ to the set of linear superoperators $L(\calH_{\SIN})\to L(\calH_{\SOUT})$ defined by 
    \[ (g(A))(X)=\sum_{i,j\in[R]} A_{i,j}E_i X E_j^\dag.\]
    Note that $g$ is a linear function, and therefore it is also continuous (as any linear function between finite-dimensional normed spaces is continuous). 
    
    We claim that $Y$ is exactly the intersection $\ima(g)$ (which is compact, as the image of a compact set under a continuous map), with the set of CPTP maps $L(\calH_{\SIN})\to L(\calH_{\SOUT})$, denoted by $\s{CPTP}(\calH_{\SIN},\calH_{\SOUT})$  (which is also compact by \cite[Proposition 2.28]{watrous2018theory}). This  will imply that $Y$ is compact as an intersection of compact sets. Indeed, let $A\in \s{M}(\calE)$ be a positive matrix with $\norm{A}_{\infty}\leq 1$. Consider the matrix $C=\sqrt{A}$. Note that $g(A)$ is the CP map with Kraus representation $\ppp{F_\alpha}_{\alpha\in [R]}$ such that 
    \[F_{\alpha}=\sum_{i\in [R]}C_{i,\alpha}E_i, \]
    as we have
    \begin{equation}
        \sum_{\alpha\in [R]}F_{\alpha} X F_{\alpha}^\dag= \sum_{\alpha\in [R]}\sum_{i,j\in [R]} C_{i,\alpha} C_{j,\alpha}^* E_i X E_{j}^\dag= \sum_{i,j\in [R]} E_i X E_{j}^\dag\sum_{\alpha\in [R]}  C_{i,\alpha} C_{\alpha,j}=\sum_{i,j\in [R]} A_{i,j }E_i X E_{j}^\dag=g(A)(X)\label{eq:reveresethecalc}
    \end{equation}
    Here, we used the fact that  $C$ is Hermitian as a square root and therefore $C_{\alpha,j}^*=C_{j,\alpha}$. We also note that $\norm{C}_{\infty}=\sqrt{\norm{A}_{\infty}}\leq 1$. On the other hand assume that $\calN$ is a CP map with Kraus representation $\ppp{F_{\alpha}}_{\alpha\in [S]}$, for some $S\in \N$, such that 
     \begin{equation}
         F_{\alpha}=\sum_{i\in [R]}C_{i,\alpha}E_i, \quad \text{and}\quad \norm{C}_{\infty}\leq 1.\label{eq:KrausRepC}
     \end{equation}
     We want to show that $\calN\in \ima(g)$. To that end, take $A=CC^\dag \succeq 0$, we have $\norm{A}_\infty=\norm{C}_\infty^2\leq 1$ and by reversing the calculation of \eqref{eq:reveresethecalc} we obtain $g(A)=\calN$. We have proved that $\ima(g)$ is the set of all CP maps $L(\calH_{\SIN})\to L(\calH_{\SOUT})$ that admit a Kraus representation of the form \eqref{eq:KrausRepC}, whose intersection with $\s{CPTP}(\calH_{\SIN},\calH_{\SOUT})$ gives exactly $Y$.


\section{Amplitude damping noise for the Binomial code}\label{app:binomialAD}
We give the control calculation and verify the exact photon-loss capability stated in Section~\ref{sec:binomialExperiment}. The input space is $\calH_D^{\mathrm{even}}$, with $D$ even, and the output space is $\Span\{\ket0,\ldots,\ket D\}$. The parity restriction is part of the definition in~\eqref{eq:NEerrorSet}; every simulated code lies in this input space.

\subsection{Control of the truncated amplitude damping map}
Fix an even integer $D$, an integer $0\leq t<D$, and a loss probability $0<\gamma<1$. We consider the truncated amplitude damping map
\[
 \calN_{\gamma,\leq t}(X)=\sum_{\ell=0}^{t}A_\ell(\gamma)X A_\ell(\gamma)^\dag,
 \qquad X\in L(\calH_D^{\mathrm{even}}),
\]
and seek an error set independent of $\gamma$, which controls $\calN_{\gamma,\leq t}$ in the sense of Definition~\ref{def:errorSetAQEC}. In the constant-excitation setting of~\cite[Example~3]{Elimelech2026Theory}, the attenuation is a scalar for each total loss count. Here it depends on the occupied photon number $m$. To write $A_\ell(\gamma)$ as a linear combination of fixed error operators, \textbf{we seek coefficients that depend on $\gamma$ but not on $m$}, leaving the dependence on $m$ in the action of the error operators. This is accomplished as follows. Recall that the maximal photon number is assumed not to exceed $D$. For each even $m\leq D$, multiplication by $\p{\gamma+(1-\gamma)}^{(D-m)/2}=1$ makes the exponents of $\gamma$ and $1-\gamma$ in the definition of $A_{\ell}(\gamma)$ sum to $D/2$ in every term of the expansion. Thus, the dependence on $\gamma$ is expressed through the same scalar factors $\gamma^{k/2}(1-\gamma)^{(D-k)/2}$ for every $m$, while the remaining dependence on $m$ is absorbed into error operators independent of $\gamma$.

Specifically, let $m, \ell\le m\leq D$ be even and put $h=(D-m)/2$. As suggested above, we
multiply the expression for $A_\ell(\gamma)\ket m$ in~\eqref{eq:NEAD} by $\p{\gamma+(1-\gamma)}^h=1$ to write it using the factors that are independent of $m$. Expanding the binomial gives
\begin{equation}\label{eq:NEADexpansion}
\begin{aligned}
 A_\ell(\gamma)\ket m
 &= (\gamma+(1-\gamma))^{h}\sqrt{\binom m\ell}\,\gamma^{\ell/2}(1-\gamma)^{(m-\ell)/2}\ket{m-\ell}\\
 &=\sqrt{\binom m\ell}\gamma^{\ell/2}(1-\gamma)^{(m-\ell)/2}
   \sum_{j=0}^{h}\binom hj\gamma^j(1-\gamma)^{h-j}\ket{m-\ell}\\
 &=\sum_{j=0}^{h}\gamma^{(\ell+2j)/2}(1-\gamma)^{(D-\ell-2j)/2}
   \sqrt{\binom m\ell}\binom hj\ket{m-\ell}\\
 &=\sum_{\substack{\ell\leq k\leq D\\k-\ell\text{ even}}}
   \gamma^{k/2}(1-\gamma)^{(D-k)/2}
   \sqrt{\binom m\ell}\binom{(D-m)/2}{(k-\ell)/2}\ket{m-\ell}.
\end{aligned}
\end{equation}
In the last equality we set $k=\ell+2j$ and use the convention that binomial coefficients with half-integer entries are zero. For $m<\ell$, the action of $A_\ell(\gamma)$ is zero.

The last expression separates the dependence on $\gamma$ from the action on $\ket m$, which is the key feature of the construction of $\calE_{t,D}$. For each $(\ell,k)$, the factor
\[
 \sqrt{\binom m\ell}\binom{(D-m)/2}{(k-\ell)/2}\ket{m-\ell}
\]
defines an operator independent of $\gamma$. Dividing this operator by $\sqrt{\binom Dk}$ gives exactly $E_{\ell,k}$ in~\eqref{eq:NEerrorSet}, and multiplies its coefficient by $\sqrt{\binom Dk}$. We therefore obtain
\begin{equation}\label{eq:NEcontrolIdentity}
 A_\ell(\gamma)|_{\calH_D^{\mathrm{even}}}
 =\sum_{\substack{\ell\leq k\leq D\\k-\ell\text{ even}}}c_k(\gamma)E_{\ell,k},
 \qquad c_k(\gamma)=\sqrt{\binom Dk\gamma^k(1-\gamma)^{D-k}}.
\end{equation}
This normalization makes $c_k(\gamma)^2$ the probability of $k$ successes in the binomial distribution $\mathrm{Bin}(D,\gamma)$. In particular,
\[
 \sum_{k=0}^{D}|c_k(\gamma)|^2=1,
\]
and the squared coefficients appearing in any one Kraus operator sum to at most one. We conclude the following.

\begin{lemma}[Control of truncated amplitude damping]\label{lem:NEADcontrol}
Let $D$ be a positive even integer and $0\leq t<D$ an integer. On the input space $\calH_D^{\mathrm{even}}$, the truncated AD map $\calN_{\gamma,\leq t}$ admits the Kraus expansion~\eqref{eq:NEcontrolIdentity} whose coefficient matrix $C_\gamma$ satisfies $\norm{C_\gamma}_\infty\leq1$ for all $0\le\gamma\le 1$. Thus, it satisfies the contraction condition of Definition~\ref{def:errorSetAQEC} for the same error set $\calE_{D,t}$ throughout this interval.
\end{lemma}
\begin{proof}
Expansion~\eqref{eq:NEcontrolIdentity} gives
\[
 (C_\gamma)_{\ell,(\ell',k)}=
 \begin{cases}c_k(\gamma),&\ell'=\ell,\\0,&\ell'\ne\ell.\end{cases}
\]
The rows have disjoint supports, so $C_\gamma C_\gamma^\dag$ is diagonal. Its largest diagonal entry is the largest squared Euclidean norm of a row, and therefore
\[
 \norm{C_\gamma}_\infty^2
 =\max_{0\leq\ell\leq t}\sum_{\substack{\ell\leq k\leq D\\k-\ell\text{ even}}}|c_k(\gamma)|^2
 \leq\sum_{k=0}^{D}|c_k(\gamma)|^2=1.\qedhere
\]
\end{proof}

\bibliographystyle{abbrvurl}
\bibliography{references-doi}

\end{document}